\documentclass[11pt]{article}
\usepackage[margin=1in]{geometry}
\usepackage[colorlinks = true,
            linkcolor = blue,
            urlcolor  = blue,
            citecolor = blue,
            anchorcolor = blue]{hyperref}
\usepackage[utf8]{inputenc}
\usepackage{authblk}
\usepackage{setspace,etoolbox}
\usepackage{amsfonts, amssymb, amsthm, amsmath, mathrsfs, mathtools}
\usepackage{natbib, enumitem, setspace, multicol}
\usepackage{graphicx, float, caption, subcaption}
\usepackage{bm}
\usepackage{verbatim}
\usepackage{pdfpages}
\usepackage{fancyvrb}
\usepackage{hyperref}
\usepackage[linesnumbered,ruled]{algorithm2e}
\usepackage{algorithmic}
\usepackage{listings}
\usepackage{bm}

\usepackage{makecell}

\usepackage{natbib}
\newtheorem{corollary}{Corollary}[section]

\newtheorem{lemma}{Lemma}[section]
\theoremstyle{definition}
\newtheorem{example}{Example}[section]
\newtheorem{theorem}{Theorem}[section]

\newtheorem{ass}{Assumption}[section]

\newtheorem{defn}{Definition}[section]
\newtheorem{prop}{Proposition}[section]

\newtheorem*{result*}{Result}
\newtheorem*{remark*}{Remark}
\newtheorem*{claim*}{Claim}

\usepackage{bm}

\newcommand{\dto}{\stackrel{d}{\longrightarrow}}
\newcommand{\Cipw}{\mathcal{C}^{\textrm{W-proj}}}
\newcommand{\Cnull}{\mathcal{C}^{\mathrm{bal}}}
\newcommand{\Cipwlamb}{\mathcal{C}^{\textrm{W-proj}}_{\bm \Lambda}}
\newcommand{\Cipwprojlamb}{\mathcal{C}^{\textrm{proj}}_{\bm \Lambda}}
\newcommand{\Cnulllamb}{\mathcal{C}^{\mathrm{bal}}_{\bm \Lambda}}
\newcommand{\wipwlamb}{\hat{\bm w}^{\textrm{W-proj}}_{\bm \Lambda}}

\newcommand{\wipwlambnobm}{\hat{w}^{\textrm{W-proj}}_{\bm \Lambda}}
\newcommand{\wipwprojlamb}{\hat{\bm w}^{\textrm{proj}}_{\bm \Lambda}}

\newcommand{\wnulllamb}{\hat{\bm w}^{\mathrm{bal}}_{\bm \Lambda}}

\newcommand{\wipwlambtilde}{\tilde{\bm w}^{\mathrm{IPW}}_{\bm \Lambda}}

\newcommand{\phiipw}{\phi^{\textrm{W-proj}}}
\newcommand{\phinull}{\phi^{\mathrm{bal}}}
\newcommand{\sone}{\textrm{W-proj}}
\newcommand{\stwo}{\textrm{proj}}
\newcommand{\sthree}{\textrm{bal}}
\newcommand{\Pto}{\stackrel{\mathbb P}{\longrightarrow}}

\newcommand{\Var}{\mathrm{Var}}
\newcommand{\ind}{\mbox{$\perp\!\!\!\perp$}}

\newcommand{\iid}{\stackrel{i.i.d.}{\sim}}
\newcommand{\var}{\mathrm{Var}}

\allowdisplaybreaks
\numberwithin{equation}{section}

\title{\textbf{Low-rank Covariate Balancing Estimators\\ under Interference}}
\author{Souhardya Sengupta\thanks{Department of Statistics, Harvard University. Email: \href{mailto:ssengupta@g.harvard.edu}{ssengupta@g.harvard.edu} } \qquad Kosuke Imai\thanks{Department of Government and Department of Statistics, Harvard University.  1737 Cambridge Street, Institute for Quantitative Social Science, Cambridge MA 02138, U.S.A.  Email: \href{mailto:imai@harvard.edu}{imai@harvard.edu} URL:
\href{https://imai.fas.harvard.edu}{https://imai.fas.harvard.edu}} \qquad Georgia Papadogeorgou\thanks{Department of Statistics, University of Florida. 230 Newell Drive, Gainesville FL 32611, U.S.A. Email: \href{mailto:gpapadogeorgou@ufl.edu}{gpapadogeorgou@ufl.edu} URL: \href{https://gpapadogeorgou.netlify.app}{https://gpapadogeorgou.netlify.app}}}
\date{}

\begin{document}
\maketitle
\begin{abstract}
   In observational studies, causal effect estimation becomes difficult when the outcome of one unit depends on the treatments of many other units. Existing approaches rely on inverse-probability weighting (IPW) with estimated propensity scores, which can lead to misspecification bias especially when treatment uptake exhibits complex dependencies across units. To address this challenge, we propose a general class of \textit{covariate balancing estimators} under interference. Existing balancing estimators are limited to settings without spillover effects. We develop our approach under the assumed low-rank structure of potential outcomes, which encompasses a broad class of commonly used assumptions, including anonymous, nearest neighbor, and additive interference, while allowing for more complex, study-specific forms of dependence. We show that the proposed balancing estimators are consistent and asymptotically normal. We also develop an asymptotically valid confidence interval and a data-driven procedure for selecting among candidate low-rank structures. Simulation and empirical studies demonstrate that the proposed estimator is substantially more efficient than the standard IPW estimator. 
   
\end{abstract}

\section{Introduction}
Although most traditional causal inference methods assume the absence of spillover effects, researchers across disciplines find that the outcome of one person is often influenced by the treatments assigned to others beyond their own treatment \citep[e.g.,][]{christakis2007,nickerson2008,banerjee2013,paluck2016}. The past two decades have witnessed the emergence of causal inference methods that explicitly address this problem of ``interference between units,'' with respect to effect estimation and inference \citep[e.g.,][]{rosenbaum2007,hudgensandhalloran2008,ettvanderweele2010,Bowers_Fredrickson_Panagopoulos_2013,aronow2017,papadogeordou219,barkley2020,imai:jian:mala:21,hu2022,papa:etal:22}, experimental designs \citep[e.g.,][]{jagadeesanetal2020,YuChristinaLee2022Ettt,viviano2022experimentaldesignnetworkinterference,kandiros2025conflictgraphdesignestimating,viviano2025causalclusteringdesigncluster}, and policy learning \citep[e.g.,][]{kitagawanadwang2021,chanchenyuandkang2016, viviano2024policy,zhang2025individualizedpolicyevaluationlearning}.

The fundamental challenge of interference is that the number of possible treatment patterns grows exponentially, rendering the treatment space high-dimensional. When the true propensity score is known, causal effects can be estimated using inverse probability weighting (IPW) \citep{ettvanderweele2010,papadogeordou219}. In observational settings, however, the propensity score is unknown and must be estimated. Existing approaches rely on IPW with propensity scores estimated via mixed-effects regression models \citep[e.g.,][]{kang2023propensity,parkkang2022,ettvanderweele2010}. The reliance on parametric models in such high-dimensional settings raises concerns about bias arising from model misspecification.

In the absence of spillover effects, covariate balancing estimators offer a robust alternative to IPW. By directly balancing covariates between treated and control groups, these methods enable causal effect estimation under the unconfoundedness assumption while mitigating bias from model misspecification \citep[e.g.,][]{Hainmueller_2012,imaiandratkovic2013,zubuzaretta2015,benmichael2021balancingactcausalinference}. However, existing covariate balancing approaches are not readily applicable in settings with spillover effects or complex dependencies in treatment assignment across units.

In this paper, we address this gap by developing a general class of covariate balancing estimators that accommodates interference and complex dependencies in treatment assignments across units. We first show that, without restrictions on the interference pattern, no weighting estimator is uniformly unbiased in the absence of knowledge of the propensity score. This result highlights the need for structural assumptions on interference. We introduce \textit{low rank structures} as a flexible class of assumptions that captures a wide range of interference patterns, including anonymous or stratified interference \citep[e.g.,][]{hudgensandhalloran2008}, nearest neighbor interference \citep[e.g.,][]{sussman2017elementsestimationtheorycausal}, and additive interference \citep[e.g.,][]{zhang2025individualizedpolicyevaluationlearning}. Beyond these commonly used cases, our framework can accommodate more complex, study-specific interference patterns not previously considered in the literature.



Exploiting this low-rank structure, we propose covariate balancing estimators that are asymptotically unbiased for a broad class of causal estimands without requiring knowledge of the true propensity score. We establish asymptotic normality of the proposed estimator under clustered network, or partial, interference settings. We further characterize its bias and variance under misspecification of the low-rank structure. Furthermore, we develop a consistent variance estimator and construct asymptotically valid confidence intervals. The proposed approach yields substantial efficiency gains over the standard IPW estimator.

Finally, we propose a data-driven method for selecting among a set of candidate low-rank structures, enabling inference even when knowledge of the interference structure is imperfect. We also show that in the case that the propensity score is known (as happens in experimental designs), our methodology can be modified to allow IPW-style estimation that leverages low-rank structures. Simulation studies confirm the theoretical advantages, and a real-world application demonstrates the practical utility of our proposed methodology.

\subsection*{Related literature}

Significant progress has been made in the development of methods for estimating causal effects in the presence of interference between units. As noted above, most of the existing literature focuses on experimental settings, where the true propensity score is known \citep[e.g.,][]{jagadeesanetal2020,imai:jian:mala:21,viviano2022experimentaldesignnetworkinterference,kandiros2025conflictgraphdesignestimating,viviano2025causalclusteringdesigncluster,toulis2013,eichhorn2024loworderoutcomesclustereddesigns,holtz2025}.  In contrast, the methods for observational studies with interference are relatively underdeveloped.  We develop methods that can be applied to both experimental and observational studies.  

When the propensity score is known, the IPW estimator is commonly employed for estimation. Many consider partial interference, which assumes that interference occurs within each of the non-overlapping clusters of units \citep[e.g.,][]{ettvanderweele2010,hudgensandhalloran2014,liuhudgensbacker-dreps2016,papadogeordou219,barkley2020,imai:jian:mala:21}. Our methods have finite-sample estimation validity and only assume partial interference for asymptotic inference. 

Due to the high variances of IPW estimators, researchers frequently impose additional structural assumptions (beyond partial interference) on the interference pattern. 
A broad class of assumptions are based on a graph on the units, where a unit's outcome is affected only by the treatment status of its neighbors \citep[e.g.,][]{sussman2017elementsestimationtheorycausal,liu2019}. A related assumption is an exposure mapping, which assumes that interference depends on a function of the characteristic of the neighborhood \citep[e.g.,][]{aronow2017,Forastiere03042021}. Anonymous interference (also known as stratified interference) is a special case, where a unit's potential outcome depends only on the number (or proportion) of treated units rather than their identities \citep[e.g.,][]{Sobel01122006,hudgensandhalloran2008}.
More recently, \cite{YuChristinaLee2022Ettt} and \cite{zhang2025individualizedpolicyevaluationlearning} consider additive interference, in which one's potential outcomes depend on an additive function of treatment assignments of all units. 

Many of these interference assumptions may appear to be distinct from one another. As a result, the literature has developed a tailored estimation method under each assumption.
In contrast, our proposed low-rank structure framework encompasses many such restrictions as special cases. 
Our method also shows how the IPW estimator can be modified, given a complex interference structure, to improve its efficiency.

In observational studies with unknown propensity scores, existing work typically employs parametric propensity score models.  The most popular choice is a mixed effect logistic regression \citep[e.g.,][]{parkkang2022,liuhudgensbacker-dreps2016,ettvanderweele2010,perez-heydrich2014}, while  \cite{kang2023propensity} extends this model to a semi-parametric additive model.  Other approaches include doubly robust estimation, in which the outcome is also modeled \citep[e.g.,][]{liu2019,kilpatrick2024}.  However, they also rely on parametric models, suggesting that model misspecification can lead to biased estimates.
In the presence of complex dependence and interference, model misspecification is of particular concern for both the propensity score and outcome models.

To address this problem, we develop a covariate balancing  approach that yields unbiased estimates of causal effects by sidestepping the need to model the propensity score.  While there exist a large number of weighting estimators that achieve covariate balance, all of them assume no interference between units \citep[e.g.,][]{Hainmueller_2012,imaiandratkovic2013,zubuzaretta2015,hazlett2015,wong2018,benmichael2021balancingactcausalinference}. To our knowledge, balancing methods have not been extended to a general observational study under interference.  \cite{ben-michael2024} develops a related balancing method for clustered observational studies under partial interference, in which all units in each cluster are assumed to have the same treatment status. In contrast, our weighting estimator is applicable to a more general interference setting, in which units in each cluster can have different treatment conditions.  

Finally, there is an important connection between our method and the balancing method developed by \cite{additivefactorial} for observational studies with a factorial treatment. Both methods tackle the problem of an exponentially growing number of treatment patterns.
In fact, their method can be seen as a special case of ours when an additive interference structure is imposed, and the cluster sizes are equal; in this case, the number of factors corresponds to the cluster size.
However, our framework is more general in that it can accommodate varying cluster sizes and leverage a broader class of restrictions on interference patterns.  We also analyze how the mis-specification of the low-rank structure affects the performance of the proposed weighting estimator.

\section{Setup}
\label{sec:setup}

In this section, we lay out the assumptions and define causal estimands of interest.  Finally, we also define a general class of weighting estimators considered in the remainder of the paper.

\subsection{Assumptions}
Assume that we have $n$ clusters with the $c^{\mathrm{th}}$ cluster containing $M_c$ units, resulting in a total of $M = \sum_{c=1}^n M_c$ units. Let $\bm X_{ci}\in \mathcal X\subset \mathbb R^p$ and $A_{ci}\in \mathcal{A}=\{0,1\}$ denote $p$ pre-treatment covariates and the binary treatment variable, respectively, for the $i^{\mathrm{th}}$ unit in the $c^{\mathrm{th}}$ cluster.  Note that $\bm X_{ci}$ may contain cluster-level covariates including the cluster size $M_c$, individual-level characteristics, and their interactions.  
Furthermore, we use $\bm X_c = \left(\bm X^\top_{c1}, \ldots, \bm X^\top_{cM_c}\right)^\top \in \mathcal{X}^{\mathcal{M}_c}$ and $\bm A_c = (A_{c1},\dots, A_{cM_c})^\top \in \mathcal{A}^{M_c}$ to denote the stacked vectors of covariates and the treatments across all units in the $c^{\mathrm{th}}$ cluster, respectively.  
Let $\boldsymbol a_c\in \mathcal{A}^{M_c}$ denote a treatment pattern among the units that belong to the $c^{\mathrm{th}}$ cluster, and $\bm a \in \mathcal{A}^{M}$ for the treatment pattern for all $M$ units in the study.
Lastly, $Y_{ci}(\bm a)$ represents the potential outcome for the $i^{\mathrm{th}}$ unit in the $c^{\mathrm{th}}$ cluster under $\bm a \in \mathcal{A}^M$. Throughout this paper, we assume partial interference, as formally outlined next.
\begin{ass}[Partial interference]
 \label{assm:partial}
    For all $i=1,2,\ldots,M_c$, $c=1,2,\ldots,n$, and $\bm a, \bm a' \in \mathcal{A}^M$ with $\bm a_c = \bm a'_c$, we have that $Y_{ci}(\boldsymbol a)=Y_{ci}(\bm a')$. Therefore, a potential outcome can be written as a function of its cluster-level treatment pattern, i.e., $Y_{ci}(\bm a_c)$.
\end{ass}
Next, we use $Y_{ci}$ to denote the observed outcome of the $i^{\mathrm{th}}$ unit in the $c^{\mathrm{th}}$ cluster, which is linked to the potential outcome through the following consistency assumption.
\begin{ass}[Consistency]
\label{assm:consistency}
    For all $i=1,2,\ldots,M_c$, $c=1,2,\ldots,n$, and $\bm a_c \in \mathcal{A}^{M_c}$, $\bm A_c=\bm a_c$ implies $Y_{ci} = Y_{ci}(\bm A_c)$.
\end{ass}
Furthermore, throughout this paper, we also maintain the following standard set of superpopulation causal inference assumptions.
\begin{ass}[IID Sampling of Clusters]
    \label{assm:super_pop}
    There exists a probability measure $\mathcal G$, such that $\left(\left\{Y_{ci}(\bm a_c)\right\}_{i,\bm a_c}, \bm X_c, \bm A_c\right)\iid \mathcal G$.
\end{ass}
Assumption~\ref{assm:super_pop} clarifies that all cluster information, including the cluster characteristics $\bm X_c$ and cluster size $M_c$, are random and drawn i.i.d. from a super-population.
\begin{ass}[Unconfoundedness]
    \label{assm:unconfound}
    At the cluster level, the treatment assignment is independent of the set of all potential outcomes, conditional on covariates.  That is, $\left(\left\{Y_{ci}(\bm a_c)\right\}_{i,\bm a_c}, \bm X_c, \bm A_c\right)\sim \mathcal G$  implies $\bm A_c \ind \{Y_{ci}(\bm a_c)\}_{i, \bm a_c} \mid \bm X_c$ almost surely.
\end{ass}
\begin{ass}[Positivity]
    \label{assm:pos}
    For any $\bm a_c \in \mathcal{A}^{M_c}$, the (cluster-level) propensity score satisfies $0 < e(\bm a_c; \bm X_c) := \mathbb P(\bm A_c = \bm a_c\mid \bm X_c)$, almost surely under $\mathcal G$.
\end{ass} 
All our causal estimands of interest are defined with respect to this super-population distribution.  We emphasize that, under this framework, the cluster-level characteristics $\bm X_c$, which include the cluster size $M_c$, is also a random variable drawn from $\mathcal G$.  Thus, the \textit{propensity score} $e(\bm a_c; \bm x_c)$ defined in Assumption~\ref{assm:pos} is a function of treatment patterns, $\bm a_c$, and covariates, $\bm X_c$, of varying lengths. In this paper, we are primarily concerned with observational studies, in which the propensity score is unknown.

Finally, we assume the following model of potential outcomes.
\begin{ass}[Signal-and-noise model]   \label{assm:outcome_model}
    We assume the following model of potential outcomes,
    \begin{align}
\label{eqn:outcome_model}
        Y_{ci}(\bm a_c) = g_{ci}^{(\bm a_c)}(\bm X_c) + \epsilon_{ci}^{(\bm a_c)},
    \end{align}
    where $g_{ci}^{(\bm a_c)}(\bm x)=\mathbb{E}[Y_{ci}(\bm a_c) \mid \bm X_c=\bm x],\forall \bm x$, is the unit-specific conditional expectation function of the potential outcome under a given cluster-level treatment assignment pattern $\bm a_c \in \mathcal{A}^{\mathcal{M}_c}$, and $\epsilon_{ci}^{(\bm a_c)}$ is mean zero error term independent of the other error terms and $\{(\bm X_c,\bm A_c)\}_{1\leq c\leq n}$.
\end{ass}

This model is flexible, allowing the conditional expectation function for each potential outcome $g^{(\bm a_c)}$ to vary across units and clusters. In addition, the potential outcomes of each unit can depend on the entire set of covariates $\bm X_c$, including cluster-level characteristics such as cluster size, and individual-level characteristics of itself and other units in the cluster.

\subsection{Causal estimand of interest}

In the above setting, we consider a general class of causal estimands. We first introduce the notion of counterfactual weight. 
\begin{defn}[Counterfactual weight] 
    Let $\mathcal Z_{m}:= \mathcal{A}^m \times \mathcal X^m$ and $\mathcal Z:=\bigcup_{m\geq 1}\mathcal Z_m$ where $m$ is a given cluster size. Then, the counterfactual weight is a function $f:\mathcal Z\mapsto \mathbb R$. 
\end{defn}
A prominent example of counterfactual weight represents the probability distribution of treatment patterns resulting from a stochastic intervention, which includes a deterministic intervention as a special case.
\begin{example}[Counterfactual weight of stochastic intervention]
The counterfactual weight represents a stochastic intervention if $\textrm{Range}(f)\subseteq [0,1]$ and for any $m \geq 1$ and fixed $\bm x \in \mathcal{X}^m$,
    \[
        \sum_{\bm a\in \mathcal{A}^m} f(\bm a;\bm x) = 1.
    \]
\end{example}

Based on the definition of counterfactual weight, we define a broad class of causal estimands as the weighted average of potential outcomes among all units in a cluster drawn from $\mathcal G$.
\begin{defn}[A general class of causal estimands]
A general class of causal estimands is defined as the following weighted average of potential outcomes across all units of a cluster, using counterfactual weight $f$:
\begin{align}
\label{eqn:counterfactual}
    \mu_f := \mathbb E\left[\frac{1}{M_c}\sum_{i=1}^{M_c}\sum_{\bm a_c \in \mathcal{A}^{M_c}}Y_{ci}(\bm a_c)f(\bm a_c, \bm X_c)\right].
\end{align}
\end{defn}
In the above definition, if the counterfactual weight represents a stochastic intervention, then the innermost summand within the expectation represents the expected potential outcome for unit $i$ over the distribution of treatment assignments following $f$. However, the counterfactual weight $f$ need not represent a stochastic intervention.
The following examples show that a variety of causal effects belong to the class of causal estimands defined above:
\begin{example}[Global average treatment effect]
    \label{ex:mu1}
    Let $\bm 1_c$ and $\bm 0_c$ represent the $M_c$-dimensional vector of ones and that of zeros, respectively. Define the following counterfactual weight: $f(\bm 1_c,\bm x_c) = 1$, $f(\bm 0_c,\bm x_c) = -1$, and $f(\bm a_c, \bm x_c) = 0$ for all $\bm a_c \notin\{ \bm 1_c,\bm 0_c\}$ and $\bm x_c \in \mathcal{X}^{M_c}$. Then, $\mu_f$ equals the global average treatment effect:
    \[
        \mu_f  = \ \mathbb E\left[\frac{1}{M_c}\sum_{i=1}^{M_c}\left(Y_{ci}(\bm 1_c) - Y_{ci}(\bm 0_c)\right)\right].
    \]
\end{example}

\begin{example}[Direct effect under a stochastic intervention]
    \label{ex:mu2}
    Following \cite{papadogeordou219}, we define the direct effect of one's own treatment with respect to a stochastic intervention $\ell$ as:
    \begin{align*}
&  \tau_{\ell}  
        := \mathbb E\Bigg[\frac{1}{M_c}\sum_{i=1}^{M_c}\sum_{\bm a\in \{0,1\}^{M_c-1}} \Big( Y_{ci}(a_{ci}=1,\bm a_{c,-i} = \bm a)\ell(\bm A_{c,-i}=\bm a\mid A_{ci} = 1,\bm X_c) - \\
        & \hspace{180pt} Y_{ci}(a_{ci}=0,\bm a_{c,-i} = \bm a)\ell(\bm A_{c,-i}=\bm a\mid A_{ci} = 0,\bm X_c) \Big) \Bigg],
    \end{align*}
    where,
    \[
        \ell(\bm A_{c,-i}=\bm a\mid A_{ci} = a^\prime;\bm X_c):= \frac{\ell(\bm A_{c,-i}=\bm a, A_{ci} = a^\prime;\bm X_c)}{\sum_{\bm a \in \{0,1\}^{m_c-1}} \ell(\bm A_{c,-i}=\bm a, A_{ci} = a^\prime;\bm X_c)},
    \]
    for all $\bm a \in \{0,1\}^{M_c-1}$ and $a^\prime \in \{0,1\}$. Now, consider the counterfactual weight,
    \[
    f(\bm a_c, \bm X_c) = \sum_{j=1}^{M_c}(-1)^{1-a_{cj}}\ell(\bm A_{c,-j}=\bm a_{c,-j}\mid A_{c,j}=a_{cj}, \bm X_c),
    \]
    which yields $\mu_{f} = \tau_\ell$.
\end{example}

\subsection{Weighting estimators}

In general, covariate balancing estimators can be viewed as weighing estimators. Therefore, we consider the following broad class of weighting estimators:
$$T({\bm w}) := \frac{1}{n}\sum_{c=1}^n\sum_{i=1}^{M_c} w_{ci}Y_{ci} = \frac{\bm w^\top \bm Y}{n}
$$ 
where the weight vector $\bm w$ is a function of treatment assignment patterns $\bm A$ and covariates $\bm X$, defined as $\bm w := \bm w(\bm A, \bm X)=[\bm w_1(\bm A_{1}, \bm X_{1}), \ldots, \bm w_n(\bm A_{n}, \bm X_{n})]^\top$ with 
$\bm w_{c}(\bm A_{c}, \bm X_{c}) = [w_{c1}(\bm A_{c}, \bm X_{c}), \ldots$, $ w_{c M_c}(\bm A_c, \bm X_c)]^\top$. 

We focus on the choice of $\bm w$ when the true propensity score is unknown. However, if the propensity score is known, this class of weighting estimators includes, as a special case, the following standard IPW estimator,
\begin{equation}
    T_{\mathrm{IPW}} := T({\bm w}_{\mathrm{IPW}}) = \frac{1}{n}\sum_{c=1}^n \sum_{i=1}^{M_c}\frac{f(\bm A_c;\bm X_c)}{M_c\cdot e(\bm A_c;\bm X_c)} Y_{ci}, \label{eqn:IPWest}
\end{equation}
where the IPW weight for the $i^{\text{th}}$ unit in the $c^{\text{th}}$ cluster is given by $w_{\mathrm{IPW},ci}(\bm A_c, \bm X_c)=\frac{f(\bm A_c, \bm X_c)}{M_c\cdot e(\bm A_c, \bm X_c)}$, and $\bm w_\text{IPW}$ is the vector of $w_{\mathrm{IPW},ci}$ stacked first across $i$ and then across $c$.  It is straightforward to show that under Assumptions~\ref{assm:partial}--\ref{assm:outcome_model}, this IPW estimator is unbiased for $\mu_f$:
$$
\mathbb E[T_{\mathrm{IPW}}] = \frac{1}{n} \sum_{c=1}^n \sum_{i=1}^{M_c} \mathbb E\left[\mathbb E\left[Y_{ci}(\bm A_c)\frac{f(\bm A_c,\bm X_c)}{M_c\cdot e(\bm A_c, \bm X_c)}\ \middle| \ \bm X_c\right]\right] = \mu_f.
$$

Given this setup, we develop unbiased balancing estimators when the propensity score is unknown, which, as we show, requires a structural assumption on the interference pattern.

\section{Low-rank weighting estimators}
\label{sec:low_rank}

In this section, we introduce a flexible class of assumptions on the outcome regression functions, referred to as \textit{low-rank structures}. Special cases of this class include many commonly invoked restrictions on interference in the literature. We leverage this general class of low-rank structures and derive a class of weighting estimators. Furthermore, we show how to choose optimal weights and derive the statistical properties of the resulting optimal weighting estimators.

\subsection{Uniformly unbiased weighting estimators}
\label{sec:unbiased_weights}

We begin by showing that, without restrictions on the interference pattern or knowledge of the true propensity score, uniformly unbiased estimation of $\mu_f$ is impossible. This result motivates the need to impose structure on the interference pattern. 

We begin by introducing some additional notation.  Fix a lexicographic ordering on all the possible $2^M$ treatment patterns across all units in the sample and define the enumeration $\{0,1\}^{M} = \{\bm a^{(j)}: j=0,1,\dots, 2^{M}-1\}$. Define $\bm R$ as an $M\times (M\cdot 2^M)$ block matrix given by,
\begin{align}
\label{eqn:r}
\bm R := (\mathbb I(\bm A = \bm a^{(0)}), \ldots, \mathbb I(\bm A = \bm a^{(2^M-1)}))^\top\otimes \bm I_M,
\end{align}
where $\mathbb I(E)$ denotes the indicator of event $E$, and $\bm I_M$ denotes the $M\times M$ identity matrix. Next, for each $\bm a \in \mathcal{A}^M$, define $\bm g^{(\bm a)}:=[g_{11}^{(\bm a_1)}(\bm X_1), g_{12}^{(\bm  a_1)}(\bm X_1),\ldots, g_{nM_n}^{(\bm a_{n})}(\bm X_n)]$ as an ${M}$-dimensional vector of outcomes under treatment pattern $\bm a$ with $g_{ci}^{(\bm a_c)}(\bm X_c)$ arranged first across $i$ and then across $c$.  Let $\bm g$ be an $(M \cdot 2^M)$-dimensional vector, defined as
$$\bm g:= \left({\bm g^{(\bm a^{(0)})}}^\top, \dots, {\bm g^{(\bm a^{(2^M-1)})}}^\top\right)^\top.$$
Then, Equation~\eqref{eqn:outcome_model} implies that the model for the observed data can be written as
\begin{equation}
    \bm Y = \bm R\bm g + \bm \epsilon, \label{eq:model}
\end{equation}
where a typical element of $\bm \epsilon$ is $\epsilon_{ci}^{(\bm A_c)}$ and $\bm \epsilon$ is also arranged first across $i$ and then across $c$.

Now, the unbiasedness of the IPW estimator implies $\mu_f = \mathbb E[{\bm w_{\mathrm{IPW}}}^\top\bm R\bm g]$. Hence, the class of weights $\bm w$ that makes the estimator $T(\bm w)$ unbiased is given by
\begin{equation}
    \mathcal C(\bm g, e) = \left\{\bm w: \mathbb E\left[(\bm w - {\bm w}_{\mathrm{IPW}})^\top\bm R\bm g\right] = 0\right\}. \label{eq:unbiased_estimators}
\end{equation}
The expectation in Equation~\eqref{eq:unbiased_estimators} is taken with respect to the sampling of clusters from $\mathcal{G}$ with the unknown, true propensity score. 

Unfortunately, since $\bm g$ is unknown, it is not possible to further characterize this class of unbiased estimators $\mathcal C(\bm g, e)$ based on the observed data. The following theorem shows that unbiased weighted estimation of causal effects is not possible without additional assumptions. Moreover, if we were to assume that the propensity score is known, the IPW estimator is the only uniformly unbiased weighting estimator.
\begin{theorem}[Uniformly unbiased weighting estimators]
    \label{thm:subclass1} Suppose that Assumption~\ref{assm:partial}--~\ref{assm:unconfound} holds. 
    If the propensity score $e$ is known, then the IPW estimator is the only uniformly unbiased weighting estimator over all possible choices of $\bm g$, that is, $\bm w \in \mathcal C(\bm g, e)$ for all $\bm g$ implies $\bm w = \bm w_{\mathrm{IPW}}$. If the propensity score $e$ is unknown and 
    \begin{align}
    \label{eqn:impossibility_condition}
        \mathbb P\left(\exists~c~{\rm and }~\bm a_c\in \{0,1\}^{M_c}\setminus\{\bm A_c\}~\textrm{ such that }f(\bm a_c, \bm X_c)\neq 0\right)>0.
    \end{align}
    Then there is no uniformly unbiased weighting estimator over all possible choices of $\bm g$ and $e$.
\end{theorem}
This result motivates the need for a structural assumption that places restrictions on $\bm g$.  Next, we introduce a general class of such restrictions.

\subsection{Low-rank structure as a general assumption about interference}
\label{sec:examples_low_rank}

The following assumption forms the basis of all methodological developments in this paper.
\begin{ass}[Low-rank structure of the outcome model] \label{ass:low_rank}
    In the model of Equation~\eqref{eq:model}, $\bm g$ lies in a lower dimensional subspace and satisfies the following \textit{low-rank structure},
    \begin{align}
    \label{eqn:low_rank_assm}
    \bm g = \bm \Lambda \bm h,
    \end{align}
    where $\bm h$ is a vector with $\dim(\bm h) \leq \dim(\bm g)$, and $\bm \Lambda$ can depend on $\bm X$, but not on the observed treatment patterns, $\bm A$.
\end{ass}

We now demonstrate that this low-rank structure encompasses, as special cases, a wide range of common interference assumptions used in the literature. Furthermore, we later show in Section~\ref{sec:lm_balancing} that functional-form assumptions, such as a linear model for the outcome regression, can also be encoded as a low-rank structure.
\begin{example}[Stratified/anonymous interference \citep{Sobel01122006,hudgensandhalloran2008}]
\label{ex:1}
    The assumption states that within each cluster, the potential outcome of one unit depends only on the number of treated units in the same cluster, regardless of who receives the treatment.  Let $r_{ci}(\bm a_c) \le s \le M_c$ denote the number of treated units in the prespecified neighborhood of unit $i$ in cluster $c$, and $s$ is a global upper bound on neighborhood size.  Under this assumption, we have $g_{ci}^{(\bm a_c)}(\bm X_c) = h_{ci}^{(r_{ci}(\bm a_c))}(\bm X_c)$ for all $c,i,\bm a_c$.  Thus, $g_{ci}^{(\bm a_c)}(\bm X_c)$ takes one value in the set $\{h_{ci}^{(j)}(\bm X_c):0\leq j\leq s\}$, which depends on $\bm X_c$. We form $\bm h$ by stacking $h_{ci}^{(j)}$ across $j$, $i$, and $c$ in that order.  Finally, let $\bm u_t$ denote an $s$-dimensional vector whose $t^{\mathrm{th}}$ element is equal to one while all the other elements are zero. Then, we have $\bm g = \bm \Lambda \bm h$, where $\bm \Lambda$ is a $(\sum_{c=1}^nM_c\cdot2^{M_c})\times (s\cdot M)$ matrix given by,
\begin{align*}
    \bm \Lambda =
    \begin{bmatrix}
\textrm{diag}(\bm u(\bm a^{(0)})) \\
\textrm{diag}(\bm u(\bm a^{(1)})) \\
\vdots \\
\textrm{diag}(\bm u(\bm a^{(2^M-1)}))
    \end{bmatrix},
    \quad \text{where} \quad \textrm{diag}(\bm u(\bm a^{(j)})) = \begin{bmatrix}
        \bm u_{r_{11}(\bm a^{(j)}_{1})} & \bm 0 & \cdots & \bm 0\\
        \bm 0 & \bm u_{r_{12}(\bm a^{(j)}_{1})} &  \cdots & \bm 0\\
        \vdots & \vdots & \ddots & \vdots\\
        \bm 0 & \bm 0 & \cdots & \bm u_{r_{nM_n}(\bm a^{(j)}_{n})}    
    \end{bmatrix}
    \end{align*}
    for $j=1,2,\ldots,2^M-1$.
\end{example}

\begin{example}[$k$-nearest neighbor interference \citep{sussman2017elementsestimationtheorycausal,alzubaidihiggins2024}]
\label{ex:2} The assumption states that the potential outcome of one unit depends only on the treatment pattern of its $k$-nearest neighbors, which are determined using a distance metric based on covariates.  
Define $l_{ci}(\bm a_c; k, \bm X_c)$ as the index for a total of $2^k$ possible treatment patterns among the $k$-nearest neighbors for the $i^{\mathrm{th}}$ unit in the $c^{\mathrm{th}}$ cluster.  Then, for any unit $i$ and treatment pattern $\bm a_c$, we have $g_{ci}^{(\bm a_c)}(\bm X_c) = h_{ci}^{(l_{ci}(\bm a_c; k, \bm X_c))}(\bm X_c)$ for some set of functions $\left\{h_{ci}^{(\bm j)}(\bm X_c)\right\}_{\bm j\in \{0,1\}^k}$. As done in Example~\ref{ex:1}, define $\bm u_{l}$ as a $2^k$-dimensional vector whose $l^{\mathrm{th}}$ element is equal to one and 0 elsewhere. Then, we can define $\bm \Lambda$ exactly as in Example \ref{ex:1} by replacing $\bm u_{r_{ci}(\bm a^{(j)}_c)}$ with $\bm u_{l_{ci}(\bm a^{(j)}_c;k,\bm X_c)}$.
\end{example}

\begin{example}[Additive interference \citep{YuChristinaLee2022Ettt,zhang2025individualizedpolicyevaluationlearning}]
\label{ex:3}
This assumption states that the conditional expectation of the potential outcome is an additive function of treatment assignments of units in the same cluster, implying
$$
g^{(\bm a_c)}_{ci}(\bm X_c) = \sum_{j=1}^{M_c}h_{cij}^{(a_{cj})}(\bm X_c),
$$
where $h_{cij}^{(a_{cj})}(\bm X_c)$ denotes the effect of the $j^{\mathrm{th}}$ unit in the $c^\mathrm{th}$ cluster whose treatment status is $a_{cj}\in \{0,1\}$ on the $i^{\mathrm{th}}$ unit of the same cluster. Define $\bm h$ as the $2^{\sum_c 2M_c^2}$-dimensional vector formed by stacking $\left(h_{cij}^{(0)}(\bm X_c),h_{cij}^{(1)}(\bm X_c)\right)^\top$ across $j$, $i$, and then $c$ in that order. Next, define the $2M_c$-dimensional vector $\bm v_{ci}(\bm a_c) := (\mathbb I(a_{c1} = 0), \mathbb I(a_{c1} = 1),\dots, \mathbb I(a_{c{M_c}} = 0), \mathbb I(a_{c{M_c}} = 1))^\top$. Then, the low-rank structure assumption is satisfied with the following $\bm \Lambda$, which is a $\sum_{c=1}^nM_c\cdot2^{M_c}\times \sum_{c=1}^n2M_c^2$ matrix,
\begin{align*}
    \bm \Lambda = \begin{bmatrix}
        \bm v_{11}(\bm a_1) & \bm 0 & \cdots & \bm 0\\
        \bm 0 & \bm v_{12}(\bm a_1) & \cdots & \bm 0\\
        \vdots & \vdots & \ddots & \vdots\\
        \bm 0 & \bm 0 & \cdots & \bm v_{nM_n}(\bm a_n)
    \end{bmatrix}.
\end{align*}
\end{example}

Moreover, the proposed low-rank assumption can combine multiple restrictions, thereby enabling researchers to specify their own study-specific low-rank structures.  We illustrate this feature of the proposed methodology with the following example, which combines the restrictions from Examples~\ref{ex:2}~and~\ref{ex:3}.
\begin{example}[$k$-nearest neighbor additive interference model] 
\label{ex:4} Suppose that for each unit, interference is an additive function of the treatment assignments of the $k$-nearest neighbors. To formulate this restriction, we first write $\bm g = \bm \Lambda_1\bm h$, where $\bm \Lambda_1$ encodes the $k$-nearest neighbor interference restriction from Example~\ref{ex:2} and $\bm h$ has entries $h_{ci}^{(l_{ci}(\bm a_c; k, \bm X_c))}(\bm X_c)$. We then place the additive restriction on $\bm h$ as done in Example~\ref{ex:3} and write it as $\bm h = \bm \Lambda_2\bm q$, where $\bm q$ has a typical entry $q_{cij}^{(a_{cj})}(\bm X_c)$ that quantifies the additive effect of the $j^{\mathrm{th}}$ nearest neighbor on the $i^{\mathrm{th}}$ unit in the $c^{\mathrm{th}}$ cluster when the treatment status of unit $j$ is $a_{cj}$. As a result, we have $\bm \Lambda = \bm \Lambda_1\bm \Lambda_2$. In general, one can follow this strategy to break down a complex low-rank structure into an ensemble of simpler structures to obtain $\bm \Lambda$.
\end{example}

Imposing low-rank structures changes the unknown object from $\bm g$ to $\bm h$, where the elements of $\bm h$ can be viewed as \textit{effective treatment effects}. In Example~\ref{ex:2}, each $\bm h_{ci}^{(\bm j)}(\bm X_c)$ denotes the effective treatment effect of having the treatment pattern $\bm j$ in the neighborhood. Similarly, in Example~\ref{ex:3}, $h_{cij}^{(t)}(\bm X_c)$ is the effective treatment effect, representing the additive contribution of the $j^{\mathrm{th}}$ unit under the treatment status $t \in \{0,1\}$. The complexity of this effective treatment space is much smaller than that of the original treatment space allowing us to improve on estimation efficiency.

\subsection{Deriving optimal weights under the low-rank structure}
\label{sec:low_rank_weights}

We now show how to leverage the low-rank structure and obtain weighting estimators without knowledge of the true propensity score.
Under the low-rank assumption, the class of unbiased weighting estimators is a function of $\bm h$ rather than $\bm g$ as in Equation~\eqref{eq:unbiased_estimators} and can be expressed as:
\begin{equation}
\mathcal C_{\bm \Lambda}(\bm h, e) = \left\{\hat{\bm w}: \mathbb E\left[(\hat{\bm w} - \hat{\bm w}_{\mathrm{IPW}})^\top\bm R \bm \Lambda \bm h\right] = 0\right\}.
\label{eq:unbiased_estimators_R}
\end{equation}
We extend the discussion that precedes Theorem~\ref{thm:subclass1} by requiring our weighting estimator to be uniformly unbiased across choices of $\bm h$ and $e$. The following theorem shows that, unlike Theorem~\ref{thm:subclass1}, this requirement can yield non-trivial classes of weights. 
\begin{theorem}[A uniformly unbiased weighting estimator under a low-rank structure]
\label{thm:subclass2}
Suppose that Assumptions~\ref{assm:partial}--~\ref{assm:unconfound}~and~\ref{ass:low_rank} hold. 
Then, a uniformly unbiased weighting estimator (across all possible $\bm h$ and $e$ in Equation~\eqref{eq:unbiased_estimators_R}), whose weights do not depend on the unknown true propensity score, satisfies the following relation:
\begin{align}
\label{eqn:baldefn}
\bm w \in \mathcal C_{\bm \Lambda}(\bm h, e),\forall \bm h,e \implies \bm w \in \Cnulllamb:=\{\hat{\bm w}: \bm \Lambda^\top\bm R^\top\hat{\bm w} = \bm \Lambda^\top\bm f\},
\end{align}
where $\bm f$ denotes the vector formed by stacking $$\bm f^{(j)}:=\left(f(\bm a_1^{(j)};\bm X_1)(\bm 1_{M_1}\otimes \bm I_{d_h})^\top,\cdots, f(\bm a_n^{(j)};\bm X_n)(\bm 1_{M_n}\otimes \bm I_{d_h})^\top\right)^\top$$ on top of each other across $j=0,\cdots, 2^M-1$, where $d_h = \mathrm{dim}(\bm h)$.
\end{theorem}
Note that $\Cnulllamb$ is empty if the following equation does not have a solution in ${\bm w}$,
\begin{align} \label{eqn:balancing_equation}
    \bm \Lambda^\top\bm R^\top{\bm w} = \bm \Lambda^\top\bm f.
\end{align}
In fact, the proof of Theorem~\ref{thm:subclass2} shown in Appendix~\ref{sec:proof_subclass2} suggests that the difference in the left and right hand sides of this equation is proportional to the bias of the resulting weighting estimator. Hence, we call Equation~\eqref{eqn:balancing_equation} the \textit{balancing equation under interference} and $\Cnulllamb$ the class of \textit{balancing weights}.  This method of covariate adjustment differs from how an IPW estimator achieves balance. While the IPW weights specifically aim at weighting a data point to make the resulting weighted distribution resemble the target population, the balancing approach for interference outlined above directly obtains weights to eliminate the bias of the resulting weighting estimator. 

According to Theorem~\ref{thm:subclass2}, a weighting estimator with balancing weights does not result in an unconditionally unbiased estimate of $\mu_f$ as indicated by the uni-dimensional implication in Equation~\eqref{eqn:baldefn}. The reason is that if $B_{\bm \Lambda}$ denotes the event that Equation~\eqref{eqn:balancing_equation} admits a solution (i.e., the class of balancing weights is non-empty), then for any $\bm w \in \Cnulllamb$, we have $\mathbb E[T(\bm w)\mid B_{\bm \Lambda}] = \frac{1}{n}\mathbb E\left[\bm f^\top\bm \Lambda\bm h\mid B_{\bm \Lambda}\right]$, which is not necessarily equal to $\mu_f = \frac{1}{n}\mathbb E\left[\bm f^\top\bm \Lambda\bm h\right]$. However, in Section~\ref{sec:asymp_normality}, we show that under minimal restrictions, $\mathbb P(B_{\bm \Lambda})$ approaches 1 as $n\to \infty$. This make the balancing estimators asymptotically unbiased for $\mu_f$, under, for example, boundedness assumptions on $\bm f$, $\bm h$ and $\bm \Lambda$. As our simulation study in Appendix~\ref{sec:further_simulations_fixed_bias} illustrates, we often find this bias to be low in finite samples under various interference settings. Thus, with moderate and large samples, the balancing estimator is an effective choice for estimating $\mu_f$. 

Theorem~\ref{thm:subclass2} shows that by exploiting the low-rank structure assumption, which is encapsulated by $\bm \Lambda$, one can now restrict the search of unbiased weights to the much smaller class $\Cnulllamb$. This class represents a solution space of $(\bm R\bm \Lambda)^\top$, which is restricted by the complexity of $\bm \Lambda$ as further discussed in Section~\ref{sec:lm_balancing}. Besides enabling non-trivial 
 estimation, this restriction makes it possible to identify the minimum-variance weighting estimator within this class as we show next.

An optimal choice of weight within this class can be written as a solution to the following optimization problem:
\begin{equation}   \label{eqn:optimiation_problem}
\wnulllamb = \underset{w \in \Cnulllamb}{\arg\min}~\phi(\bm w),
\end{equation}
where $\phi$ is an objective function of choice. In this paper, we use $\phinull(\bm w):=\|\bm w\|^2$ while selecting a weight from $\Cnulllamb$. As the next theorem shows, this choice directly targets the variance of the proposed estimator and, under homoskedasticity, yields the uniformly minimum variance estimators within $\Cnulllamb$.

\begin{ass}[Homoskedasticity]
    \label{assm:iid_error} 
    Under the outcome model given in Equation~\eqref{eqn:outcome_model}, the error terms $\epsilon_{ci}^{(\bm a_c)}$ are i.i.d. across the indices $c,i$ and $\bm a_c$ and have a constant variance $\sigma^2>0$.
\end{ass}

\begin{theorem}[Conditionally unbiased estimator with minimum variance]
\label{thm:umvue}
Suppose that Assumptions~\ref{assm:partial}--\ref{assm:outcome_model},~\ref{ass:low_rank},~and~\ref{assm:iid_error} hold. Recall that $B_{\bm \Lambda}$ denotes the event that Equation~\eqref{eqn:balancing_equation} has a solution. Then, $\forall \bm w \in \Cnulllamb$,
\[
 n^2\mathrm{Var}(T(\bm w)\mid B_{\bm \Lambda}) = \mathrm{Var}\left(\mathbb E\left[\bm w_{\mathrm{IPW}}^\top\bm R\bm g\mid\bm X\right]\ \Bigl | \ B_{\bm \Lambda}\right) + \sigma^2\mathbb E\left[\|\bm w\|^2\mid B_{\bm \Lambda} \right].
\]
Furthermore, $\wnulllamb$ corresponds to the weighting estimator with the smalllest variance: 
\[\mathrm{Var}(T(\wnulllamb)\mid B_{\bm \Lambda})\leq \mathrm{Var}(T(\bm w)\mid B_{\bm \Lambda})\textrm{, for all }\bm w \in \Cnulllamb.
\]
\end{theorem}

The following remarks are in order. First, because the \textit{balancing estimator} $T(\wnulllamb)$ exists only when the balancing equation has a solution, we evaluate its variance conditional on the event $B_{\bm \Lambda}$ in finite samples. Second, Theorem~\ref{thm:umvue} shows that the variance of the proposed estimator can be divided into two components. The first component only depends on the inner product between the IPW weights and the outcome regression functions (and in particular does not depend on the low-rank structure), while the second component is a function of the norm of the weights alone, which is minimized. 

Finally, we derive a closed-form expression of this optimal weight.
\begin{theorem}[Expressions of the low-rank weighting estimators]
\label{thm:weights_expression} If
$\wnulllamb$ exists, it takes the form  $\wnulllamb = (\bm \Lambda^\top\bm R^\top)^{+}\bm \Lambda^\top\bm f$ where $\bm X^{+}$ denotes the Moore-Penrose pseudo-inverse of $\bm X$.
\end{theorem}
Theorem~\ref{thm:weights_expression} follows directly from the fact that pseudoinverse estimators are the minimum-normed solution to a least-squares problem. Note that, regardless of whether the balancing equations are feasible or not, one can always compute this mathematical expression for $\wnulllamb$. However, $\Cnulllamb$ is non-empty and contains this quantity only when the balancing equation is feasible. In fact, it enables us to determine whether the optimization problem in Equation~\eqref{eqn:optimiation_problem} is feasible or not by directly checking whether $(\bm \Lambda^\top\bm R^\top)^{+}\bm \Lambda^\top\bm f\in \Cnulllamb$. We can also use Theorem~\ref{thm:weights_expression} to obtain an expression of the second component of $\mathrm{Var}\left(T(\wnulllamb)\mid B_{\bm \Lambda}\right)$ from Theorem~\ref{thm:umvue}. We thus see that it is only this second component that is affected by the low-rank structure.

As noted earlier, because of conditioning on $B_{\bm \Lambda}$, we are unable to establish the finite-sample unbiasedness of this optimal balancing estimator.  However, in Section~\ref{sec:asymptotics}, we analyze the asymptotic behavior of the estimator and derive various properties, including its asymptotic unbiasedness for $\mu_f$.

We note that all the finite-sample properties of the balancing estimator discussed in this section holds even for a single network $n=1$ with arbitrary interference. The partial interference assumption becomes necessary for statistical inference as discussed in Section~\ref{sec:asymptotics}, where we develop asymptotic theories as the nubmer of clusters tend to infinity. 

\subsection{Linear outcome models and balancing estimators}
\label{sec:lm_balancing}

Recall that the balancing estimator requires $\Cnulllamb$ to be nonempty, which holds when $\bm \Lambda^\top \bm f$ lies in the column space of $(\bm R \bm \Lambda)^\top$. In this case, the dimension of $\Cnulllamb$ equals that of the corresponding solution space, i.e., the null space of $(\bm R \bm \Lambda)^\top$, which is a $\mathrm{dim}(\bm h) \times M$ matrix. A sufficient condition for this null space to be nontrivial is that the length of $\bm h$ is less than $M$.  However, the low rank assumptions in Examples~\ref{ex:1}–\ref{ex:4} take the general form $g_{ij}^{(\bm a_c)}(\bm X_c) = \bm \Lambda_{ij}(\bm a_c)\bm h_{ci}$ for all $c, i$. If at least one $\bm h_{ci}$ is multi-dimensional, then $\bm h$, formed by stacking the $\bm h_{ci}$, has length greater than $M$.

The core issue is that these low rank structural assumptions allow the length of $\bm h$ to grow with $M$, since each $\bm h_{ci}$ is distinct. To be able to solve the balancing equation, it is therefore necessary to impose additional restrictions that relate the $\bm h_{ci}$ terms and prevent $\dim(\bm h)$ from increasing linearly with $M$. As we show next, a linear model assumption for the outcome provides a natural way to achieve this without sacrificing interpretability.  We emphasize that even under no interference, the linear model assumption is generally required for the unbiasedness of balancing estimators \citep{Hainmueller_2012,zubuzaretta2015}. 
\begin{example}[Linear model]
\label{ex:5}
Consider the following linear model: $\bm g_{ci}^{(\bm a_c)}(\bm X_c) = \bm X_{ci}^\top\bm \beta_{ci}{(\bm a_c)}$. Form the vector $\bm \beta$ by stacking $\bm \beta_{ci}{(\bm a_c)}$ across $i$, $c$, and $\bm a$ in that order. Then, we can define $\bm \Lambda(\bm X)$ as the following block diagonal matrix:
$$
\bm \Lambda(\bm X) = \begin{bmatrix}
    \bm \Lambda_1(\bm X_1) & \bm 0 & \cdots & \bm 0\\
    \bm 0 & \bm \Lambda_2(\bm X_2) & \cdots & \bm 0 \\
    \vdots & \vdots & \ddots & \vdots \\
    \bm 0 & \bm 0 & \cdots & \bm \Lambda_n(\bm X_n)
\end{bmatrix} \ \text{where} \ \bm \Lambda_c(\bm X_c) \ = \ 
\bm I_{2^{M_c}}\otimes \begin{bmatrix}
    \bm X_{c1}^\top & \bm 0^\top & \cdots & \bm 0^\top\\
    \bm 0^\top & \bm X_{c2}^\top & \cdots& \bm 0^\top\\
    \vdots & \vdots & \ddots & \vdots\\
    \bm 0^\top & \bm 0^\top & \cdots & \bm X_{cM_c}^\top
    \end{bmatrix},
$$
and $\bm g = \bm \Lambda(\bm X)\bm \beta$.
\end{example}

Although this model does not lead to any dimension reduction (i.e., $\dim(\bm \beta)=\dim(\bm g)$), it separates the treatment effect parameters $\bm \beta$ from covariates. One can then impose a restriction directly on $\bm \beta$ as $\bm \beta = \bm \Lambda^{(\bm \beta)}\bm h$, leading to $\bm g = \bm \Lambda \bm h$, where $\bm \Lambda = \bm \Lambda(\bm X)\bm \Lambda^{(\bm \beta)}$. Here, one can interpret $\bm \Lambda^{(\bm \beta)}$ as a \textit{coefficient-level} low-rank structure matrix. Note that the model still allows for heterogeneous effects through the covariates, which are now absorbed into the low-rank structure (along with the restrictions on $\{\bm \beta_{ci}(\bm a_c)\}$). Because any information about the interference is directly imposed on $\{\bm \beta_{ci}(\bm a_c)\}$, which are of fixed length, this allows for the resulting $\bm h$ to be of fixed-length. 
For example, if all the clusters have the same size, $M/n$, one may consider the following constant effect assumption: $\bm \beta_{ci}{(\bm a_c)} = \bm h^{(\bm a_c)}$ across all $i$ and $c$, for any $\bm a_c$. Because all units share the same coefficients under this assumption, the length of $\bm{h}$, which is formed by stacking $\bm h(\bm a_c)$ across $\bm a_c \in \{0,1\}^{M/n}$, is reduced by a factor of $M$ compared to the length of $\bm{g}$.

In Appendix~\ref{sec:linear_low_rank_ex}, we explicitly show how one can now directly impose the assumptions discussed in the examples of Section~\ref{sec:examples_low_rank} on $\bm \beta_{ci}{(\bm a_c)}$ and obtain their analogs for the above linear model. We emphasize that the linear model introduced above is flexible. In particular, the set of covariates can include any set of basis functions that spans $\left\{g_{ci}^{(\bm a_c)}(\cdot
):\bm a_c\right\}$. The choice of basis functions has been extensively discussed in the literature for the non-interference case (for example, see \cite[Section 5.2]{benmichael2021balancingactcausalinference} and the references therein). A similar discussion extends to the interference case.


Lastly, in Appendix~\ref{sec:plug-in}, we show that the plug-in estimator of $\mu_f$ using an OLS estimate of $\bm h$ is equivalent to $T(\wnulllamb)$. Despite this equivalence, our covariate balancing framework provides several practically important features that are not available under this model-based approach. First, one still needs to check the feasibility of the balancing equations to determine whether this plug-in estimator is unbiased for $\frac{1}{n}\mathbb E[\bm f^\top\bm \Lambda \bm h\mid B_{\bm \Lambda}]$. Second, it is through this balancing framework that we establish the variance-optimality of the plug-in estimator (equal to the balancing estimator) in $\Cnulllamb$. Third, the class of covariate balancing estimators $\Cnulllamb$ is general and allows one to flexibly choose different optimal weights. For example, the aforementioned equivalence relation only holds when the objective function $\phinull$ is squared norm, but other choices, such as entropy \citep{Hainmueller_2012,zhaopercival2016}, are also possible.  Fourth, our framework enables us to construct a balancing estimator that tolerates a pre-specified degree of covariate imbalance, which in some cases may lead to a substantial efficiency gain at the cost of small bias \citep[see][]{zubuzaretta2015}.  We leave the detailed analysis of these extensions for future work. In the next section, we propose a method to assess imbalance when the balancing equations are not feasible.

\subsection{Assessing covariate imbalance when the balancing equation is infeasible}
\label{sec:assessing_imbalance}

When the balancing equation does not have a feasible solution, the balancing estimator will be biased. In such situations, it is important to quantify the degree of covariate imbalance for assessing the potential bias of the weighting estimator \citep{rubin2008}.  Below, we consider a balance assessment method in the setting with interference that reduces to a standard measure when there is no spillover effects \citep[e.g.,][]{zubuzaretta2015}.

For any given weighting estimator $T(\bm w)$, we define the \textit{imbalance vector} as $$\bm \nu := \frac{1}{n}(\bm \Lambda^\top\bm R^\top\bm w - \bm \Lambda^\top\bm f).$$ 
This $j^{th}$ component of the imbalance vector, $\nu_j$, describes the degree of slack in the $j^{th}$ component of the system of balancing equations in Equation~\eqref{eqn:balancing_equation}.
When balance is achieved for the $\mathrm{dim}(\bm h)$ system (i.e., the balancing equations are feasible), $\nu_j=0,\forall j$. If the value of $\nu_j$ is small for all $j$, \textit{imbalance} is small in all the balancing equations, indicating a small bias of the weighting estimator. In fact, the relation between imbalance and bias is given by 
$$\mathbb E[\bm \nu^\top\bm h] = \mathbb E[T(\bm w)] - \mu_f.$$  
Thus, if all the components of $\bm \nu$ are bounded between $-\delta$ and $\delta$ (for some $\delta>0$), the absolute bias of the estimator is bounded by $\delta\|\bm h\|_1$. 

Since there is no reference scale for interpreting the magnitude of $\nu_j$, we consider a measure of \textit{relative imbalance} based on the notion of the effective treatment introduced at the end of Section~\ref{sec:examples_low_rank}. In the case of the linear model assumption (as discussed in Section~\ref{sec:lm_balancing}), where the low-rank structure induces $\ell$ effective treatments, we divide $\bm h$ into $\ell$ vectors of length $p$, i.e., $\bm h = (\bm h_1^\top,\cdots, \bm h_{\ell}^\top)^\top$, where $\bm h_j\in \mathbb R^p$ and $\bm \Lambda_{ci}(\bm a_c)^\top\bm h_j$ represents the $j^{\mathrm{th}}$ effective treatment effect on the $i^{\mathrm{th}}$ unit in the $c^{\mathrm{th}}$ cluster under cluster treatment pattern $\bm a_c$. We similarly write $\bm \nu = (\bm \nu_1^\top,\cdots, \bm \nu_{\ell}^\top)^\top$, for $\bm \nu_j\in \mathbb R^p$. 

Recall that $\mu_f = \frac{1}{n}\mathbb E[\bm \Lambda^\top \bm h] = \mathbb E\left[\frac{1}{M_c}\sum_{i=1}^{M_c}\sum_{\bm a_c}\bm \Lambda_{ci}(\bm a_c)^\top\bm h\right]$. To obtain a scale of comparison for the $t^{\mathrm{th}}$ component of $\bm \nu_j$, we use the standard deviation of the covariate pre-multiplying the $t^{\mathrm{th}}$ component of $\bm h_j$ in the expression of $\mu_f$ above, that is:
\[
    \sigma_{tj} = \sqrt{\Var\left(\frac{1}{M_c}\sum_{i=1}^n\sum_{\bm a_c\in \{0,1\}^{M_c}}\bm \Lambda_{ci}(\bm a_c)^\top\bm u_{tj}\right)},
\]
where $\bm us_{tj}$ is a Euclidean vector of length $p\cdot \ell$, with 1 at the position corresponding to the $t^{\mathrm{th}}$ coordinate for the $j^{\mathrm{th}}$ effective treatment. 

Let $\hat \sigma_{tj}$ denote its empirical analogue and $\nu_{tj}$ denote the imbalance in the $t^{\mathrm{th}}$ covariate multiplying the $j^{\mathrm{th}}$ effective treatment (which is also equal to the $t^{\mathrm{th}}$ component of $\bm \nu_j$). We then define the corresponding \textit{relative imbalance} as $\nu_{tj}^*:=\nu_{tj}/\hat{\sigma}_{tj}$. We also construct an omnibus measure of imbalance for the $t^{\mathrm{th}}$ covariate (for $1\leq t\leq p$) as follows: Let $m_j$ be the number of units for which the $j^{\mathrm{th}}$ effective treatment is observed. Then, $\bar \nu_t = \sum_{1\leq j\leq \ell} \nu^*_{tj}m_j/\sum_{1\leq j\leq \ell}m_j$ is an overall measure of the relative imbalance in the $t^{\mathrm{th}}$ covariate. As a rule of thumb, we recommend a threshold of $10\%$ for these quantities.

\section{Asymptotic inference}
\label{sec:asymptotics}

In this section, we derive the asymptotic properties of our balancing estimator and evaluate its asymptotic efficiency. We first show that the proposed balancing estimator without the knowledge of the true propensity score, is more efficient than the IPW estimator based on the true propensity score. Furthermore, we propose a consistent estimator of the variance that yields asymptotically valid confidence intervals for the balancing estimator. Lastly, we develop data-adaptive methodologies for choosing a suitable low-rank structure among several candidate structures.

\subsection{Asymptotic normality}
\label{sec:asymp_normality}

We first state the asymptotic normality of the IPW estimator, which provides our benchmark for the comparison of the estimators' asymptotic variances.
\begin{theorem}[Asymptotic normality of the standard IPW estimator]
\label{thm:ipw_normality}
Suppose that Assumptions~\ref{assm:partial}--\ref{assm:outcome_model} and \ref{assm:iid_error} hold and the propensity score is known. Let $\bm g_c^{(\bm A_c)}(\bm X_c)$ be the vector formed by stacking $g_{ci}^{(\bm A_c)}(\bm X_c)$ across $i$, and similarly define the vector $\bm w_{\mathrm{IPW},c}$. Then, assuming $\mathbb E\left\|\bm g^{(\bm A_1)}_{1}(\bm X_1)\right\|^4<\infty$ and $\mathbb E\left\|\bm w_{\mathrm{IPW},1}\right\|^4<\infty$, we have, 
\begin{align*}
    \sqrt{n}\left(T({\bm w}_{\mathrm{IPW}}) - \mu_f\right)\dto \mathcal N(0, \sigma^2_{\mathrm{IPW}}),
\end{align*}
where $\sigma^2_{\mathrm{IPW}} = \Var\left[{\bm g_c^{(\bm A_c)}}^\top\bm w_{\mathrm{IPW},c} \right] + \sigma^2 \mathbb E\left[\left\|\bm w_{\mathrm{IPW},c} \right\|^2\right]$.
\end{theorem}
Theorem~\ref{thm:ipw_normality} shows that the asymptotic variance of the IPW estimator equals a sum of two terms: the first term captures the variation of the IPW-weighted outcome regression function, while the second term is related to the squared norm of the IPW weights. Interestingly, the signal strength only affects the first term, while the noise variance $\sigma^2$ only affects the second term.

Next, we establish the asymptotic normality of our balancing estimator under the linear model assumption described in Section~\ref{sec:lm_balancing}, ensuring that the length of $\bm h$ does not vary with $n$. When this condition does not hold, even though one cannot obtain balancing estimators, one can modify IPW estimators to leverage the low-rank structure so long as the true propensity score is known. We present an asymptotic analysis of this case in Appendix~\ref{sec:varying_asymp}. 

In this section, we extend our analyses by allowing for the misspecification of the low-rank structure where we denote the true low-rank structure by $\boldsymbol \Lambda^\ast$ and the structure used to derive the balancing weights by $\bm \Lambda$. If $\bm \Lambda = \bm \Lambda^*$, then the low-rank structure is \textit{exactly specified}, if the column-space of $\bm \Lambda^*$ is contained in that of $\bm \Lambda$, it is \textit{overspecified}, and otherwise, the low-rank structure is \textit{misspecified}. We call a low-rank structure \textit{correctly specified}, if it is not misspecified. Next, we state some required assumptions.


\begin{ass}[The balancing equations are feasible]
\label{assm:balancing_feasibility}
For the event $B_{\bm \Lambda}$ defined in Theorem~\ref{thm:umvue}, we have $\mathbb P(B_{\bm \Lambda})\to 1$ as the number of clusters $n\to \infty$.
\end{ass}
Before stating the next assumption, we introduce some notation. Consider the low-rank structure assumption $\bm g = \bm \Lambda \bm h$. Given a treatment pattern $\bm a_c\in \{0,1\}^{M_c}$ for cluster $c$, we define $\bm \Lambda_c(\bm a_c)$ as the sub-matrix of $\bm \Lambda$ such that
\[
    \bm \Lambda_c(\bm a_c)\bm h = (g_{c1}^{(\bm a_c)}(\bm X_c), \cdots, g_{cM_c}^{(\bm a_c)}(\bm X_c))^\top.
\]
\begin{ass}[Regularity conditions under a fixed low-rank structure] 
\label{assm:clt_technical_fixed} 
Let the following regularity conditions hold.
\begin{enumerate}[label=(\alph*)]
    \item $\mathbb E[\bm \Lambda_c(\bm A_c)^\top\bm \Lambda_c(\bm A_c)]$ is invertible.
    \item $\mathbb E\left[\sigma^4_{\max}\left(\bm \Lambda_c(\bm A_c)\right)\right]<\infty$ and $\mathbb E\left[\sigma^2_{\max}(\bm \Lambda_c(\bm A_c)^\top\bm \Lambda^\ast_c(\bm A_c))\right]$, where $\sigma_{\max}(\bm M)$ denotes the maximum singular value of a matrix $\bm M$.
    \item $\mathbb E\left[M_c^2\right]<\infty$.
    \end{enumerate}
\end{ass}

Assumption~\ref{assm:balancing_feasibility} requires the asymptotic feasibility of the balancing estimator, while Assumption~\ref{assm:clt_technical_fixed} presents technical conditions required for establishing asymptotic normality. In Appendix~\ref{sec:balancing_feasibility}, we show that under mild restrictions on the covariate distribution, Assumption~\ref{assm:balancing_feasibility} holds for a rich class of low-rank structures (and in particular for all the interference patterns discussed in Examples~\ref{ex:1}--\ref{ex:4} under the linear model presented in Example~\ref{ex:5}).  

We now establish the asymptotic normality of the balancing estimator under the potential misspecification of the low-rank structure.
\begin{theorem}[Asymptotic normality under potential misspecification]
\label{thm:clt_fixed}
Let the potential outcome regression functions satisfy $\bm g = \bm \Lambda^\ast \bm h$, and consider the linear model assumption. 
Then, under Assumptions~\ref{assm:partial}--\ref{assm:outcome_model},~\ref{assm:iid_error}, \ref{assm:balancing_feasibility}, \ref{assm:clt_technical_fixed}, as well as assuming $\mathbb E\left[\|\bm w_{{\rm IPW},1}\|^4\right]<\infty$, we have that
\begin{align*} &\sqrt{n}\left(T(\wnulllamb) - \mu_f(\bm \Lambda, \bm \Lambda^\ast)\right)\dto \mathcal N\left(0, \sigma_{\sthree}^2(\bm \Lambda, \bm \Lambda^\ast)\right),
\end{align*}
where $\mu_f(\bm \Lambda, \bm \Lambda^\ast)$ and $\sigma^2_{\sthree}(\bm \Lambda, \bm \Lambda^\ast)$ are defined as follows: First, define $\bm h_{\bm \Lambda^\ast} := \bm V^{-1}\mathbb E\left[\bm \Lambda_c(\bm A_c)^\top\bm g_c^{(\bm A_c)}\right]$, where $\bm V = \mathbb E[\bm \Lambda_c(\bm A_c)^\top\bm \Lambda_c(\bm A_c)]$, and the theoretical projection operator ${\mathcal P}$ as ${\mathcal P}\bm l:= \bm \Lambda_c(\bm A_c)\bm V^{-1}\mathbb E[\bm \Lambda_c(\bm A_c)^\top\bm l]$ for any vector $\bm l$. Then 
$$
\begin{aligned}
 \mu_f(\bm \Lambda, \bm \Lambda^\ast) & = \mathbb E\left[\bm w_{\mathrm{IPW},c}^\top {\mathcal P}\bm g_c^{(\bm A_c)}\right] \\
    \sigma^2_{\sthree}(\bm \Lambda, \bm \Lambda^\ast) & = \mathrm{Var}\left[\left((I - {\mathcal P})\bm g_c^{(\bm A_c)}\right)^\top\left({\mathcal P}\bm w_{\mathrm{IPW},c}\right) + \mathbb E\left[\bm w_{\mathrm{IPW},c}^\top {\mathcal P}\bm g_c^{(\bm A_c)}\mid \bm X_c\right]\right] +  \sigma^2\mathbb E\left[\left\| {\mathcal P}\bm w_{\mathrm{IPW},c}\right\|^2\right].
    \end{aligned}
$$
\end{theorem}
We provide intuition about the asymptotic mean of our estimators, $\mu_f(\bm \Lambda, \bm \Lambda^\ast)$. This quantity represents the correlation between the true IPW weights, $\bm w_{\mathrm{IPW},c}$, and the part of the true low-rank structure in the potential outcomes captured by the assumed low-rank structure, ${\mathcal P}\bm g_c^{(\bm A_c)}$. If the low-rank structure is misspecified, this term will reflect the bias, which is equal to the part of the outcome that is not captured by the assumed low-rank structure and is correlated with the treatment assignment. This indicates that the misspecification bias is driven by the ``residual'' part of the outcome that is correlated with the treatment assignment. The following corollary, which follows immediately from Theorem~\ref{thm:clt_fixed}, summarizes the results under correct specification of the low-rank structure.
\begin{corollary}[Asymptotic normality under correct specification]
\label{cor:correct_specification_fixed}
Consider the setting of Theorem~\ref{thm:clt_fixed}. If the low-rank structure is either exactly specified or over-specified, that is, $\mathrm{col.sp}(\bm \Lambda)\geq\mathrm{col.sp}(\bm \Lambda^\ast)$, where $\mathrm{col.sp}(\bm M)$ denotes the column-space of a matrix $\bm M$, then we have:
\begin{enumerate}[leftmargin=*]
    \item $\mu_f(\bm \Lambda, \bm \Lambda) = \mu_f$, so that the balancing estimator is asymptotically unbiased,
    \item $\sigma^2_{\sthree}(\bm \Lambda, \bm \Lambda) = \mathrm{Var}\left[\mathbb E[{\bm g_c^{(\bm A_c)}}^\top\bm w_{\mathrm{IPW},c}\mid \bm X_c]\right] + \sigma^2\mathbb E\left[\left\| {\mathcal P}\bm w_{\mathrm{IPW},c}\right\|^2\right]$.
    \item $\sigma^2_{\sthree}(\bm \Lambda, \bm \Lambda)\leq \sigma^2_{\mathrm{IPW}}$, so that the balancing estimator is at least as efficient as the IPW estimator.
    \end{enumerate}
\end{corollary}

Corollary~\ref{cor:correct_specification_fixed} fills the missing gaps in the finite sample analysis of Theorem~\ref{thm:umvue}. The second result of Corollary~\ref{cor:correct_specification_fixed} similarly establishes that the asymptotic variance of the balancing estimator $\sigma^2_{\rm bal}$ can be decomposed into two terms, where the first term does not depend on the chosen low-rank structure. This term is provably smaller than the corresponding first term of $\sigma^2_{\rm IPW}$ due to the following variance decomposition, 
$$
\mathrm{Var}\left[{\bm g_c^{(\bm A_c)}}^\top\bm w_{\mathrm{IPW},c}\right] = \mathrm{Var}\left[ \mathbb E\left[{\bm g_c^{(\bm A_c)}}^\top\bm w_{\mathrm{IPW},c}\mid \bm X_c\right]\right] + \mathbb E\left[ \mathrm{Var}\left[{\bm g_c^{(\bm A_c)}}^\top\bm w_{\mathrm{IPW},c}\mid \bm X_c\right]\right].
$$ 
This efficiency gain of the balancing estimator stems from the fact that the balancing equation achieves balance in each given sample, as opposed to the IPW estimator, which balances \textit{in expectation} \citep{Zubizarreta01112011}.

The second component of $\sigma^2_{\sthree}$ is also provably smaller than the corresponding component of $\sigma^2_{\mathrm{IPW}}$, i.e., $\sigma^2\mathbb E\left[\|{\mathcal P} \bm w_{\mathrm{IPW,c}}\|^2\right] \le \sigma^2\mathbb E\left[\|\bm w_{\mathrm{IPW,c}}\|^2\right]$. This is exactly where the balancing estimator exploits the interference structure for efficiency gain, and in particular, a more informative low-rank structure yields a lower rank for ${\mathcal P}$, leading to a greater gain. Thus, although over-specification of the low-rank structure would still yield unbiased estimates, the second term in the asymptotic variance is larger for oversimplified models, indicating that a parsimonious yet correct model should be preferred. The impact of the reduction in this second component is larger when the signal-to-noise ratio is small.

In summary, we find that the balancing estimator, despite not requiring knowledge of the true propensity score, can be substantially more efficient than the IPW estimator by satisfying the balancing equation and exploiting the assumed low-rank structure.

Finally, to conduct inference on $\mu_f$ (or $\mu_f(\bm \Lambda, \bm \Lambda^\ast)$) using an asymptotic confidence interval, we need a consistent estimator for the asymptotic variance of the balancing estimator. Following the strategy used in \cite{chanetal2016,additivefactorial}, the next theorem derives such an estimator by using the consistent sandwich variance estimator for the $Z$-estimator. The exact expression of this variance estimator is given in Equation~\eqref{eqn:var_est_exprn} of the Appendix.
\begin{theorem}[Consistent variance estimator]
\label{thm:var_estimation}
Let $\hat\sigma_{\sthree}^2$ denote the variance estimator as described in Equation~\eqref{eqn:var_est_exprn} of the Appendix. Then, under the setting and the assumptions in Theorem \ref{thm:clt_fixed}, we have $\hat \sigma^{2}_{\sthree}\Pto \sigma^2_{\sthree}(\bm \Lambda, \bm \Lambda^\ast)$, as $n\rightarrow \infty$. Hence, for any $\alpha \in (0,1)$, the confidence interval given by
\[
    \hat C^{(1-\alpha)}_{\bm \Lambda} = \left(T\left(\wnulllamb\right)\pm z_{\alpha/2}\frac{\hat \sigma_{\sthree}}{\sqrt{n}}\right),
\]
satisfies, $\underset{n\to \infty}{\lim\inf}~\mathbb P\left(\mu_f(\bm \Lambda, \bm \Lambda^*)\in  \hat C^{(1-\alpha)}_{\bm \Lambda}\right)\geq 1-\alpha$. Here $z_{\alpha/2}$ satisfies $\mathbb P(Z\geq z_{\alpha/2})=\alpha/2$ for $Z\sim \mathcal N(0,1)$.
\end{theorem}

\subsection{Choosing a low-rank structure}
\label{sec:choose_structure}

Our methods developed so far rely on the specification of a correct low-rank structure for unbiased inference, which might be difficult to assume in practice. In this section, we present a data adaptive method for choosing a low-rank structure among candidate structures.

Suppose that we have a collection of low-rank structures $\bm \Lambda_1,\cdots, \bm \Lambda_L$ such that for $1\leq i\leq j\leq L$, $\mathrm{col.sp.}(\bm \Lambda_i)\leq \mathrm{col.sp.}(\boldsymbol \Lambda_j)$, implying that $\bm \Lambda_i$ is a more informative structure than $\bm \Lambda_j$. Assume also that one of these structures is correctly specified. Since overspecification of the low-rank structure still results in correct specification, the least restrictive structure $\bm \Lambda_L$ is guaranteed to be correctly specified under this assumption. However, recovering the most informative low-rank structure from this available class (defined to be the correct structure $\bm \Lambda_l$ for the smallest possible $l$) can result in more efficient inference. 

The next theorem presents a consistent test of whether or not a particular structure is correctly specified.  The theorem also shows how to consistently recover the most informative correct structure. As before, we focus on the linear model assumption of Section~\ref{sec:lm_balancing}, while Appendix~\ref{sec:varying_asymp} considers the case where this assumption does not hold.

\begin{theorem}[Consistent test of low-rank structure]
\label{thm:decide_true_structure}
    Suppose that Assumptions~\ref{assm:partial}--\ref{assm:outcome_model} hold under the linear model assumption discussed in Section~\ref{sec:lm_balancing}. 
    In addition, suppose that Assumption~\ref{assm:clt_technical_fixed} is satisfied for every hypothesized structure $\bm \Lambda_l,\forall l$. Also, suppose that $L$ is fixed and does not vary with $n$. Consider the hypothesis $H_{l}:\bm g \in \mathrm{col.sp.}\left(\bm \Lambda_l\right)\textrm{ versus }H_L:\bm g \in \mathrm{col.sp.}\left(\bm \Lambda_L\right)$ (note that the true low-rank structure is the alternate hypothesis while the hypothesized structure is the null).  Define,
    \[
        S_{lL} := \left(\frac{\left(\hat{\bm w}^{\sthree}_{\bm \Lambda_l} - \hat{\bm w}^{\sthree}_{\bm \Lambda_L}\right)^\top\bm y}{\hat\sigma \left\|\hat{\bm w}^{\sthree}_{\bm \Lambda_l} - \hat{\bm w}^{\sthree}_{\bm \Lambda_L}\right\|}\right)^2.
    \]Then, we have the following results:
    \begin{enumerate}[label = (\alph*)]
        \item If $H_{l}$ holds, we have  $S_{lL} \dto \chi^2_1$.
        \item If $H_{l}$ does not hold and $\mu_f(\bm \Lambda_l,\bm \Lambda_L) \neq \mu_f(\bm \Lambda_L, \bm \Lambda_L)$, then, $ S_{lL}\Pto\infty$.
        \item For any given $\alpha\in (0,1)$, let $\chi^2_{1;\alpha}$ denote the $(1-\alpha)$-quantile of the $\chi_1^2$ distribution. Furthermore, let 
        $l^*:=\arg\min\left\{l: \bm \Lambda_l\textrm{ is a true low-rank structure}\right\}$, and $\hat l = \arg\min\left\{l:S_{lL}<\chi^2_{1;\alpha}\right\}$. Then, we have that $\mathbb P\left(\bm \Lambda_{\hat l}\textrm{ is a true low-rank structure}\right)\rightarrow 1$, as $n\rightarrow \infty$, and $\underset{n\to \infty}{\lim\inf}~\mathbb P\left(\hat l = l^*\right)\geq 1-\alpha$.
    \end{enumerate}
\end{theorem}
According to the theorem, $\hat l$ asymptotically recovers the correct, most restricted low-rank structure among the $L$ low-rank structures considered, with probability at least $1-\alpha$. To obtain a consistent estimator $\hat \sigma$ of $\sigma$, a standard choice can be the linear regression-based estimator using $\bm \Lambda_L$, i.e., $\sqrt{\frac{\|(\bm I - \bm P_{\bm R\bm \Lambda_L})\bm y\|^2}{n-\mathrm{rank}(\bm R\bm \Lambda_L)}}$.

Note that even though Theorem~\ref{thm:decide_true_structure} shows how to consistently choose a correct low-rank structure indexed by $\hat l$, constructing the confidence interval $\hat C^{(1-\alpha)}_{\bm \Lambda_{\hat l}}$ using the same data may lead to under-coverage due to the post-selection bias arising from the data-dependent choice of $\hat l$. However, our simulation experiments show that this bias may be of small magnitude, as this confidence interval exhibits robust empirical coverage (see Section~\ref{sec:experiments_simulated}). 


Finally, we note that our asymptotic analysis and methods makes the assumption of homoskedastic error (i.e., Assumption~\ref{assm:iid_error}). However, the theory presented in this section can also accommodate heteroskedasticity, where we assume that $\Var\left(\epsilon_{ci}^{(\bm A_c)}\right) = \sigma^2_{ci}$, and the $\sigma^2_{ci}$'s are allowed to vary so long as $\left\{\frac{1}{n}\sum_{c=1}^n\sum_{i=1}^{M_c}\sigma_{ci}^2\right\}_{n\geq 1}$ is a convergent sequence. A similar result appears in \cite{additivefactorial}, and the results in this section can be extended similarly.

\section{Leveraging knowledge of the true propensity score}
\label{sec:propensity_extensions}

So far, we have focused on developing balancing estimators under interference without knowing the propensity scores.  However, when the propensity score is known, our methodology can be modified to improve upon IPW-type estimation by leveraging a low-rank structure. We summarize some key results while leaving the details to Appendix~\ref{sec:extension_known_propensity}.

\subsection{Projection estimators}

In this setting, we construct an optimal weighting estimator in the way similar to 
Theorem~\ref{thm:subclass2}. Specifically, we choose a weight from $\Cipwlamb(e)=\{\bm w \in \mathcal C_{\bm \Lambda}(\bm h, e),\forall \bm h\}$ while keeping $e$ fixed. In this known propensity score case, we do not require the linear model assumption to ensure non-triviality of $\Cipwlamb(e)$. However, this assumption is needed to obtain an estimator that is more efficient than the IPW estimator. Given the low-rank structure assumption $\bm g = \bm \Lambda \bm h$, under the linear model assumption with fixed $\bm h$, we obtain the following class of weights,
\[
    \Cipwprojlamb(e) := \{\bm w: \bm \Lambda^\top\bm R^\top(\bm w - {\bm w}_{\mathrm{IPW}}) = \bm 0\}\subseteq \Cipwlamb(e).
\]
We denote optimal weights within this class by $\wipwprojlamb(e)$ and refer to the resulting estimator as the \textit{projection estimator}. The projection estimator is, according to Theorem~\ref{thm:propensity_estimators_fixed_h}, provably at least as efficient as the IPW estimator, for any fixed number of clusters. 

Note that $\bm w_{\rm IPW}$ always belongs to $\Cipwprojlamb(e)$. In fact, similar to the balancing estimator, one also needs a linear model assumption that ensures that the matrix $\bm R\bm \Lambda$ has more rows than columns to ensure that $\Cipwprojlamb(e)$ contains more than just the IPW weights. Otherwise, choosing optimal weights directly from $\Cipw(e)$ 
yields the following set of weights, which always exist non-trivially (without making the linear model assumption),
\begin{align*}
    (\wipwlamb(e))_{ci} = \bm R^\top_{c}\bm D_{c}^{-1/2}\bm P_{\bm D_{c}^{1/2}\bm \Lambda_{ci}}\bm D_{c}^{1/2}\tilde{\bm w}_{\mathrm{IPW},c,i},
    \end{align*}
where $\bm R_c = \left(\mathbb I\left(\bm A_c = \bm a^{(0)}_c\right), \dots, \mathbb I\left(\bm A = \bm a^{(2^{M_c}-1)}_c\right)\right)^\top$, $\bm D_c = \mathbb E\left[\bm R_c^\top\bm R_c\middle| \bm X_c\right]$ and $\tilde{\bm w}_{\mathrm{IPW},c,i}$ is formed by stacking the IPW weights vector across all the observable treatment patterns $\bm a \in \{0,1\}^{M_c}$ for the $c^{\rm th}$ cluster. We call $T\left(\wipwlamb(e)\right)$ the \textit{weighted projection estimator}.

Unlike the balancing and the projection estimator, Corollaries~\ref{cor:correct_specification_indiv}~and~\ref{cor:propensity_correct_specification_fixed} in the appendix show that the weighted projection estimator does not necessarily has a smaller variance than the IPW estimator. Nevertheless, when the signal-to-noise ratio is not too high, the weighted projection estimator is expected to be more efficient than the IPW estimator.

\subsection{Asymptotic efficiency comparison}

Under the linear model assumption with fixed $\bm h$, all of the proposed low-rank weighting estimators exist.  Therefore, we briefly summarize the efficiency comparison in this setting, while leaving the details to Appendix~\ref{sec:extension_known_propensity}.  The asymptotic variances of all proposed estimators are given by, 
\begin{align*}
    \sigma^2_{\mathrm{IPW}} &= \Var \left(\bm w_{\mathrm{IPW},c}^\top\bm g_c^{(\bm A_c)}\right)  + \sigma^2 \mathbb E\left[\|\bm w_{\mathrm{IPW}}\|^2\right]\\
    \sigma^2_{\sone} &= \mathrm{Var}\left(\mathbb E\left[\bm w_{\mathrm{IPW},c}^\top {\bm g_c^{(\bm A_c)}}\mid\bm X_c\right]\right) + \mathbb E\left(\mathrm{Var}\left[ \left({\mathcal P}\bm w_{\mathrm{IPW},c}\right)^\top {\bm g_c^{(\bm A_c)}}  \mid \bm X_c\right]\right) + \sigma^2\mathbb E\left[\left\| {\mathcal P}\bm w_{\mathrm{IPW},c}\right\|^2\right]\\
    \sigma^2_{\stwo} &= \Var \left(\bm w_{\mathrm{IPW},c}^\top\bm g_c^{(\bm A_c)}\right)+ \sigma^2\mathbb E\left[\left\| {\mathcal P}\bm w_{\mathrm{IPW},c}\right\|^2\right]\\
    \sigma^2_{\sthree} &= \mathrm{Var}\left(\mathbb E[\bm w_{\mathrm{IPW},c}^\top\bm g_c^{(\bm A_c)}\mid \bm X_c]\right) + \sigma^2\mathbb E\left[\left\| {\mathcal P}\bm w_{\mathrm{IPW},c}\right\|^2\right],
\end{align*}
where ${\mathcal P}$ is defined in Theorem~\ref{thm:clt_fixed}.  Then, we can immediately obtain the following relations: 
$$\sigma^2_{\sthree}\leq \sigma^2_{\stwo}\leq \sigma^2_{\mathrm{IPW}}, \quad \text{and} \quad \sigma^2_{\sthree}\leq \sigma^2_{\sone}.$$
This result implies that the balancing estimator is asymptotically more efficient than all the other weighting estimators that use knowledge of the true propensity score. As discussed in Section~\ref{sec:asymp_normality}, this primarily stems from the variance reduction the balancing estimator achieves by satisfying the balancing equations, which none of the other estimators do. For example, the projection estimator does not satisfy the balancing equations, and hence, the first term of $\sigma^2_{\rm proj}$ is the same as that of $\sigma^2_{\rm IPW}$. However, because it still exploits the low-rank structure, the second term of $\sigma^2_{\rm proj}$ and $\sigma^2_{\rm bal}$ are identical. Thus, regardless of knowledge of the true propensity score, under Assumption~\ref{assm:balancing_feasibility}, the balancing estimator has the smallest asymptotic variance and remains our recommended choice for causal effect estimation under appreciable sample size.

We note that even if the $\bm h$ is not fixed across $n$, the weighted projection estimator still exists and can leverage the low-rank structure as well as the known propensity score. The variance of the weighted projection estimator has an expression similar to that of $\sigma^2_{\sone}$ in this case as well. See Appendix~\ref{thm:propensity_estimators_fixed_h} for details. Furthermore, we show in Appendix~\ref{sec:weighted_generalization} that the weighted projection estimator generalizes many existing IPW-type estimators from the literature that assume some specific interference patterns.

\section{Simulation studies}
\label{sec:experiments_simulated}

In this section, we examine the finite sample performance of the proposed methodology through simulation studies.

\subsection{Setup}

We consider the setting of Section \ref{sec:asymptotics}, in which
$\bm h$ is fixed and does not depend on the sample size. In our simulation, we set $n \in \{100, 300, 500, 700\}$.  For each cluster $c$, we independently draw $M_c \in \{10, 15\}$, each with the probability 1/2. For each unit $i$ in the $c^{\mathrm{th}}$ cluster, we generate a four-dimensional (i.e., $p=4$) covariate $\bm X_{ci} = (X_{ci,1},  X_{ci,2},  X_{ci,3}, X_{ci,4})^\top$ as follows: First we generate $\bm Z_{ci} = (Z_{ci,1},  Z_{ci,2},  Z_{ci,3}, Z_{ci,4})^\top \iid \mathcal N_4\left(\bm 0, \bm \Sigma(\rho)\right)$, drawn independently across all the units, where $\bm \Sigma$ is a Toeplitz structure with correlation $\rho$. We then normalize each covariate as $X_{ci,j} = \frac{Z_{ci,j}}{\sqrt{\sum_{i^\prime=1}^{M_c}Z^2_{ci^\prime,j}}}$, for all $1\leq c\leq n, 1\leq i\leq M_c, 1\leq j \leq p$.

To define the propensity score and the counterfactual weights, we consider the class of probabilities 
$\pi_{ci}(\bm X_c;\kappa) = \Phi\left(\frac{1}{\sqrt{p}}\sum_{j=1}^p \bar X_{c\cdot, j} + \kappa\bar X_{ci,\cdot}\right)$,
parametrized by the \textit{counterfactual deviation}$, \kappa$. Here, $\bar X_{c\cdot, j}$ and $\bar X_{ci,\cdot}$ represent the average value of the $j^\text{th}$ covariate in cluster $c$, and the average value of a unit's four covariates, respectively. We define the propensity score for the $i^{\mathrm{th}}$ individual in the $c^{\mathrm{th}}$ cluster as $e_{ci}(\bm X_c) = \pi_{ci}(\bm X_c;0)$, and treatments follow $A_{ci}\mid \bm X_c\sim \mathrm{Bernoulli}( \pi_{ci}(\bm X_c; 0))$, independently across $i$.
Our counterfactual weights are based on a stochastic interventions from this same class. 

Next, we assume the signal-and-noise model of Equation~\eqref{eqn:outcome_model} with $g^{(\bm a_c)}_{ci}(\bm X_c) = \sum_{j=1}^3\beta_{ci}^{(j)}(\bm a_c)x_{ci,j} + \beta_{ci}^{(4)}(\bm a_c)\bar{x}_{c\cdot,4}$, where $\bm \beta_{ci}(\bm a_c) = \left(\beta_{ci}^{(1)}(\bm a_c),\cdots, \beta_{ci}^{(4)}(\bm a_c)\right)^\top$ and $\bm x_{ci} = (x_{c,i,1},\cdots, x_{c,i,4})^\top$. Note that $\bar{x}_{c\cdot, 4}$ is a cluster-level covariate shared by the potential outcome of all units in a cluster. For the low-rank structure, we assume nearest neighbor interference (where the nearest neighbor graph is computed based on the distance between the covariates $\bm X_{ci}$) restricted to the first five nearest neighbors. Hence, we generate $\bm \Lambda_{ci}$ as described in Example~\ref{ex:6} of the appendix and set $\bm \beta_{ci}{(\bm a_c)} = \bm \Lambda_{ci}(\bm a_c)\bm h$. Appendix~\ref{sec:further_simulations_fixed} presents additional simulation studies under stratified and additive interference settings described in Examples~\ref{ex:1}~and~\ref{ex:3}, respectively. 

Lastly, we set
$\bm h = \gamma_{\mathrm{nn}}\times (1,\cdots, 32)^\top\otimes \bm 1_4$,
for a constant $\gamma_{\mathrm{nn}}$ that is calibrated based on the signal-to-noise ratio (SNR), which is defined as $\mathrm{Var}\left[{\bm g_c^{(\bm A_c)}}^\top\bm w_{\mathrm{IPW,c}}\right]\Big /\sigma^2 \mathbb E\left[\|\bm w_{\mathrm{IPW},c}\|^2\right]$. This particular definition of SNR is motivated by the expression of the variance of the IPW estimator derived in Theorem~\ref{thm:clt_fixed}. In that expression, the term in the numerator exactly equals the part of $\sigma^2_{\mathrm{IPW}}$ affected by the signal component of the potential outcomes alone, while the one in the denominator is affected by the noise (and is also the part reduced by making a correct low-rank assumption, as the discussion after Corollary~\ref{cor:correct_specification_fixed} suggests).

\subsection{Results}

We compare the performance of the proposed balancing estimators with the standard IPW estimator with true or estimated propensity score, where, as commonly done in practice, a mixed-effects probit regression is used to estimate the propensity score \citep{kang2023propensity}.  In addition, for the remainder of this section, the term ``IPW estimator'' without any qualifier would by default indicate the IPW estimator with the true propensity score.

We focus on a $k$-nearest neighbor interference structure with $k=3$, where the nearest neighbor graph within a cluster is formed based on the distances between the covariates of units.  We consider two cases: (1) we know the true low-rank structure of $k=3$, and evaluate the balancing estimator, and (2) we use the data-adaptive estimator described in Section~\ref{sec:choose_structure} to choose a low-rank structure among the five candidate low-rank structures $\{\bm \Lambda_l\}_{1\leq l\leq 5}$ representing $l$-nearest neighbor interference. Under each case, we construct 95\% confidence intervals based on the respective balancing estimator as well as both the IPW estimators discussed in the previous paragraph. Recall that the confidence interval based on the latter estimator does not have a theoretical coverage guarantee.

\begin{figure}[t!]
    \centering
    \includegraphics{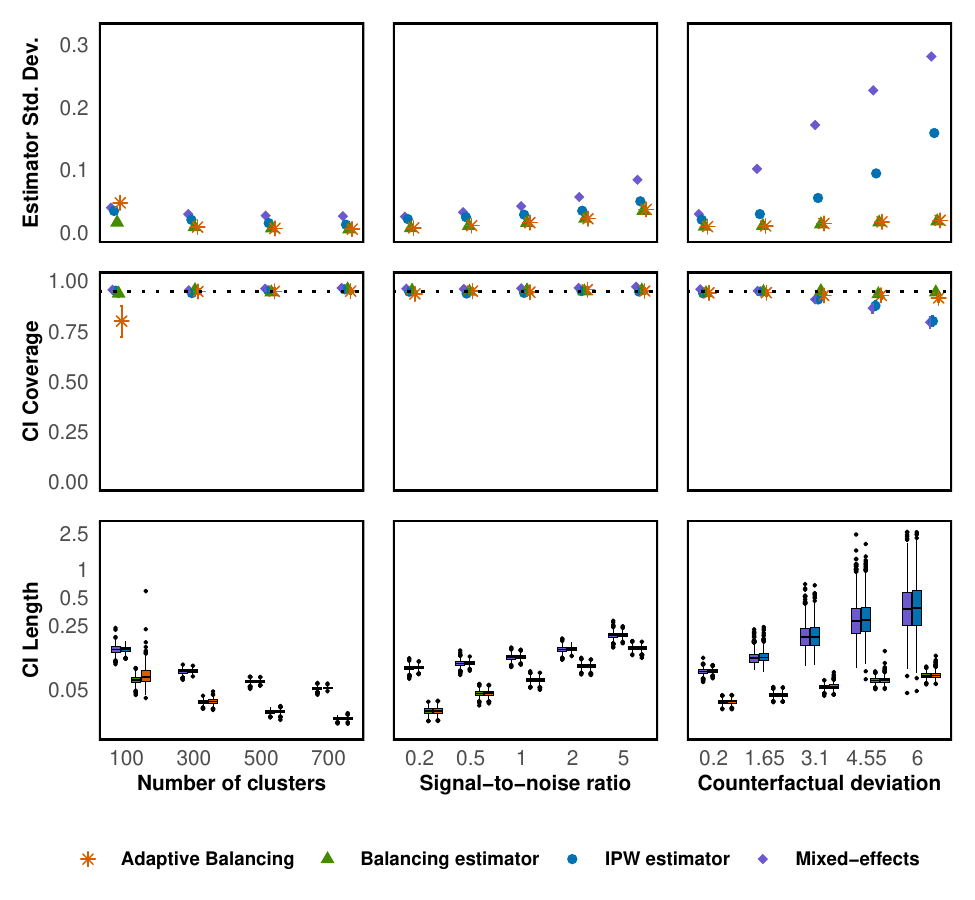}
    \caption{Empirical Performance of the Estimators and their 95\% Confidence Intervals. The three rows from the top depict the standard deviation of the estimators, empirical CI coverages (with its plus or minus two standard errors), and empirical CI length.  The results are based on 1,000 Monte Carlo simulations.  In the left column, we fix the deviation of the counterfactual policy from the realized one $\kappa = 0.2$ and signal-to-noise ratio $\textrm{SNR} = 0.2$, while varying the number of clusters $n$. In the middle column, we set $n = 300$ and $\kappa = 0.2$ while varying $\textrm{SNR}$. In the right column, we fix $n = 300$ and $\textrm{SNR} = 0.2$, while varying $\kappa$. For all simulations, we set the variance $\sigma^2 = 1$. The plots in the bottom row plots the log-transformed CI lengths though the y-axis label is in the original scale for interpretability.
    }
    \label{fig:fixed_ci_main}
\end{figure}

Figure~\ref{fig:fixed_ci_main} summarizes the results, with each point in any of the plots obtained by averaging across 1,000 independent replicates of the corresponding simulation setting. For the balancing estimator with knowledge of the true structure, we average over the cases where the balancing equation is feasible. For the data-adaptive balancing estimator, we average over the cases in which the balancing equation corresponding to all the five low-rank structures is feasible. For all the simulation settings considered here, the balancing equation using the true structure was always feasible. 
In contrast, we were able to compute its data-adaptive variant for 100$\%$ of the instances in all cases except for the following settings: for the left column, we could compute it only for 10$\%$ of the instances when $n=100$; for all the settings in the middle column and $\kappa=6$ in the right column, we were able to compute it in more than 99$\%$ of the instances.

The first row of the figure estimates the standard deviation of each estimator. We find that across different settings, the balancing estimator and its data-adaptive variant exhibit a significantly smaller variability than the IPW estimator with the true propensity score or the estimated propensity score based on the mixed effects model. \color{black} The percentage decrease in standard deviation of the balancing estimator over the IPW estimator varies from about 30$\%$ (corresponding to the setting with the highest SNR in the center panel) to approximately $88\%$ (corresponding to $\kappa=6$ in the rightmost panel).

Next, we turn to the performance of the confidence intervals (CIs). The leftmost panels of the figure present a setting with a low signal-to-noise ratio of 20$\%$. On average, the CI of the balancing estimator is about $53\%$ shorter than that of the standard IPW estimator. We also find that in many cases the adaptive balancing estimator produces CIs whose lengths are similar to those of the balancing estimator with knowledge about the true low-rank structure. However, for the setting with the smallest number of clusters (i.e., $n=100$), the CI of the adaptive balancing estimator leads to undercoverage. Although the data-adaptive-balancing CI does not have a theoretical guarantee, it exhibits robust empirical coverage for all the other settings. 

The center panels show that as the signal-to-noise ratio (SNR) increases, the lengths of all the CIs start to grow.  However, the length of the IPW-based CI remains substantially higher than the length of the CIs of the balancing estimators. This is consistent with our theoretical result that the asymptotic variance of the balancing estimator benefits more from exploiting the low-rank structure in the low SNR regime, as outlined in Section~\ref{sec:asymp_normality}. In the smallest SNR regime, the CI of the balancing estimator is about $66\%$ shorter than that of the IPW estimator.  As before, the performance of the adaptive balancing estimator is similar to that of the non-adaptive estimator. These findings indicate the superior performance of the balancing estimators and their CIs. 

Finally, the rightmost panels, which summarize the results with a varying value of $\kappa$, show that for a fixed sample size of $n=300$, the empirical coverage of the IPW estimator-based CI using both the true and the estimated propensity score starts to deteriorate as $\kappa$ becomes more extreme (that is, the counterfactual intervention diverges further from the propensity score). In contrast, the CIs of the balancing estimators achieve better coverage in part because they have a much smaller variance. We also find that the performance of the IPW-based CI with the estimated propensity score is generally similar to that of the IPW-based CI with the true propensity score in all of our simulation settings.

Appendix~\ref{sec:further_simulations_fixed_ci} reports additional simulation results, where we also compare the estimators described in Section~\ref{sec:propensity_extensions} that leverage a low-rank structure along with the knowledge of the true propensity score.  This appendix section also presents the results under stratified and additive interference, showing that our main conclusions continue to hold. Additionally, the results shown in Appendix~\ref{sec:further_simulations_fixed_bias} also empirically demonstrate that the proposed estimators have negligible bias in small samples.

\section{Application to the diffusion of microfinance study}
\label{sec:experiments_real}

In this section, we apply our balancing weights methodology to analyze the diffusion of microfinance dataset \citep{banerjee2013}.

\subsection{Data}

The dataset is based on a non-randomized experiment that was conducted across 43 villages in the state of Karnataka in India. The goal of the study is to assess the spread of information about a microfinance program and eventual participation. In each village, some households are referred to as ``leaders'' if their members hold important societal roles.  A primary hypothesis of the study is that these leaders are more capable of spreading information to others.  Motivated by this hypothesis, the authors informed these leaders about the microfinance program at the beginning of the study. At various points in time during the study, the information about each household's participation in the program was collected. 

Since the selection of leaders is not random, this study represents an example of an observational study. Our unit of analysis is the household, and the treatment indicator variable denotes whether or not a household was a leader and informed about the microfinance program at the beginning of the study. 
The outcome variable of interest is the household participation in the microfinance program at the end of the study.  The household-level covariates include type of roof, number of rooms, number of beds, religion, type of electricity, and type of latrine.

\cite{banerjee2013} acknowledge the presence of interference and summarize this interference pattern as a graph on the households for each of the villages. In particular, the graph draws an edge between two households if their members interact with each other. Such interactions include visiting, borrowing/lending money and material goods, and receiving/giving advice (see \cite{banerjee2013} for details). Furthermore, the authors state that due to geographical distances, interference is likely to be restricted within each village, suggesting that the partial interference assumption may hold. 



\subsection{Analysis setup}
\label{sec:real_data_analysis_setup}

We apply our proposed methodology to estimate the causal effects based on the following several different counterfactual weights:
\begin{enumerate}[label = (\alph*)]
    \item {\bf Uniformly random selection:} As a baseline comparison, we consider the uniformly random selection of households to be informed, while keeping the number of selected households within each village identical to the number of leaders in the study $L_c$. Formally, this counterfactual policy is given by:
    \[
        f_{\text{random}}(\bm a_c,\bm X_c) = \frac{1}{\binom{M_c}{L_c}}\mathbb I\left(\bm 1^\top\bm a_c = L_c\right),\forall \bm a_c\in \{0,1\}^{M_c}.
    \]

    \item {\bf Most connected households:} We also consider assigning the treatment to $L_c$ households in each cluster that have the greatest number of connected edges. If we use $\mathcal P_c$ to denote the set of $L_c$ most connected households, then this counterfactual policy is given by:
    \[
        f_{\text{most}}(\bm a_c,\bm X_c) =\mathbb I\left(\{1\leq i\leq M_c:a_{ci}=1\} = \mathcal P_c\right), \forall \bm a_c \in \{0,1\}^{M_c}.
    \]
    \item {\bf Least connected households:} Lastly, we consider treating $L_c$ households in each cluster that have the least number of connected edges. If we use $\mathcal Q_c$ to denote the set of $L_c$ least connected households, then this counterfactual policy is given by:
    \[
        f_{\text{least}}(\bm a_c,\bm X_c) =\mathbb I\left(\{1\leq i\leq M_c:a_{ci}=1\} = \mathcal Q_c\right), \forall \bm a_c \in \{0,1\}^{M_c}.
    \]
    
\end{enumerate}

Following Section~\ref{sec:lm_balancing}, we impose the following linear model: $g_{ci}^{(\bm a_c)}(\bm X_c) = \bm X_{ci}^\top\left(\bm \Lambda_{ci}^{(\bm \beta)}(\bm a_c)\otimes \bm I_p\right)\bm h$, where $p=17$ is the dimension of the covariates. Under this model, we consider the following candidate low-rank structures and use them to illustrate the methodology developed in this paper. For brevity of presentation, we briefly describe these structures, while their explicit mathematical expressions are deferred to Appendix~\ref{sec:imbal_real_data}.
\begin{enumerate}
    \item {\bf No interference:} The potential outcome of a unit only depends on its own treatment status. We denote this low-rank structure by $\bm \Lambda_{\text{none}}$. This is likely a misspecified structure as \cite{banerjee2013} acknowledges the presence of interference.
    
    \item {\bf Anonymous interference within neighborhood:} The potential outcome of a unit is an additive function of its own treatment status and the number of treated households in its neighborhood.  We denote this anonymous interference structure by $\bm \Lambda_{\text{neighbor1}}$.  However, the number of treated neighbors can take a large number of unique values, leading to many effective treatments.  This makes it impossible to achieve covariate balance in this data set. We therefore also consider a discretization of this variable to reduce the number of effective treatments.  Specifically, we use three categories (low, medium, and high), where the thresholds are the 33$^{rd}$ and 67$^{th}$ percentiles of all values across all the units. Then, we use $\bm \Lambda_{\text{neighbor1}}^\ast$ to denote anonymous interference based on this coarsened variable.
    
    \item \textbf{Anonymous interference within neighborhoods of neighbors}: We assume that the potential outcome of one unit is an additive function of its own treatment, treatment of neighbors, as well as neighbors of neighbors. Like the previous intervention, we assume anonymous interference based on the number of the treated among one's neighbors, as well as among one's neighbors of neighbors. We denote this structure by $\bm \Lambda_{\text{neighbor2}}$.  We also discretize both variables as done above using the three categories (low, medium, and high).  We denote the resulting structure based on the two coarsened variables by $\bm \Lambda_{\text{neighbor2}}^*$.
    
\end{enumerate}


\subsection{Findings}

For each of the above three policies  and the five candidate low-rank structures, we obtain the omnibus relative imbalance measures for all $17$ covariates as described in Section~\ref{sec:assessing_imbalance}. We find that when the coarsened versions of low-rank structures ($\bm \Lambda_{\text{none}}$, ${\bm \Lambda}^*_{\text{neighbor1}}$, and ${\bm \Lambda}^*_{\text{neighbor2}}$) were used, the maximum omnibus relative imbalance for any covariate across any combination of the policies and coarsened low-rank structures is less than 0.0001, implying a negligible degree of imbalance.  Therefore, balancing equations are satisfied under all the coarsened structures. 

\begin{figure}[t!]
    \centering \includegraphics{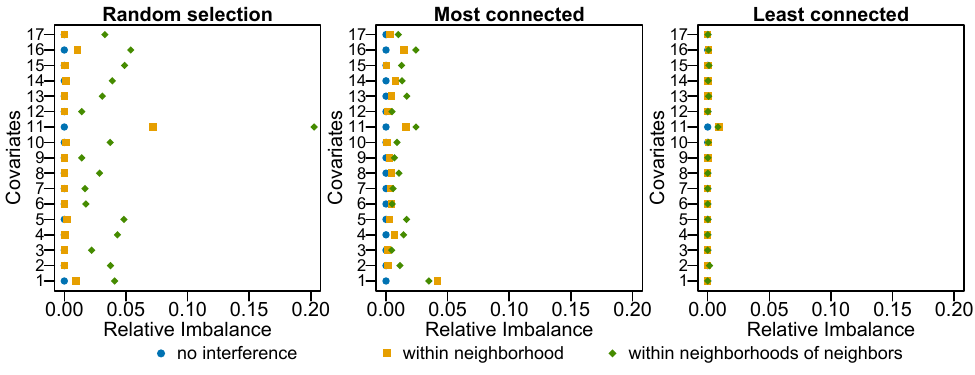}
    \caption{Relative covariate imbalance measures for all the covariates across the three non-coarsened low-rank structures.  These low-rank structures include no interference (blue circles), anonymous interference within neighborhoods (yellow squares), and anonymous interference within neighborhoods of neighbors (green diamonds). The imbalance measure is evaluated under three stochastic interventions: random selection of households (left panel), selection of most connected households (middle), and selection of least connected households (right).}
    \label{fig:imbal_real_data}
\end{figure}

Under the non-coarsened low-rank structures (${\bm \Lambda}_{\text{neighbor1}}$ and ${\bm \Lambda}_{\text{neighbor2}}$), however, the relative imbalance for some covariates remains non-negligible. Figure~\ref{fig:imbal_real_data} summarizes these results. We find that for interference restricted to within neighborhoods, the maximum relative imbalance across the covariates is about 0.075. When interference is restricted to the neighborhoods of neighbors, we obtain the maximum relative imbalance of 0.2. As expected, we find that a more informative low-rank structure results in better covariate balance.

\begin{table}[!b]
    \centering
    \begin{tabular}{|ll|r|r|r|}
    \hline null & alternative
         &  $f_{\textrm{random}}$ & $f_{\textrm{most}}$ & $f_{\textrm{least}}$\\
         \hline
         $H_{\text{none}}$ & $H_{\text{neighbor1}}$ & \textbf{0.085} & \textbf{0.074}&  \textbf{0.024}\\
         \hline
         $H_{\text{none}}$ & $H_{\text{neighbor2}}$ & 0.367 & 0.131 & \textbf{0.058}\\
         \hline
         $H_{\text{neighbor1}}$ & $H_{\text{neighbor2}}$ & 0.562  & 0.535 & 0.414 \\
         \hline
    \end{tabular}
     \caption{Asymptotic p-values for the tests of three candidate low-rank structures under the three stochastic interventions (listed across columns). Entries less than 0.1 are boldfaced.}
     \label{tab:pval_structre_test}
\end{table}

We apply Theorem~\ref{thm:decide_true_structure} and choose a low-rank structure among the coarsened low-rank assumptions, which lead to negligible covariate imbalance. We consider the following three hypotheses: $H_{\text{none}}:\bm g\in \mathrm{col.sp.}({\bm \Lambda}_{\text{none}})$, $H_{\text{neighbor1}}:\bm g\in \mathrm{col.sp.}({\bm \Lambda}_{\text{neighbor1}}^*)$, and $H_{\text{neighbor2}}:\bm g\in \mathrm{col.sp.}({\bm \Lambda}_{\text{neighbor2}}^*)$.  We then conduct three pairwise tests under each stochastic intervention. Note that this is done for the illustrative purpose, as in practice researchers will use the most over-specified structure as the alternative. Our analysis also illustrates that the statistical power depends on how close the most over-specified structure is to the exact specification.


Table~\ref{tab:pval_structre_test} presents the p-values of the pairwise tests. We find that assuming ${\bm \Lambda}_{\text{neighbor2}}^*$ as correct low-rank structure (as in Theorem~\ref{thm:decide_true_structure}) does not reject any of the other structures except for under the intervention $f_{\textrm{least}}$ where it rejects $H_{\textrm{none}}$ (at $6\%$ level). Assuming ${\bm \Lambda}_{\text{neighbor1}}^*$ as correct low-rank structure, however, significantly detects ${\bm \Lambda}_{\text{none}}^*$ to be false (at 10$\%$ level) under all the counterfactual interventions (and hence, showing the presence of interference). Thus, in order for our low-rank structure detection method described in Section~\ref{sec:choose_structure} to be powerful, the most uninformative, correct structure needs to be as well-specified as possible.

\begin{table}[t]
     \centering
     \begin{tabular}{|l|c|c|c|}
     \hline
      & $f_{\textrm{random}}$ & $f_{\textrm{most}}$  & $f_{\textrm{least}}$ \\
     \hline
      No interference    & \makecell{16.92$\%$\\ (15.46, 19.65) } &  \makecell{17.92$\%$ \\ (15.32, 19.51)} &  \makecell{16.27$\%$\\ (15.54, 19.73)} \\
     \hline
     Within neighborhoods & \makecell{17.32$\%$\\ (15.24, 19.39)} & \makecell{17.99$\%$\\ (15.91, 20.07)}  & \makecell{16.1$\%$\\ (14.02, 18.17)}\\
     \hline
     Within neighborhoods of neighbors &  \makecell{17.55$\%$\\ (14.94, 18.89)} & \makecell{17.42$\%$\\ (15.96, 19.90)} & \makecell{17.63$\%$\\ (14.30, 18.24)}\\
     \hline
     \end{tabular}
     \caption{Point estimates and 95$\%$ confidence intervals of the program participation rates under the three stochastic interventions (columns) and three candidate low-rank structures (rows). Note that for interference restricted to within neighborhoods and neighborhoods of neighbors, we only use the coarsened low-rank structures.}
     \label{tab:inference_real_data}
\end{table}


Lastly, Table~\ref{tab:inference_real_data} summarizes the estimates of program participation rate (and their 95\% confidence intervals) based on the balancing estimator under different stochastic interventions and low-rank structures.  Recall that as shown in Table~\ref{tab:pval_structre_test}, the coarsened neighborhood interference is the likely low-rank structure among the three structures considered here.  Under this low-rank structure, the estimated participation rates across the three counterfactual interventions are all similar.  The point estimates are slightly higher for the intervention on the most connected households while the lowest estimated participation rate is obtained for the intervention on the least connected households.  However, the confidence intervals are wide and overlapping.  We also note that the average participation rate when the leaders are treated is about $17.99\%$, similar to the estimate under the intervention on the most connected households. The results suggest that we do not have strong evidence for distinguishing the effectiveness of the three intervention strategies considered in our analysis.

\section{Concluding Remarks}

In this paper, we developed new weighting estimators for unbiased estimation of causal effects under interference.  We demonstrated that the proposed low-rank structure of interference patterns can lead to a substantial efficiency gain.  Our low-rank structure encompasses a wide range of common interference patterns and enables researchers to adopt their own interference assumptions in a given study. Importantly, our methodology does not require knowledge of the true propensity score, making it applicable in observational studies. 
We derived asymptotic properties of these low-rank weighting estimators, and demonstrated their effectiveness through a wide range of simulations and a real-data study. 

There are several possible extensions of the proposed methodology. First, while we require partial interference, the extension to a single network interference setting is of interest. For example, one may use graph asymptotics techniques \citep[e.g.,][]{liandwager2022} to further explore the behavior of the proposed estimators in this setting.  We emphasize that the finite sample properties of our estimator developed in Section~\ref{sec:low_rank_weights} hold even under such single network setting.

Second, our balancing estimator requires the balancing condition to hold exactly. While we have shown that this requirement is likely to be met as the sample size increases, it may not be satisfied in small samples.  Developing a method that enables for a small degree of imbalance is useful.  We may also consider the settings in which a low-rank structure holds only approximately. For example, some researchers have considered approximate neighborhood interference \citep[e.g.,][]{leung2022}.  

\section*{Acknowledgements}
The authors thank an anonymous reviewer of IQSS's rapidPeer pre-peer review program for helpful comments. SS also thanks Daniel Nevo and Yuzhou Lin for their helpful discussions and comments.

\bibliography{refs}

\clearpage
\appendix

\begin{center}
\LARGE \textbf{Supplementary Appendix}
\end{center}

\section{Low-rank weighting estimators with the known propensity score}
\label{sec:extension_known_propensity}

In this appendix, we provide a detailed discussion of low-rank weighting estimators when the true propensity score is known.  We also present the formal results, which are discussed in Section~\ref{sec:propensity_extensions}. First, we introduce an additional notation: for any choice of weight $\bm w$, which depends on the treatment pattern $\bm A$ of all units and the set of covariates $\{\bm X_c\}_{c=1}^n$, let $\tilde{\bm w}(\bm a)$ denote the {\it potential} weights one would obtain when $\bm A = \bm a$ (for ease of presentation, we have suppressed the possible notational dependence of $\tilde{\bm w}(\bm a)$ on the set of covariates $\{\bm X_c\}_{c=1}^n$). Stack $\tilde{\bm w}(\bm a)$ across different values of $\bm a$ to form the vector of potential weights $\tilde{\bm w}:= (\tilde{\bm w}(\bm a^{(0)}), \ldots, \tilde{\bm w}(\bm a^{(2^M-1)}))^\top$. Note that the {\it observed} weights satisfy $\bm w = \bm R\tilde{\bm w}$ where $\bm R$ is defined in Equation~\eqref{eqn:r}. Then the following is an equivalent representation of the class $\mathcal C_{\bm \Lambda}(\bm h, e)$ described in Equation~\eqref{eq:unbiased_estimators_R}.
\begin{equation}
\label{eq:appendix_unbiased_estimators_R}
\mathcal C_{\bm \Lambda}(\bm h, e) = \left\{\bm R \tilde{\bm w}: \mathbb E\left[(\tilde{\bm w} - \tilde{\bm w}_{\mathrm{IPW}})^\top\bm R^\top\bm R \bm \Lambda \bm h\right] = 0\right\}.
\end{equation}

\subsection{Projection and weighted projection estimators}

Similar to Theorem~\ref{thm:subclass2}, we characterize the class of uniformly unbiased weighting estimators across the choice of $\bm h$ since $e$ is known. The following theorem now formally defines this class.
\begin{theorem}
    \label{thm:extension_subclass2}
    Under the setting and assumption of Theorem~\ref{thm:subclass2}, but for a known, true propensity score $e$, a uniformly (i.e., across all possible $\bm h$ values in Equation~\eqref{eq:appendix_unbiased_estimators_R}) unbiased weighting estimator is characterized by
    \begin{align}
        \label{eqn:eprojdefn}
        \bm w \in \mathcal C_{\bm \Lambda}(\bm h, e),\forall \bm h\Longleftrightarrow  \bm w \in \Cipwlamb(e) := \left\{\bm R\tilde{\bm w}: \bm \Lambda^\top\bm D(\tilde{\bm w} - \tilde{\bm w}_{\mathrm{IPW}}) = \bm 0,\textrm{ a.s.}\right\},
    \end{align}
    where, $\bm D = \mathbb E[\bm R^\top\bm R\mid \bm X]$.
\end{theorem}
Unlike the definition of $\Cnulllamb$ given in Equation~\eqref{eqn:baldefn}, Equation~\eqref{eqn:eprojdefn} provides a necessary and sufficient condition, so that $\Cipwlamb(e)$ is an exact characterization of uniformly unbiased weights. Furthermore, $\Cipwlamb(e)$ is non-empty as it always contains the IPW estimator. Another important characteristic of $\Cipwlamb(e)$ is that it contains more than just the IPW estimator even for a minimal restriction imposed through $\bm \Lambda$ (unlike, $\Cnulllamb$, which requires the linear model assumption of Section~\ref{sec:lm_balancing}). To see this, observe that the dimension of $\Cipwlamb(e)$ is the same as the dimension of the null-space of $\bm \Lambda^\top\bm D$. The matrix $\bm D$ includes the propensity scores as its diagonal elements, so that under Assumption~\ref{assm:pos}, it is invertible. This in turn implies that the dimension of $\Cipwlamb(e)$ matches the dimension of the null-space of $\bm \Lambda^\top$, which is a $\mathrm{dim}(\bm h)\times \mathrm{dim}(\bm g)$ matrix. A minimal restriction for the low-rank structure $\bm \Lambda$ ensures that $\mathrm{dim}(\bm h)< \mathrm{dim}(\bm g)$, resulting in a non-trivial nullspace.

Next, following the same approach used in Section~\ref{sec:low_rank_weights}, we choose weights from $\Cipwlamb(e)$, by optimizing some objective function (e.g., $L_2$ norm). Noting that $\Cipwlamb(e)$ directly imposes a restriction on $\tilde{\bm w}$ (rather than on $\bm w$, as $\Cnulllamb$ does), our choice of the objective function naturally extends to $\phi^{\sone}({\bm w}) = \mathbb E\left[\phinull(\bm R\tilde{\bm w})\mid \bm X\right] = \tilde{\bm w}^\top\bm D\tilde{\bm w}$. We use $\wipwlamb(e)$ to denote the set of optimal weights obtained by solving Equation~\eqref{eqn:optimiation_problem} with $\phi = \phiipw$ and $\Cnulllamb$ replaced by $\Cipwlamb(e)$. 

We first obtain the expression of $\wipwlamb(e)$ under a general low-rank structure assumption: $g_{ci}^{(\bm a_c)}(\bm X_c) = \bm \Lambda_{ci}(\bm a_c)\bm h_{ci}$, as in Examples~\ref{ex:1}--\ref{ex:3}.  To distinguish it from the specifications discussed in Section~\ref{sec:lm_balancing}, we refer to it as the \textit{unit-specific} low-rank structural assumption. Under this general assumption, each unit is allowed to have a different $\bm h_{ci}$ and is not required to share a fixed-length $\bm h$ vector as in Section~\ref{sec:lm_balancing}. Recall from Section~\ref{sec:lm_balancing} that the balancing estimator does not exist in this case. The following theorem shows that in such a case, the expression of the weight for each unit simplifies substantially and solely depends on the cluster characteristics to which the unit belongs; its proof crucially hinging on the fact that $\bm \Lambda$ is block-diagonal in this case. 
\begin{theorem}
\label{thm:simplified_indiv_weights}
    Suppose that the following low-rank assumption holds: $g_{ci}^{(\bm a_c)}(\bm X_c) = \bm \Lambda_{ci}(\bm a_c)\bm h_{ci}$. Then, solving Equation~\eqref{eqn:optimiation_problem} with $\phi = \phiipw$ and $\Cnulllamb$ replaced by $\Cipwlamb(e)$ is equivalent to solving the problem for each individual weight $w_{ci}$ using the low-rank structure, $\bm \Lambda_{ci}$. Under Assumption~\ref{assm:pos}, this yields,
    \begin{align}
    \label{eqn:misspecified_indiv}
    (\wipwlamb(e))_{ci} = \bm R^\top_{c}\bm D_{c}^{-1/2}\bm P_{\bm D_{c}^{1/2}\bm \Lambda_{ci}}\bm D_{c}^{1/2}\tilde{\bm w}_{\mathrm{IPW},c,i},
    \end{align}
    where, $\bm R_c = \left(\mathbb I\left(\bm A_c = \bm a^{(0)}_c\right), \dots, \mathbb I\left(\bm A_c = \bm a^{(2^{M_c}-1)}_c\right)\right)^\top$ and $\bm D_c = \mathbb E\left[\bm R_c^\top\bm R_c\middle| \bm X_c\right]$ with $\{0,1\}^{M_c} = \left\{\bm a_c^{(j)}:0\leq j\leq 2^{M_c}-1\right\}$ denoting an enumeration. 
\end{theorem}

Inspired by Equation~\eqref{eqn:misspecified_indiv}, we refer to $T\left(\wipwlamb(e)\right)$ as the \textit{weighted projection estimator}. Although Theorem~\ref{thm:simplified_indiv_weights} shows that $T(\wipwlamb(e))$ can be obtained in cases that do not require the linear model assumption, we pay a price for this flexibility. The variance of the weighted projection estimator is no longer provably smaller than that of the IPW estimator, even under the asymptotic regime (as Section~\ref{sec:varying_asymp} discusses). Recall that this fact holds true under the linear model assumption (as in Section~\ref{sec:lm_balancing}); in this case, the balancing estimator significantly outperforms the IPW estimator. Interestingly, we show in Section~\ref{sec:weighted_generalization} that many existing weighting estimators in the literature which exploit a known propensity score are a special case of the weighted projection estimator.

As discussed above, the variance of $T(\wipwlamb(e))$ cannot provably be smaller than that of the IPW estimator. This begs the question as to whether we can leverage the knowledge of the true propensity score and create a weighting estimator that is provably no less efficient than the IPW estimator. We show that if $\bm h$ is fixed as in Section~\ref{sec:lm_balancing}, it is possible to find such a weighting estimator. To begin with, consider the following sub-class of $\Cipwlamb(e)$:
\[
    \Cipwprojlamb(e) := \{\bm w: \bm \Lambda^\top\bm R^\top(\bm w - {\bm w}_{\mathrm{IPW}}) = \bm 0\}.
\]
Note that $ \Cipwprojlamb(e)$ is also always non-empty, as it also trivially contains the IPW estimator. However, similarly to $\Cnulllamb$, in order for it to contain more than that we need to make substantial restrictions. In particular, we need to ensure that $\bm R\bm \Lambda$ has more rows than columns. A linear model assumption, such as that in Section~\ref{sec:lm_balancing}, satisfies that. We again optimize $\phi = \phinull$ over $\bm w \in \Cipwprojlamb$ to choose an optimal weight, which we denote by $\wipwprojlamb(e)$. The following theorem derives the variance of $T(\wipwprojlamb(e))$ along with listing the expressions of $\wipwprojlamb(e)$ as well as $\wipwlamb(e)$. As a consequence, it follows that $T(\wipwprojlamb(e))$ has provably smaller variance than the IPW estimator.

\begin{theorem}  \label{thm:propensity_estimators_fixed_h}
    Assume the linear model assumption of Section~\ref{sec:lm_balancing}. Then, it holds that
    \begin{enumerate}[label = (\alph*)]
         \item $\wipwlamb(e) = \boldsymbol R \boldsymbol D^{-1/2}\boldsymbol P_{\boldsymbol D^{1/2}\boldsymbol \Lambda}\boldsymbol D^{1/2}\tilde{\boldsymbol w}_{\mathrm{IPW}}$,
        \item $\wipwprojlamb(e) = \bm P_{\bm R\bm\Lambda}{\bm w}_{\mathrm{IPW}}$.
    \end{enumerate}
    Furthermore, under Assumptions~\ref{assm:partial}--\ref{assm:outcome_model},~\ref{ass:low_rank},~and~\ref{assm:iid_error}, we have that
    \[
        \forall \bm w \in \Cipwprojlamb(e), \mathrm{Var}(T(\bm w)) = \mathrm{Var}(\bm w_{\mathrm{IPW}}^\top\bm R\bm g)/n^2 + \sigma^2\mathbb E\|\bm w\|^2/n^2,
    \] 
    and also that $T(\wipwprojlamb(e))$ has the a smaller variance than any $T(\bm w)$, for $\bm w \in \Cipwprojlamb(e)$.
\end{theorem}
Inspired by item (b) of Theorem~\ref{thm:propensity_estimators_fixed_h}, we call $T(\wipwprojlamb(e))$ the \textit{projection estimator}. Because $\bm w_{\mathrm{IPW}}\in \Cipwprojlamb(e)$, from Theorem~\ref{thm:propensity_estimators_fixed_h} we see that the projection estimator has a smaller variance than the IPW estimator. Furthermore, this variance can be decomposed into parts, where only the second part, that is, $\sigma^2\mathbb E\mathbb \|\bm P_{\bm R\bm \Lambda}\hat{\bm w}_{\mathrm{IPW}}\|^2/n^2$ is affected by the low-rank structure. Analogous to the phenomenon described after Corollary~\ref{cor:correct_specification_fixed}, we see that more informative low-rank structures (corresponding to lower rank of $\bm \Lambda$) results in higher reduction of the variance. Also, the extent of this reduction depends on how large $\sigma^2$ is compared to the signal size. In general, when a signal-to-noise ratio is not too high, this reduction can be significant.

\subsection{Asymptotic analysis of estimators exploiting the known propensity score}
\label{sec:varying_asymp}

Next, we conduct an asymptotic analysis analogous to the one in Section~\ref{sec:asymptotics}, when the true propensity score is known.
First, we derive a central limit theorem (CLT) for $T(\wipwlamb(e))$ under the general low-rank structural assumption: $g_{ci}^{(\bm a_c)}(\bm X_c) = \bm \Lambda_{ci}(\bm a_c)\bm h_{ci}$. Because Theorem~\ref{thm:simplified_indiv_weights} implies that the weights of all units in a cluster are functions of the cluster characteristics, the weights for each cluster are i.i.d. across clusters. With this observation, the next theorem is a direct application of the classical CLT, and hence, its proof is omitted.  Like Theorem~\ref{thm:clt_fixed}, this theorem allows for potential misspecification of the low-rank structure.  

\begin{ass}[Regularity conditions under a unit-specific low-rank structure] 
\label{assm:clt_technical_indiv}
The following two conditions are satisfied:
\begin{enumerate}[label = (\alph*)]
    \item $\mathbb E \left\|\bm g_{c}^{(\bm A_c)}\right\|^4<\infty$, $ \mathbb E\left[\left(\sum_{i=1}^{M_c}\left\|\tilde{\bm w}_{{\rm IPW},ci}\right\|^2\right)^2\right]  <\infty$, for all $c$.
    \item (Strong positivity) We assume that the propensity score satisfies
    \[
        \inf_{\bm a_c \in \{0,1\}^{M_c}}\mathbb P(A_c = \bm a_c\mid \bm X_c)>\eta,~\textrm{almost surely},
    \]
    for some $\eta>0$. Note that this randomness is over the full set of cluster characteristics, including its size.
\end{enumerate}
\end{ass}

\begin{theorem}[Asymptotic normality under a unit-specific low-rank structure]
\label{thm:clt_indiv}
    Suppose that the low-rank structure can be written as
    $g_{ci}^{(\bm a_c)}(\bm X_c) = \bm \Lambda_{ci}^\ast h_{ci}$,
    where $\bm \Lambda^\ast_{ci}$ does not depend on any characteristics beyond the $c^{\mathrm{th}}$ cluster. Let $\bm \Lambda_{ci}$ be the hypothesized structure used to obtain the weights, also only depend on the $c^{\mathrm{th}}$ cluster. Now, define 
    $$\bm g_{ci}(\bm X_c):= \left(g_{ci}^{\left(\bm a_c^{(0)}\right)}(\bm X_c),\dots, g_{ci}^{\left(\bm a_c^{(2^{M_c}-1)}\right)}(\bm X_c)\right)^\top, \quad  \bm f_c:=(f(\bm a_c^{(0)},\bm X_c),\dots, f(\bm a_c^{{(2^{M_c}-1)}},\bm X_c))^\top.$$ Then, under Assumptions~\ref{assm:partial}--\ref{assm:outcome_model}, \ref{assm:iid_error}~and~\ref{assm:clt_technical_indiv}, we have that
    \[    \sqrt{n}\left(T (\wipwlamb(e)) - \mu_f(\bm \Lambda, \bm \Lambda^\ast)\right) \dto N\left(0, \tilde\sigma^2_{\sone}(\bm \Lambda, \bm \Lambda^\ast)\right),
    \]
    where,
    \begin{align*}
        \mu_{f}(\bm \Lambda, \bm \Lambda^\ast) &= \mathbb E\left[\frac{1}{M_c}\sum_{i=1}^{M_c}\bm f_{c}^\top\bm D_{c}^{-1/2}\bm P_{\bm D_{c}^{1/2}\bm \Lambda_{ci}}\bm D_{c}^{1/2}\bm g_{ci}(\bm X_c)\right],\textrm{ and }\\
        \tilde\sigma^2_{\sone}(\bm \Lambda, \bm \Lambda^\ast) &= \mathrm{Var}\left(\sum_{i=1}^{M_c}\left(\wipwlamb(e)\right)_{ci} g^{(\bm A_c)}_{ci}(\bm X_c)\right) + \sigma^2 \mathbb E\left[\sum_{i=1}^{M_c} \left\|\bm P_{\bm D_{c}^{1/2}\bm \Lambda_{ci}}\bm D_{c}^{1/2}\tilde{\bm w}_{\mathrm{IPW},c}\right\|^2\right].
    \end{align*}    
\end{theorem}
According to this theorem, under a misspecified low-rank structure, the weighting estimator is centered around the expectation (under $f$, when it is a stochastic intervention) of the propensity scores-weighted projection of the original potential outcome regression vector onto the space defined by the assumed low-rank structure $\bm \Lambda$. 


Finally, similar to Corollary~\ref{cor:correct_specification_fixed}, the following corollary shows that under correct specification, the estimator is unbiased and the variance can be further simplified.
\begin{corollary}[Correctly specified unit-specific low-rank structure] \label{cor:correct_specification_indiv}
     Under the setting of Theorem~\ref{thm:clt_indiv}, if the low-rank structure is correctly specified, that is, $\bm \Lambda_{ci} = \bm \Lambda^\ast_{ci}$ for all $c$ and $i$, then we have that 
     $\mu_f(\bm \Lambda, \bm \Lambda) = \mu_f$. Furthermore, we have
     \begin{equation}
     \label{eqn:asymp_var_decomp}
     \begin{aligned}
         \tilde \sigma^2_{\sone}(\bm \Lambda, \bm \Lambda) \ = \ & \mathrm{Var}\left(\mathbb E\left[\bm w_{\mathrm{IPW},c,i}g_{ci}^{(\bm A_c)}(\bm X_c)\middle| \bm X_c\right]\right) + \mathbb E\left[\mathrm{Var}\left(\bm g_c^{(\bm A_c)^\top}\wipwlamb(e)_c\mid\bm X_c\right)\right] \\
         & \quad + \sigma^2 \mathbb E\|\wipwlamb(e)_{c}\|^2,
         \end{aligned}
     \end{equation}
     where $\wipwlamb(e)_c$ is a vector formed by stacking $\wipwlamb(e)_{ci}$ across $i$.
\end{corollary}
Thus, under correct specification, only the second and third terms in the expression of 
$\tilde\sigma^2_{\sone}(\bm \Lambda, \bm \Lambda)$ depend on the low-rank structure. Now, the proof of Theorem~\ref{thm:clt_indiv} presented in Appendix~\ref{sec:proof_thm_clt_indiv} establishes the following relation with regard to the third term of Equation~\eqref{eqn:asymp_var_decomp}:
\begin{align*}
    \mathbb E\|\wipwlamb(e)_{c}\|^2
    = \sum_{i=1}^{M_c}\mathbb E\|\bm P_{\bm D_c^{1/2}\bm \Lambda_{ci}}\bm D_c^{1/2}\tilde{\bm w}_{\mathrm{IPW},c,i}\|^2\leq \sum_{i=1}^{M_c}\mathbb E\|\bm D_c^{1/2}\tilde{\bm w}_{\mathrm{IPW},c,i}\|^2=\mathbb E\|{\bm w}_{\mathrm{IPW},c}\|^2.
\end{align*}
This implies that the term $ \sigma^2\mathbb E\|\wipwlamb(e)_{c}\|^2$ is smaller than the corresponding term $\sigma^2\mathbb E\|\bm w_{\mathrm{IPW},c}\|^2$ of the IPW estimator. As previously noted in the discussion following Theorem~\ref{thm:clt_fixed}, this difference is due to the assumed low-rank structure.  Indeed, the more restrictive our assumption is, the lower the rank of the projection matrix is, making the difference more substantial.

Unfortunately, we cannot directly compare the second term of Equation~\eqref{eqn:asymp_var_decomp} under a low-rank structure and the same term without the structure (in which case, the term $\wipwlamb(e)_c$ would be replaced by $\bm w_{\mathrm{IPW},c}$). Thus, we cannot provably establish an ordering among $\bar\sigma^2_{\sone}(\bm \Lambda, \bm \Lambda)$ and $\sigma^2_{\mathrm{IPW}}$. However, we expect that when the signal-to-noise ratio is not too high, the reduction in the third term obtained under a substantial low-rank assumption would dominate the effect of the second term, resulting in the superior performance of $T(\wipwlamb(e))$.

Next, we establish the asymptotic normality of $T(\wipwlamb(e))$ and $T(\wipwprojlamb(e))$ under the linear model assumption of Section~\ref{sec:lm_balancing}.

\begin{theorem}[Asymptotic normality under linear model assumption]
\label{thm:known_propensity_clt_fixed}
    Under the settings and assumptions of Theorem~\ref{thm:clt_fixed} (but without Assumption~\ref{assm:balancing_feasibility}), we have that for $(\hat{\bm w}, \bar\sigma^2(\bm \Lambda, \bm \Lambda^\ast))$ = $(\wipwlamb(e),\sigma^2_{\sone}(\bm \Lambda, \bm \Lambda^\ast))$ or $(\wipwprojlamb(e),\sigma^2_{\stwo}(\bm \Lambda, \bm \Lambda^\ast))$,
    \begin{align*}
        \sqrt{n}\left(T(\hat {\bm w}) - \mu_f(\bm \Lambda, \bm \Lambda^\ast)\right)\dto \mathcal N\left(0, \bar \sigma^2(\bm \Lambda, \bm \Lambda^\ast)\right),
\end{align*}
    where,
    \begin{align*}
    \sigma^2_{\stwo}(\bm \Lambda, \bm \Lambda^\ast) &= \mathrm{Var}\left[\left((I - {\mathcal P})\bm g_c^{(\bm A_c)}\right)^\top\left({\mathcal P}\bm w_{\mathrm{IPW},c}\right) + \bm w_{\mathrm{IPW},c}^\top {\mathcal P}\bm g_c^{(\bm A_c)}\right] +  \sigma^2\mathbb E\left[\left\| {\mathcal P}\bm w_{\mathrm{IPW},c}\right\|^2\right]\\
    \sigma^2_{\sone}(\bm \Lambda, \bm \Lambda^\ast) &= \mathrm{Var}\left({\bm g_c^{(\bm A_c)}}^\top {\mathcal P}\bm w_{\mathrm{IPW,c}} - \mathbb E\left[\left({\mathcal P}\bm g_c^{(\bm A_c)}\right)^\top {\mathcal P}\bm w_{\mathrm{IPW,c}}\ \middle| \ \bm X_c\right] + \mathbb E[\bm w_{\mathrm{IPW},c}^\top {\mathcal P}\bm g_c^{(\bm A_c)}\mid \bm X_c]\right)\\
    & \quad +  \sigma^2\mathbb E\left[\left\| {\mathcal P}\bm w_{\mathrm{IPW},c}\right\|^2\right].
    \end{align*}
\end{theorem}
\begin{corollary}[Correctly specified fixed low-rank structure] \label{cor:propensity_correct_specification_fixed}
    If the low-rank structure is either exactly specified or over-specified, then,
    $$
    \begin{aligned}
        \sigma^2_{\stwo}(\bm \Lambda, \bm \Lambda) = & \mathrm{Var}\left[{\bm g_c^{(\bm A_c)}}^\top\bm w_{\mathrm{IPW},c}\right] + \sigma^2\mathbb E\left[\left\| {\mathcal P}\bm w_{\mathrm{IPW},c}\right\|^2\right],\\
        \sigma^2_{\sone}(\bm \Lambda, \bm \Lambda)  = & \mathrm{Var}\left[\mathbb E\left({\bm g_c^{(\bm A_c)}}^\top\bm w_{\mathrm{IPW},c}\Bigr |\bm X_c\right)\right] + \mathbb E\left[\mathrm{Var}\left( {\bm g_c^{(\bm A_c)}}^\top {\mathcal P}\bm w_{\mathrm{IPW},c}\ \Bigr | \ \bm X_c\right)\right] \\
        & \qquad + \sigma^2\mathbb E\left[\left\| {\mathcal P}\bm w_{\mathrm{IPW},c}\right\|^2\right]\end{aligned}$$
\end{corollary}
Like Theorem~\ref{thm:propensity_estimators_fixed_h}, a more informative low-rank structure yields a smaller value of $\sigma^2_{\stwo}(\bm \Lambda, \bm \Lambda)$ by reducing the second term alone. In particular, we obtain the ordering $\sigma^2_{\sthree}(\bm \Lambda, \bm \Lambda)\leq \sigma^2_{\stwo}(\bm \Lambda, \bm \Lambda)\leq \sigma^2_{\mathrm{IPW}}$, which leads to a surprising fact that despite the lack of knowledge of the true propensity score, the balancing estimator is the most efficient. As discussed in Section~\ref{sec:asymp_normality}, a large decrease in this variance is due to the balancing equation satisfied by $\wnulllamb$. Finally, similar to Corollary~\ref{cor:correct_specification_indiv}, $\sigma^2_{\sone}$ is not always smaller than $\sigma^2_{\mathrm{IPW}}$. However, we expect this to be the case when the signal-to-noise ratio is not too high.

In sum, the balancing estimator is preferred under the linear model setting so long as Assumption~\ref{assm:balancing_feasibility} is satisfied. As discussed in Section~\ref{sec:balancing_feasibility}, this assumption holds under mild conditions. However, if the assumption is not met, the projection estimator should be used. If the linear model assumption does not hold, the weighted projection estimator is the only non-trivial choice. As mentioned earlier, we expect this estimator to outperform the IPW estimator so long as the signal-to-noise ratio is not too high.

Finally, the results given in Theorems~\ref{thm:var_estimation}~and~\ref{thm:decide_true_structure} can also be extended to the estimators discussed in this Section~\ref{sec:extension_known_propensity}. For example, under a general low-rank structure assumption, the asymptotic normality of $T(\wipwlamb(e))$ can be established using the classic CLT (see the discussion immediately preceding Theorem~\ref{thm:clt_indiv}). Hence, the sample variance can be used to consistently estimate the asymptotic variance. Under the linear model assumption of Section~\ref{sec:lm_balancing}, Theorem~\ref{thm:var_estimation} can be extended to the projection and weighted projection estimators as well; the expressions of these variance estimators are given in Equation~\eqref{eqn:var_est_exprn}.

In addition, under the linear model assumption, Theorem~\ref{thm:decide_true_structure} also holds if instead of the balancing weights, one uses the weights of the projection estimator. We prove this in Section~\ref{sec:proof_decide_true_structure}. If not under the linear model assumption, a similar result can also be derived for the weighted projection estimator. In particular, define, 
\[
    \hat \sigma_{lL}^{-2} = \frac{1}{n}\sum_{c=1}^n \left(\sum_{i=1}^n \left(\hat{\bm w}^{\sone}_{\bm \Lambda_l}(e)_{ci} - \hat{\bm w}^{\sone}_{\bm \Lambda_L}(e)_{ci}\right)y_{ci}\right)^2 - \left(T\left(\hat{\bm w}^{\sone}_{\bm \Lambda_l}(e)\right) - T\left(\hat{\bm w}^{\sone}_{\bm \Lambda_L}(e)\right)\right)^2,
\]
and with the abuse of notation, $S_{lL} = n\hat\sigma_{lL}^{-2}\left(T\left(\hat{\bm w}^{\sone}_{\bm \Lambda_l}(e)\right) - T\left(\hat{\bm w}^{\sone}_{\bm \Lambda_L}(e)\right)\right)^2$. Note that $T\left(\hat{\bm w}^{\sone}_{\bm \Lambda_l}(e)\right) - T\left(\hat{\bm w}^{\sone}_{\bm \Lambda_L}(e)\right)$ is also a sum of i.i.d. terms and $\hat \sigma_{lL}^{-2}$ is the sample variance of these terms (and hence, is consistent for the population version). A straightforward application of the CLT establishes that the conclusions of Theorem~\ref{thm:decide_true_structure} also hold in this case.

\section{Weighted projection estimators as a generalization of existing weighting estimators}
\label{sec:weighted_generalization}

In this section, we highlight the fact that the weighted projection estimator is a generalization of many existing estimators proposed in the literature. For example, the literature dealing with observational studies under (partial) interference has predominantly dealt with a low-rank assumption based on a graph that summarizes the interference (for example, see \cite{liuhudgensbacker-dreps2016,kilpatrick2024} and the references therein). In particular, many of these assumptions can be encompassed under the umbrella of a discrete version of the exposure mapping \citep{aronow2017} assumption, as we define below:
\begin{defn}[Discrete exposure mapping]
    \label{defn:exposure_mappings}
    For an $i^{\mathrm{th}}$ unit in the $c^{\mathrm{th}}$ cluster, a discrete exposure mapping is a function $\phi_{ci}:\{0,1\}^{M_c}\mapsto \mathcal S_{ci}$ where $\mathcal S_{ci}$ is finite, discrete.
\end{defn}
We can make a general low-rank structural assumption based on a discrete exposure mapping: For any $i^{\mathrm{th}}$ unit in the $c^{\mathrm{th}}$ cluster, $\phi_{ci}(\bm a_c) = \phi_{ci}(\bm a_c')\implies g^{(\bm a_c)}_{ci}(\bm X_c) = g^{(\bm a_c')}_{ci}(\bm X_c)$.
In particular, if $\phi_{ci}$ is based on the treatment pattern of the set of $k$-nearest neighbors of the $i^{\mathrm{th}}$ unit, we get the $k$-nearest neighbors graph assumption, discussed in Example~\ref{ex:2}. In case it maps a treatment pattern to the number of treated in the neighborhood, we get the stratified interference assumption of Example~\ref{ex:1}. The following result gives a simplification of the weights for the weighted projection estimator under the discrete exposure mapping assumption.
\begin{lemma}    
\label{lem:epxposure_map_weights}
    Assume that $\bm \Lambda$ is a correctly specified low-rank structure as described in the beginning of Section~\ref{sec:varying_asymp} and is based on discrete exposure mappings $\{\phi_{ci}\}$ given in Definition~\ref{defn:exposure_mappings}. For the $i^{\mathrm{th}}$ unit in the $c^{\mathrm{th}}$ cluster, define,
    \begin{align*}
        e_{{\phi}_{ci}}(\bm a_c,\bm X_c) = \sum_{\substack{\bm a^*_c \in \{0,1\}^{M_c}:\\ \phi_{ci}(\bm a^*_c) = \phi_{ci}(\bm a_c)}}e(\bm a^*_c, \bm X_c)\textrm{ and } f_{{\phi}_{ci}}(\bm a_c,\bm X_c) = \sum_{\substack{\bm a^*_c \in \{0,1\}^{M_c}:\\ \phi_{ci}(\bm a^*_c) = \phi_{ci}(\bm a_C)}}f(\bm a^*_c, \bm X_c).
    \end{align*}
    Then, we have that
    \[
        \hat{\bm w}^{\sone}_{\bm \Lambda}(e)_{ci} = \frac{f_{\phi_{ci}}(\bm A_c,\bm X_c)}{M_c\cdot e_{\phi_{ci}}(\bm A_c,\bm X_c)}.
    \]
\end{lemma}

The above lemma is intuitive and suggests that our derived weights in this case exactly takes the form of an IPW estimator. However, instead of using the usual propensity score or the counterfactual weights, we use a version where they are marginalized over treatment patterns with the given exposure mapping value as the observed treatment. If our exposure mapping is the treatment pattern of the entire neighborhood, then this reduces to the IPW-type estimator discussed in, for example, \cite{liuhudgensbacker-dreps2016}.

\section{Examples of low-rank structure under linear models}
\label{sec:linear_low_rank_ex}

In this appendix, we discuss a few examples of how low-rank structures, like the ones in Examples~\ref{ex:1}--\ref{ex:4}, can be directly imposed on $\bm \beta$ under linear models. 
\begin{example}[Linear model and stratified interference restricted to the first $k$-neighbors]
\label{ex:6}
The stratified interference assumption under linear models implies that the \textit{effect} of having a given number of treated neighbors, $r$, is the same for all units without depending on $c$ or $i$. This effect is given by a vector $\bm h^{(r)}$ whose length is equal to that of $\bm X_{ci}$, for $0\le r\le k$. Form $\bm h$ by concatenating $\{\bm h^{(r)}\}_{0\leq r\leq k}$, and use $r_{ci}(\bm a_c)$ to denote the number of treated units among the first $k$-neighbors. If we set $\bm \Lambda_{ci}(\bm a_c):=[\mathbb I(r_{ci}(\bm a_c))=0,\cdots, \mathbb I(r_{ci}(\bm a_c) = k)]\otimes \bm I_p$, we have that, $\bm \beta_{ci}(\bm a) = \bm \Lambda_{ci}(\bm a_c)\bm h$. Then, the length of $\bm h$ is $(k+1)p$, and does not vary with $M$.
\end{example}

\begin{example}[Linear model and additive interference]
\label{ex:7}
First, consider the case where all cluster sizes are the same and equal $m\in \mathbb N$. Assume that the units can be enumerated and assigned types 1 through $m$. Then, the restriction can take the following form, $\bm \beta_{ci}^{(\bm a_c)} =\sum_{j=1}^{m}\bm \gamma_j(a_{cj})$.
We interpret $\bm \gamma_j(a_{cj})$ when $a_{cj}\in \{0,1\}$ as the additive contribution of the $j^{\mathrm{th}}$ unit to $\bm \beta_{ci}^{(\bm a_c)}$ based on its treatment status. The function $\bm \gamma_j$ does not depend on $c$ or $i$. Similarly to Example~\ref{ex:3}, this is a low-rank structure on $\bm \beta$, with $\bm h:=(\bm \gamma_1(0)^\top, \bm \gamma_1(1)^\top,\cdots, \bm \gamma_m(0)^\top,\bm \gamma_m(1)^\top)^\top$.  Note that $\bm h$ is of fixed length $2mp$, and does not vary with $M$.

The case with varying cluster sizes requires an additional assumption. In particular, we assume that the cluster sizes are bounded by $s$,  there are $s$ different \textit{types} of units, and each cluster is a non-empty subset of these $s$ types where each type can occur at most once. Then, the $j^{\mathrm{th}}$ type of unit contributes $\bm \gamma_j(a_{cj})$, depending on its treatment status $a_{cj}$. In this case, we can write $\bm \beta_{ci}^{(\bm a_c)} =\sum_{j=1}^{s}q_{cj}\bm \gamma_j(a_{cj})$, where now $q_{cj}$ is the indicator for the presence of the $j^{\mathrm{th}}$ type in the $c^{\mathrm{th}}$ cluster and can be thought of as a covariate of the $c^{\mathrm{th}}$ cluster. This assumption can analogously be represented as a low-rank structure, with $\bm h:=(\bm \gamma_1(0)^\top, \bm \gamma_1(1)^\top,\cdots, \bm \gamma_s(0)^\top,\bm \gamma_s(1)^\top)^\top$, a fixed vector of length $2sp$.
\end{example}

\section{The plug-in OLS estimator}
\label{sec:plug-in}

In this section, we derive an alternative interpretation of our balancing estimator as a regression-based estimator. Under the linear model assumption, the estimand is equal to
$$\mu_f = \mathbb E\left[\frac{1}{M_c}\sum_{i=1}^{M_c}\sum_{\bm a\in \{0,1\}^{M_c}}\bm \Lambda_{ci}(\bm a_c)f(\bm a_c,\bm X_c)\right]^\top\bm h,$$
and the observed data model is $\bm y = \bm R\bm \Lambda\bm h + \bm \epsilon$. Then, the plug-in regression-based estimator of the causal effect is
\begin{align*}
    \hat{T}_{\mathrm{OLS}} 
    = \frac{1}{n}\sum_{i=1}^n\left(\frac{1}{M_c}\sum_{j=1}^{M_c}\sum_{\bm a_c\in \{0,1\}^{M_c}}\bm \Lambda_{ci}(\bm a_c)f(\bm a_c,\bm X_c)\right)^\top(\bm R\bm \Lambda)^+\bm y  = \frac{1}{n}\bm f^\top\bm \Lambda (\bm R\bm \Lambda)^{+}\bm y.
\end{align*}
where $(\bm R\bm \Lambda)^+\bm y$ is the pseudoinverse-based OLS estimate of $\bm h$. Since $\wnulllamb = (\bm \Lambda^\top\bm R^\top)^+\bm \Lambda^\top\bm f$, we have $\hat{T}_{\mathrm{OLS}} = T(\wnulllamb)$. We summarize this observation in a proposition below; its proof follows immediately given the above observation, and hence is skipped.

\begin{prop}
\label{prop:ols_plug_in}
The equality $\hat{T}_{\mathrm{OLS}} = T(\wnulllamb)$ holds. Thus, if Assumptions~\ref{assm:partial}--\ref{assm:outcome_model},~\ref{ass:low_rank}~and~\ref{assm:iid_error} hold, the OLS estimator $\hat{T}_{\mathrm{OLS}}$ is uniformly unbiased for $\mathbb E\left[T(\wnulllamb)\mid B_{\bm \Lambda}\right] = \frac{1}{n}\mathbb E\left[\bm h^\top\bm \Lambda\bm f\mid B_{\bm \Lambda}\right]$ if and only if the optimization problem given in Equation~\eqref{eqn:optimiation_problem} with $\phi = \phinull$ and $\mathcal C_{\bm \Lambda} = \Cnulllamb$ is feasible.
\end{prop}

This proposition shows an equivalence between our balancing estimator and the plug-in OLS estimator. However, even when the plug-in estimator is equivalent to the balancing estimator, the balancing framework is still useful for a number of reasons. Importantly, even though the plug-in OLS estimate can always be obtained, the estimator estimates $\frac{1}{n}\mathbb E\left[\bm h^\top\bm \Lambda\bm f\mid B_{\bm \Lambda}\right]$ without bias (which, as Corollary~\ref{cor:correct_specification_fixed} shows, is asymptotically equal $\mu_f$) only if the balancing equation is feasible. See the end of Section~\ref{sec:lm_balancing} for further justifications.

\section{Feasibility of the balancing equations}
\label{sec:balancing_feasibility}

In this section, we discuss the asymptotic feasibility of the balancing equations. Assume the linear model setting introduced in Example~\ref{ex:5}:
\[
    \bm g_{ci}^{(\bm a_c)}(\bm X_c) = \bm X_{ci}^\top\bm \beta_{ci}(\bm a_c).
\]
Suppose we extend the description in Example~\ref{ex:5} a little bit more and write the low-rank structure on the coefficients as: $\bm \beta = \bm \Lambda^{(\bm \beta)}\bm h$ can be decomposed as $\bm \beta_{ci}(\bm a_c) = (\bm \Lambda_{ci}^{(\bm \beta)}(\bm a_c)\otimes \bm I_p)\bm h, \forall c,i$ (here we assume $\bm h$ to be of length $p\times \ell$). The following assumption summarizes a set of sufficient conditions under which the balancing equations asymptotically admit a solution.
\begin{ass}
    \label{assm:sufficient_feasibility}
    The following conditions are satisfied:
    \begin{enumerate}[label=(\alph*)]
        \item There exists a set $\Gamma = \{\bm \Lambda^{(1)},\cdots, \bm \Lambda^{(r)}\}\subset \mathbb R^{\ell}$ such that, $\bm \Lambda_{ci}^{(\bm \beta)}(\bm a_c)\in \Gamma,\forall c,i,\bm a_c$, and for any $c$, $\mathbb P\left(\bm \Lambda_{c1}^{(\bm \beta)}(\bm A_c) = \bm \Lambda^{(j)} \right)>0$, for all $1\leq j\leq r$.
        \item The distribution of $\bm X_{c1}$ is absolutely continuous.
    \end{enumerate}
\end{ass}
The class of low-rank structure that satisfy Assumption~\ref{assm:sufficient_feasibility}(a) is rich, and include any low-rank structure (with the linear model assumption) that results in $\bm \Lambda_{ci}^{(\bm \beta)}(\bm a_c)$ become a binary vector. This includes Examples~\ref{ex:1}--\ref{ex:3} and also the types of structures introduced in Example~\ref{ex:4}. If $\bm \Lambda_{ci}^{(\bm \beta)}(\bm a_c)$ is a binary vector, then it can have at most $2^{\ell}$ different possible choices, thereby satisfying Assumption~\ref{assm:sufficient_feasibility}(a). Also note that, we have deliberately used only the unit-level index 1 in stating the condition $\mathbb P\left(\bm \Lambda_{c1}^{(\bm \beta)}(\bm A_c) = \bm \Lambda^{(r)} \right)>0$. In particular, we only need that there is one unit per cluster, which has non-zero porbability of realizing every low-rank structure in $\bm \Gamma$.
Assumption~\ref{assm:sufficient_feasibility}(b) is not too restrictive as it only requires that at least one of the unit-level covariates in $\bm X_c$ have an absolutely continuous distribution. 

Finally, the following theorem shows that this and other mild assumptions, the balancing equations admit a solution with probability that converges to 1, enabling us to invoke the asymptotic theory established in Section~\ref{sec:asymptotics} for our balancing estimators.

\begin{theorem}
    \label{thm:bal_feasible}
    Suppose under the above linear model assumption, Assumptions~\ref{assm:pos}, \ref{assm:super_pop} and \ref{assm:sufficient_feasibility} are satisfied. Then, Assumption~\ref{assm:balancing_feasibility} is also satisfied. 
\end{theorem}

\section{Proofs}
\subsection{Proof of Theorem  \ref{thm:subclass1}}
\label{sec:proof_subclass1}
\begin{proof}
If $\bm w \in \mathcal C(\bm g, e)$ for all $\bm g$, then this statement holds true in particular for, $\bm g = \bm g(\bm X) = \mathbb E\left[\bm R^\top(\bm w - \bm w_{\mathrm{IPW}})\mid \bm X\right]$. For this choice of $\bm g$, we have that 
\begin{align*}
    &\mathbb E[(\bm w - \bm w_{\mathrm{IPW}})^\top\bm R\bm g]= \mathbb E\left[\mathbb E[(\bm w - \bm w_{\mathrm{IPW}})^\top\bm R\bm g\mid \bm X]\right]= \mathbb E\left[\left\|\mathbb E[\bm R^\top(\bm w - \bm w_{\mathrm{IPW}})\mid\bm X]\right\|^2\right] = 0.
\end{align*}
This implies that the following equalities hold almost surely,
\begin{align*}
    & \mathbb E[\bm R^\top(\bm w - \bm w_{\mathrm{IPW}})\mid \bm X] = \mathbb E[\bm R^\top\bm R(\tilde{\bm w} - \tilde{\bm w}_{\mathrm{IPW}})\mid \bm X] = \bm 0,\\
    \implies & \bm D(\tilde{\bm w} - \tilde {\bm w}_{\mathrm{IPW}}) = \bm 0,
\end{align*}
where $\bm D = \mathbb E[\bm R^\top\bm R\mid \bm X]$ denotes the diagonal matrix of propensity scores. If we focus on the entry of $\bm D$ corresponding to the assigned treatment $\bm A$, then the above equalities imply
\[
    e(\bm A, \bm X)(\bm w - \bm w_{\mathrm{IPW}}) = \bm 0 \ \textrm{a.s.},
\]
where $e(\bm a, \bm X):=\mathbb P(\bm A_1 = \bm a_1, \cdots, \bm A_n = \bm a_c\mid \bm X)$. Note that because $A$ is a discrete random variable taking values in a finite space, we must have that $e(\bm A, \bm X)>0$ almost surely. This shows $\bm w = \bm w_{\mathrm{IPW}}$ almost surely,
thereby proving the first clause of Theorem~\ref{thm:subclass1}.
    
Next, the condition $\bm w \in \mathcal C(\bm g, e)$ can also be expressed as $\mathbb E\left[(\bm w - \bm w_{\mathrm{IPW}})^\top\bm R\bm g\right] = 0$, which implies $\mathbb E[\bm w^\top\bm R\bm g] = \mathbb E[\bm w_{\mathrm{IPW}}^\top\bm R\bm g]$.
Observe that $\mathbb E[\bm w_{\mathrm{IPW}}^\top\bm R\mid\bm X] = \bm f^\top$ so that the above can be written as $\mathbb E[\bm w^\top\bm R\bm g] = \mathbb E[\bm f^\top\bm g]$.  If this equality were to hold for all choices of $\bm g$, then an argument similar to the one made above shows that we must have $\bm R^\top\bm w = \bm f$, almost surely. Now, note that $\bm R^\top\bm w = \left[\mathbb I(\bm A = \bm a^{(0)}),\cdots,\mathbb I(\bm A = \bm a^{({2^M}-1)})\right]\otimes \bm w$ so that $\bm R^\top\bm w = \bm f$ implies $f(\bm a_c,\bm X_c) = \bm 0, \forall \bm a_c \neq \bm A_c,\forall c$. This contradicts the assumption in \eqref{eqn:impossibility_condition} and hence, completes the proof.
\end{proof}

\subsection{Proof of Theorem \ref{thm:subclass2}}
\label{sec:proof_subclass2}

\begin{proof}
Assume that $\bm w \in \mathcal C_{\bm \Lambda}(\bm h, e),\forall \bm h, e$. We follow the proof in Section~\ref{sec:proof_subclass1} above while replacing the matrix $\bm R$ with $\bm R\bm \Lambda$, implying that for any choice of $e$, the following equalities hold almost surely,
\begin{align}
\label{eqn:classequations}
\begin{aligned}
    &\bm \Lambda^\top\bm D(\tilde {\bm w} - \tilde {\bm w}_{\mathrm{IPW}}) = \bm 0
    \implies &\bm \Lambda^\top\bm D\tilde{\bm w} = \bm \Lambda^\top\bm f,
\end{aligned}
\end{align}
where $\tilde{\bm w}$ is as defined in Section~\ref{sec:extension_known_propensity}. Now, fix any treatment pattern $\bm a^\ast$ and consider the propensity score, $e$, which sets $e(\bm a^\ast, \bm X) = 1-1/m$, for $m\in \mathbb N$. Because $\bm w$, and hence, $\tilde{\bm w}$, does not depend on the true propensity score, any particular choice of $e$ does not alter $\tilde{\bm w}$. Then, under this family of scores for different values of $m$, taking $m\rightarrow\infty$ in the above equation yields, 
\begin{align*}
    \bm \Lambda(\bm a^\ast)^\top\tilde {\bm w}(\bm a^\ast) = \bm \Lambda^\top\bm f,\textrm{ almost surely}.
\end{align*}
Since the above holds for any treatment pattern $\bm a^*$ and in particular for the observed pattern $\bm A$, by substituting $\bm a^* = \bm A$, we obtain $\bm \Lambda(\bm A)^\top \bm w = \bm \Lambda^\top\bm f$, which holds almost surely.  This in turn implies that $\bm R^\top\bm \Lambda^\top w = \bm \Lambda^\top\bm f$ holds almost surely,
thereby proving the claim. 
\end{proof}

\subsection{Proof of Theorem~\ref{thm:extension_subclass2}}
\label{sec:proof_extension_subclass2}
\begin{proof}

Note that the fact that $\bm w \in \mathcal C_{\bm \Lambda}(\bm h, e),\forall \bm h\implies \bm w \in \Cipwlamb(e)$ follows directly from the first conclusion of Equation~\eqref{eqn:classequations}.  The other direction of implication also follows almost surely, 
\begin{align*}
    &\bm w \in \Cipwlamb(e)\\
    \Leftrightarrow & \bm \Lambda^\top\bm D\left(\tilde{\bm w} - \tilde{\bm w}_{\mathrm{IPW}}\right) = \bm 0, \\
    \Leftrightarrow& \mathbb E[\bm \Lambda^\top\bm R^\top\bm R\left(\tilde{\bm w} - \tilde{\bm w}_{\mathrm{IPW}})\mid \bm X\right] = \bm 0,\\
    \Leftrightarrow & \mathbb E[\bm \Lambda^\top\bm R^\top(\bm w - \bm w_{\mathrm{IPW}})\mid \bm X] = \bm 0, \\
    \implies  & \mathbb E[\bm h^\top\bm \Lambda^\top\bm R^\top(\bm w - \bm w_{\mathrm{IPW}})\mid \bm X] = \bm 0,\\
    \Leftrightarrow & \mathbb E[\bm g^\top\bm R^\top(\bm w - \bm w_{\mathrm{IPW}})\mid \bm X] = \bm 0, \\
    \implies & \mathbb E[\bm g^\top\bm R^\top(\bm w - \bm w_{\mathrm{IPW}})] = \bm 0,
\end{align*}
thereby showing that $T(\bm w)$ is unbiased for $\mu_f$.
\end{proof}

\subsection{Proof of Theorem~\ref{thm:umvue}}
\label{sec:proof_thm_umvue}

\begin{proof}
First note that for any choice of weights, $\bm w$, we have that
\begin{align}
\label{eqn:var_decompose}
\begin{aligned}
    \Var(T(\bm w)\mid \bm B_{\bm \Lambda}) &= \frac{1}{n^2}\left(\mathbb E\left[\Var(\bm w^\top\bm Y\mid \bm X, \bm A)\mid \bm B_{\bm \Lambda}\right] + \Var\left[\mathbb E(\bm w^\top\bm Y\mid \bm X, \bm A)\mid \bm B_{\bm \Lambda}\right]\right)\\
    &= \frac{1}{n^2}\left(\sigma^2\mathbb E\left[\|\bm w\|^2\mid \bm B_{\bm \Lambda}\right] + \Var\left(\bm w^\top\bm R\bm \Lambda\bm h\mid \bm B_{\bm \Lambda}\right)\right).
\end{aligned}
\end{align}
Now, if $\bm w \in \Cnulllamb$, by definition, we have that $\bm w^\top\bm R\bm \Lambda = \bm f^\top\bm \Lambda = \mathbb E\left[\bm w_{\mathrm{IPW}}^\top\bm R\bm \Lambda\bm h\mid \bm X\right]$, from which the first claim of Theorem~\ref{thm:umvue} follows.

Because, by definition, $\|\wnulllamb\|^2\leq \|\bm w\|^2$, almost surely, for any $\bm w \in \Cnull$, we must have that for any $\bm w \in \Cnulllamb$,
\begin{align*}
    n^2\Var\left(T(\bm w)\mid \bm B_{\bm \Lambda}\right)
    =& \Var\left(\bm f^\top\bm \Lambda\bm h\mid \bm B_{\bm \Lambda}\right) + \sigma^2 \mathbb E\left[\|\bm w\|^2\mid \bm B_{\bm \Lambda}\right]\\
    \geq & \Var\left(\bm f^\top\bm \Lambda\bm h\mid \bm B_{\bm \Lambda}\right) + \sigma^2 \mathbb E\left[\|\wnulllamb\|^2\mid \bm B_{\bm \Lambda}\right]\\
    =& n^2\Var\left(T(\wnulllamb)\mid \bm B_{\bm \Lambda}\right),
\end{align*}
thereby establishing the second claim.
\end{proof}

\subsection{Proof of Theorem \ref{thm:weights_expression}}
\label{sec:proof_thm_weights_expression}
\begin{proof}
    we use the fact that the solution to the problem: Minimize: $\|\bm x\|^2$, subject to $\bm A \bm x = \bm b$ is given by $\bm x = \bm A^+\bm b$. This immediately establishes that $\wnulllamb = \left(\bm \Lambda^\top\bm R^\top\right)^+\bm \Lambda^\top\bm f$.
\end{proof}

\subsection{Proof of Theorem~\ref{thm:propensity_estimators_fixed_h}}
\label{sec:proof_propensity_estimators_fixed_h}

\begin{proof}
Similar to the proof in Section~\ref{sec:proof_thm_weights_expression}, we have $\wipwprojlamb(e) = \left(\bm \Lambda^\top\bm R^\top\right)^+\bm \Lambda^\top\bm R^\top\bm w_{\mathrm{IPW}} = \bm P_{\bm R\bm \Lambda}\bm w_{\mathrm{IPW}}$, where we have used the fact that $\bm A^+\bm A$ is the projection matrix onto the row-space of $\bm A$, and hence, onto the column-space of $\bm A^\top$. Finally, to derive the expression for $\wipwlamb(e)$, note that the optimization problem determining it can be written as
\[
    \textrm{Minimize: }\|\bm w_*\|^2\hspace{0.5cm}\textrm{subject to }\bm \Lambda_*\bm w_* = \bm \Lambda_*\bm D^{1/2}\tilde{\bm w}_{\mathrm{IPW}},
\]
where, $\bm \Lambda_* = \bm D^{1/2}\bm \Lambda$ and $\bm w_* = \bm D^{1/2}\tilde{\bm w}_{\mathrm{IPW}}$. Then, the discussion above implies that the solution to this optimization problem is given by $\tilde{\bm w}_* = \bm P_{\bm \Lambda_*}\bm D^{1/2}\tilde{\bm w}_{\mathrm{IPW}}$. But also note that $\wipwlambtilde(e) = \bm D^{-1/2}\tilde{\bm w}_* = \bm D^{-1/2}\bm P_{\bm \Lambda_*}\bm D^{1/2}\tilde{\bm w}_{\mathrm{IPW}} = \bm D^{-1/2}\bm P_{\bm D^{1/2}\bm \Lambda}\bm D^{1/2}\tilde{\bm w}_{\mathrm{IPW}}$, from which the result follows.
Finally, one can use the approach in Section~\ref{sec:proof_thm_umvue} to derive the expression for $\Var(T(\wipwprojlamb(e))$ and show that it is the smallest among $\Var(T(\bm w))$, for any $\bm w \in \Cipwprojlamb(e)$.
\end{proof}

\subsection{Proof of Theorem \ref{thm:ipw_normality}}
\label{sec:proof_thm_ipw_normality}

\begin{proof}
Note that 
\begin{align*}
    \mathbb E[\bm w_{\mathrm{IPW},c}^\top\bm Y_c] &=\mathbb E\left[\sum_{\bm a_c\in \{0,1\}^{M_c}}\bm w_{\mathrm{IPW},c}(\bm a_c)^\top\bm Y_c(\bm a_c)\right] \\
    & = \mathbb E\left[\frac{1}{M_c}\sum_{\bm a_c\in \{0,1\}^{M_c}}\frac{f(\bm a_c, \bm X_c)}{e(\bm a_c, \bm X_c)}\bm 1^\top\bm Y_c(\bm a_c)e(\bm a_c, \bm X_c)\right] \\
    & = \mu_f.
\end{align*}
Now, $T(\bm w_{\mathrm{IPW}}) = \frac{1}{n}\sum_{i=1}^n\bm w_{\mathrm{IPW},i}^\top\bm Y_i$ is a sample mean of i.i.d. terms. Furthermore, note $\bm w_{\mathrm{IPW},c}^\top\bm Y_c = {\bm g_c^{(\bm A_c)}}^\top\bm w_{\mathrm{IPW},c} + \bm \epsilon_c^\top\bm w_{\mathrm{IPW},c}$. Then, we have
\begin{align*}
\mathbb E\left(\bm w_{\mathrm{IPW},c}^\top\bm Y_c\right)^2= & \mathbb E\left({\bm g_c^{(\bm A_c)}}^\top\bm w_{\mathrm{IPW},c}\right)^2 + \mathbb E\left(\bm \epsilon_c^\top\bm w_{\mathrm{IPW},c}\right)^2 \\
\leq & \sqrt{\mathbb E \left\| \bm g_c^{(\bm A_c)}\right\|^4\cdot\mathbb E \|\bm w_{\mathrm{IPW},c}\|^4} + \sigma^2\mathbb E\left\|\bm w_{\mathrm{IPW},c}\right\|^2\\
< & \infty,
\end{align*}
due to the Cauchy-Schwartz inequality. Thus, the central limit theorem applies and yields 
\begin{align*}
    \sigma^2_{\mathrm{IPW}} &= \Var(\bm w_{\mathrm{IPW},c}^\top\bm Y_c)\\
    &= \Var\mathbb E[ \bm w_{\mathrm{IPW},c}^\top\bm Y_c\mid \bm X_c] + \mathbb E\Var[ \bm w_{\mathrm{IPW},c}^\top\bm Y_c\mid \bm X_c]\\
    &= \Var\left[{\bm g_c^{(\bm A_c)}}^\top\bm w_{\mathrm{IPW},c}\right] + \sigma^2 \mathbb E\left[\left\|\bm w_{\mathrm{IPW},c}\right\|^2\right].
\end{align*}
\end{proof}

\subsection{Proof of Theorem \ref{thm:clt_indiv}}
\label{sec:proof_thm_clt_indiv}

\begin{proof}
Observe that under the conditions of the theorem, our weighting estimator is given by
\[
    T\left(\wipwlamb(e)\right) = \frac{1}{n}\sum_{c=1}^n\sum_{i=1}^{M_c} \wipwlambnobm(e)_{ci}Y_{ci} = \frac{1}{n}\sum_{c=1}^n \wipwlamb(e)_c^\top\bm Y_c.
\]
We intend to apply the central limit theorem. As in Section~\ref{sec:proof_thm_ipw_normality}, it suffices to show  $$\mathbb E\left(\wipwlamb(e)_c^\top\bm Y_c\right)^2<\infty.$$ Again, note $$\mathbb E\left(\wipwlamb(e)_c^\top\bm Y_c\right)^2 = \mathbb E\left[\left({\bm g_c^{(\bm A_c)}}^\top\wipwlamb(e)_c\right)^2 + \left(\bm \epsilon_c^\top\wipwlamb(e)_c\right)^2\right].$$ 
Following the same approach as in Section~\ref{sec:proof_thm_ipw_normality}, this quantity is bounded by $$\sqrt{\mathbb E \left\| \bm g_c^{(\bm A_c)}\right\|^4\cdot\mathbb E \|\wipwlamb(e)_c\|^4} + \sigma^2\mathbb E\left\|\wipwlamb(e)\right\|^2.$$
Hence, in light of Assumption~\ref{assm:clt_technical_indiv}, it suffices to show that $\mathbb E \|\wipwlamb(e)_c\|^4<\infty$. 
Define, $\bm {\bm u}_{ci} = \bm P_{\bm D_c^{1/2}\bm \Lambda_{ci}}\bm D_c^{1/2}\tilde{\bm w}_{{\rm IPW},ci}$, and note that
\begin{align*}
    &\|\wipwlamb(e)_c\|^4
    = \left(\sum_{i=1}^{M_c}\wipwprojlamb(e)^2_{c,i}\right)^2
    = \left(\sum_{i=1}^{M_c}\left(\bm R_c \bm D_c^{-1/2}\bm {\bm u}_{ci}\right)^2\right)^2
    = \left(\sum_{i=1}^{M_c}\left(\frac{\bm R_c \bm {\bm u}_{ci}}{\sqrt{e(\bm A_c; \bm X_c)}}\right)^2\right)^2\\
    \leq& \eta^{-2} \left(\sum_{i=1}^{M_c}\left(\bm R_c \bm {\bm u}_{ci}\right)^2\right)^2,
\end{align*}
where the last inequality follows due to item (b) of Assumption~\ref{assm:clt_technical_indiv}. Now let $ {\bm u}_{ci}(\bm a_c)$ denote the value of $\bm R_c \bm {\bm u}_{ci}$ if $\bm A_c = \bm a_c$. Then, we have that 
\begin{align*}
    &\mathbb E\left[\|\wipwlamb(e)_c\|^4\middle| \bm X_c\right]\\
    &\leq \eta^{-2} \mathbb E\left[ \left(\sum_{i=1}^{M_c}\left(\bm R_c \bm {\bm u}_{ci}\right)^2\right)^2\middle| \bm X_c\right] = \eta^{-2}\sum_{\bm a_c \in \{0,1\}^{M_c}} e(\bm a_c;\bm X_c) \left(\sum_{i=1}^{M_c}{\bm u}_{ci}(\bm a_c)^2\right)^2\\
    &\leq \eta^{-2} \sum_{\bm a_c \in \{0,1\}^{M_c}}\left(\sum_{i=1}^{M_c}{\bm u}_{ci}(\bm a_c)^2\right)^2 \leq \eta^{-2}\left(\sum_{\bm a_c \in \{0,1\}^{M_c}}\sum_{i=1}^{M_c}{\bm u}_{ci}(\bm a_c)^2\right)^2 = \eta^{-2}\left(\sum_{i=1}^{M_c}\|\bm {\bm u}_{ci}\|^2\right)^2\\
    &= \eta^{-2}\left(\sum_{i=1}^{M_c}\left\|\bm P_{\bm D_c^{1/2}\bm \Lambda_{ci}}\bm D_c^{1/2}\tilde{\bm w}_{{\rm IPW},ci}\right\|^2\right)^2 \leq \eta^{-2}\left(\sum_{i=1}^{M_c}\left\|\bm D_c^{1/2}\tilde{\bm w}_{{\rm IPW},ci}\right\|^2\right)^2 \leq \eta^{-2}\left(\sum_{i=1}^{M_c}\left\|\tilde{\bm w}_{{\rm IPW},ci}\right\|^2\right)^2.
\end{align*}
where the last inequality follows from the fact that the diagonal entries of $\bm D_c^{1/2}$ are square roots of propensity scores and hence are bounded above by 1. Item (a) of Assumption~\ref{assm:clt_technical_indiv} now shows that $\mathbb E\left[\|\wipwlamb(e)_c\|^4\right]\leq \eta^{-2} \mathbb E\left[\left(\sum_{i=1}^{M_c}\left\|\tilde{\bm w}_{{\rm IPW},ci}\right\|^2\right)^2\right]<\infty$, thereby proving the claim. This allows us to invoke the central limit theorem. The expressions for the asymptotic mean and variance follow from their direct evaluation on the quantity $\wipwlamb(e)_c^\top\bm Y_c$.
\end{proof}

\subsection{Proof of Theorem \ref{thm:known_propensity_clt_fixed}}
\label{sec:proof_thm_known_propensity_clt_fixed}

\begin{proof}
For simplicity, we establish the theorem only for the projection estimator. These same steps analogously establish the result for the weighted projection estimator as well; we give a proof sketch for that at the end in Section~\ref{sec:proof_sketch_weighted_projection}

The optimization problem furnishing the weights for the projection estimator is given by
\begin{align*}
    \textrm{Minimize: } \ \frac{1}{2n}\sum_{c=1}^n\|\bm w_c\|^2 \quad
    \textrm{subject to } \ \frac{1}{n}\sum_{c=1}^n\bm \Lambda^\top_c(\bm A)\bm w_c = \frac{1}{n}\sum_{c=1}^n \bm v_c,
\end{align*}
where $\bm v_c = \bm \Lambda_c(\bm A_c)^\top\bm{w}_{\mathrm{IPW},c}$.
The Lagarangian for the above optimization problem is:
\begin{align*}
    \mathcal L(\bm \lambda,\bm w) = \frac{1}{2n}\sum_{c=1}^n\|\bm w_c\|^2 + \bm \lambda^\top\left(\frac{1}{n}\sum_{c=1}^n \bm \Lambda^\top_c(\bm A_c)\bm w_{c} - \frac{1}{n}\sum_{c=1}^n \bm v_c\right),
\end{align*}
where $\bm \lambda$ denotes the Lagarange multiplier. To obtain the dual of the problem, we maximize $\mathcal L(\bm \lambda, \bm w)$ with respect to $\bm w$. The first order conditions give
\begin{align}
    \label{eqn:weights_first_order}
    \begin{aligned}
    \nabla_{\bm w_c}\mathcal L = 0
    \implies & \bm w_c +\bm \Lambda_c(\bm A_c)\bm \lambda=\bm 0\\
    \implies & \bm w_c = -\bm \Lambda_c(\bm A_c)\bm \lambda.
    \end{aligned}
\end{align}
Substituting this value of $\bm w_c$ back into the expression of $\mathcal L(\bm \lambda, \bm w)$, we get the dual problem:
\begin{align*}
    \hat{\bm \lambda} &= \underset{\bm \lambda}{\arg\max}\left(-\frac{1}{2n}\bm \lambda^\top\sum_{c=1}^n \bm \Lambda_c(\bm A_c)^\top\bm \Lambda_c(\bm A_c)\bm{\lambda} - \frac{1}{n}\bm \lambda^\top\sum_{c=1}^n \bm v_c\right)\\
    &= -\left(\frac{1}{n}\sum_{c=1}^n  \bm \Lambda_c(\bm A_c)^\top\bm \Lambda_c(\bm A_c)\right)^{-1} \left(\frac{1}{n}\sum_{c=1}^n \bm v_c\right).
\end{align*}
Then, by duality theory, the optimal weights are given by $\hat{\bm w}_c = -\bm \Lambda_c(\bm A_c)\hat{\bm \lambda}$. For notational simplicity, we have not made the dependence of $\bm \lambda$ and $\hat{\bm w}$ on $j$ explicit. Invoking the first-order conditions again, we see that the condition for dual maximization is equivalent to requiring
\begin{align*}
    \frac{1}{n}\sum_{c=1}^n \bm \Lambda_c(\bm A_c)^\top\bm \Lambda_c(\bm A_c)\hat{\bm \lambda} = -\frac{1}{n}\sum_{c=1}^n \bm v_c.
\end{align*}
Now, consider the weighting estimator given by
\begin{align}
    T(\hat{\bm w}) = \frac{1}{n}\sum_{c=1}^n \hat{\bm w}_c^\top\bm Y_c = \underbrace{-\hat{\bm \lambda}^\top\frac{1}{n}\sum_{c=1}^n \bm \Lambda_c(\bm A_c)^\top\bm \Lambda_c^*(\bm A_c)\bm h}_{=:T_1}  \underbrace{-\hat{\bm \lambda}^\top\frac{1}{n}\sum_{c=1}^n \bm \Lambda_c(\bm A_c)^\top\bm \epsilon_c}_{=:T_2}. \label{eq:Tterms}
\end{align}
In order to establish a CLT for this estimator, we first start with some observations. To begin with, define,
\[
    m_c(\bm \lambda) := -\frac{1}{2}\bm \lambda^\top\bm \Lambda_c(\bm A_c)^\top\bm \Lambda_c(\bm A_c)\bm\lambda - \bm \lambda^\top \bm v_c,
\]
which also gives:
\begin{align*}
    &\bm m_c'(\bm \lambda) = -\bm \Lambda_c(\bm A_c)^\top\bm \Lambda_c(\bm A_c)\bm \lambda - \bm v_c,\\
    &\bm m_c''(\bm \lambda) = -\bm \Lambda_c(\bm A_c)^\top\bm \Lambda_c(\bm A_c).
\end{align*}
Under Assumption~\ref{assm:clt_technical_fixed}, first define $\bm \lambda^* = -\mathbb E[\bm \Lambda_c(\bm A_c)^\top\bm \Lambda_c(\bm A_c)]^{-1}\mathbb E[\bm v_c]$ and $\bm V:= \mathbb E[\bm \Lambda_c(\bm A_c)^\top\bm \Lambda_c(\bm A_c)]$. Observe that $\mathbb E[\bm m_c'(\bm \lambda^*)] = \bm 0$. Next, we show that $\mathbb E[\|\bm m_c'(\bm \lambda^*)\|^2]<\infty$. Note $\bm v_c = \bm \Lambda_c(\bm A_c)^\top\bm w_{\mathrm{IPW},c}$, and define $\bm u_c = \bm v_c/\mathbb E\|\bm v_c\|]$. Then,
\begin{align*}
    \mathbb E[\|\bm m_c'(\bm \lambda^*)\|^2]&\leq \mathbb E[\|\bm \Lambda_c(\bm A_c)^\top\bm \Lambda_c(\bm A_c)\bm \lambda^*\|^2] + \mathbb E[\|\bm v_c\|^2]\\
    &\leq (\bm \lambda^*)^\top\mathbb E\left[\left(\bm \Lambda_c(\bm A_c)^\top\bm \Lambda_c(\bm A_c)\right)^2\right]\bm \lambda^* + \mathbb E[\|\bm v_c\|^2]\\
    &=\mathbb E[{\bm v_c}^\top]\mathbb E\left[\bm V^{-1}\left(\bm \Lambda_c(\bm A_c)^\top\bm \Lambda_c(\bm A_c)\right)^2\bm V^{-1}\right]\mathbb E[\bm v_c] + \mathbb E[\|\bm v_c\|^2]\\
    &=\mathbb E[\|\bm v_c\|^2]\left(\mathbb E[{\bm u_c}^\top]\mathbb E\left[\bm V^{-1}\left(\bm \Lambda_c(\bm A_c)^\top\bm \Lambda_c(\bm A_c)\right)^2\bm V^{-1}\right]\mathbb E[\bm u_c] + 1\right)\\
    &\leq \mathbb E[\|\bm \Lambda_c(\bm A_c)^\top\bm w_{\mathrm{IPW},c}\|^2]\left(\mathbb E{[\bm u_c}^\top]\mathbb E\left[\bm V^{-1}\left(\bm \Lambda_c(\bm A_c)^\top\bm \Lambda_c(\bm A_c)\right)^2\bm V^{-1}\right]\mathbb E[\bm u_c] + 1\right)\\
    &\leq \underbrace{\mathbb E\left[\lambda_{\max}\left(\bm \Lambda_c(\bm A_c)^\top\bm \Lambda_c(\bm A_c)\right)\cdot\|\bm w_{\mathrm{IPW},c}\|^2\right]}_{\mathrm{\textbf{I}}} \\
    & \quad \times \left(\underbrace{\mathbb E[{\bm u_c}^\top]\mathbb E\left[\bm V^{-1}\left(\bm \Lambda_c(\bm A_c)^\top\bm \Lambda_c(\bm A_c)\right)^2\bm V^{-1}\right]\mathbb E[\bm u_c]}_{\mathrm{\textbf{II}}} + 1\right)
\end{align*}
The first term of the above upper bound can be further bounded as follows,
\begin{align*}
    \mathrm{term\ \textbf{I}} &\leq \mathbb E\left[\lambda_{\max}\left(\bm \Lambda_c(\bm A_c)^\top\bm \Lambda_c(\bm A_c)\right)\cdot\|\bm w_{\mathrm{IPW},c}\|^2\right]\\
    &= \mathbb E\left[\sigma^2_{\max}\left(\bm \Lambda_c(\bm A_c)\right)\cdot \|\bm w_{\mathrm{IPW},c}\|^2\right]\\
    &\leq \sqrt{\mathbb E\left[\sigma^4_{\max}\left(\bm \Lambda_c(\bm A_c)\right)\right]\mathbb E\left[ \|\bm w_{\mathrm{IPW},c}\|^4\right]}\\
    &<\infty
\end{align*}
where the second inequality follows from the Cauchy-Schwartz inequality and the final inequality is due to Assumption~\ref{assm:clt_technical_fixed}.
Next, we show that the second term is finite. Note that if $\bm \lambda^*\neq \bm 0$, as $\|\mathbb E[\bm u_c]\|\leq 1$, we have,
\begin{align*}
    \mathrm{term\ \textbf{II}} &=\mathbb E[{\bm u_c}^\top] \mathbb E\left[\bm V^{-1}\left(\bm \Lambda_c(\bm A_c)^\top\bm \Lambda_c(\bm A_c)\right)^2\bm V^{-1}\right]\mathbb E[\bm u_c]\\
    &\leq \mathbb E\left[\lambda_{\max}\left(\bm V^{-1}\left(\bm \Lambda_c(\bm A_c)^\top\bm \Lambda_c(\bm A_c)\right)^2\bm V^{-1}\right)\right]\\
    &=\mathbb E\left[\sigma^2_{\max}\left(\bm \Lambda_c(\bm A_c)^\top\bm \Lambda_c(\bm A_c)\bm V^{-1}\right)\right]\\
    &\leq \mathbb E\left[\sigma^4_{\max}\left(\bm \Lambda_c(\bm A_c)\right)\sigma^2_{\max}\left(\bm V^{-1}\right)\right]\\
    &<\infty
\end{align*}
where the second inequality follows from the sub-multiplicativity of the spectral norm and the final inequality is due to Assumption~\ref{assm:clt_technical_fixed}.  This shows that if $\bm \lambda^*\neq \bm 0$, then $\mathbb E[\|\bm m_c'(\bm \lambda^*)\|^2]<\infty$. In case $\bm \lambda^*=\bm 0$, then we have that $\mathbb E[\|\bm m_c'(\bm \lambda^*)\|^2] = \mathbb E[\|\bm v_c\|^2]$, which equals term \textbf{I} above, and hence is finite. Finally, since $\bm m_c''(\bm \lambda)$ does not depend on $\bm \lambda$, the second order partial derivatives of $\bm m_c'(\bm \lambda)$ are all equal to 0. 

Now, let us define $\bar {\bm V}_n = \frac{1}{n}\sum_{c=1}^n \bm \Lambda_c(\bm A_c)^T\bm \Lambda_c(\bm A_c)$, and $\bar{\bm v}_n = \frac{1}{n}\sum_{c=1}^n \bm v_c$. Then note that,
\begin{align*}
    \hat{\bm \lambda} - \bm \lambda^* &= \bar{\bm V}_n^{-1}\left(-\bar{\bm V}_n\bm\lambda^* - \bar{\bm v}_n\right) =  \bar{\bm V}_n^{-1}\frac{1}{n}\sum_{c = 1}^n \bm m_c'(\bm \lambda^*) \\
    &= \mathbb E[\bm \Lambda_c(\bm A_c)^\top\bm \Lambda_c(\bm A_c)]^{-1}\frac{1}{n}\sum_{c = 1}^n \bm m_c'(\bm \lambda^*) + \left(\bar{\bm V}_n^{-1} - \mathbb E[\bm \Lambda_c(\bm A_c)^\top\bm \Lambda_c(\bm A_c)]^{-1} \right)\frac{1}{n}\sum_{c = 1}^n \bm m_c'(\bm \lambda^*).
\end{align*}
Next, note that $\mathbb E\left[\bm m_c'(\bm \lambda^*)\right] = 0$, so that $\frac{1}{n}\sum_{c = 1}^n \bm m_c'(\bm \lambda^*) = O_{\mathbb P}(n^{-1/2})$. Because $\bar{\bm V}_n^{-1}\Pto \mathbb E[\bm \Lambda_c(\bm A_c)^\top\bm \Lambda_c(\bm A_c)]^{-1}$, we get that:
\[
    \hat{\bm \lambda} = \bm \lambda^* + \mathbb E[\bm \Lambda_c(\bm A_c)^\top\bm \Lambda_c(\bm A_c)]^{-1}\frac{1}{n}\sum_{c=1}^n m_c'(\bm \lambda^*) + o_{\mathbb P}(n^{-1/2}).
\]

We get back to establishing the CLT for $T(\hat{\bm w})$. In order to establish the distribution of term $T_1$ in Equation~\eqref{eq:Tterms}, we first start with the following CLT: Define 
\[
    \bm W_c = \begin{pmatrix}\bm X_c \\ \bm Y_c \end{pmatrix} = \begin{pmatrix}
        \bm \Lambda_c(\bm A_c)^\top\bm \Lambda_c^*(\bm A_c)\bm h \\
        -\bm \Lambda_c(\bm A_c)^\top\bm \Lambda_c(\bm A_c)\bm \lambda^* - \bm v_c
    \end{pmatrix} =  \begin{pmatrix}
        \bm \Lambda_i(\bm A_c)^\top\bm \Lambda_c^*(\bm A_c)\bm h \\
        \bm \Lambda_c(\bm A_c)^\top\bm \Lambda_c(\bm A_c)\bm V^{-1}\mathbb E \bm v_c - \bm v_c
    \end{pmatrix}
\]
The variance of $\bm W_c$ exists because the variance of $\bm Y_c$ exists, which in turn follows from the above result that term \textbf{II} is finite. We can also show that the variance of $\bm X_c$ exists,
\begin{align*}
    \mathbb E[\bm X_c^\top\bm X_c] &= \mathbb E\left[\bm h^\top\left(\bm \Lambda_c(\bm A_c)^\top\bm \Lambda_c^*(\bm A_c)\right)^\top\left(\bm \Lambda_c(\bm A_c)^\top\bm \Lambda_c^*(\bm A_c)\right)\bm h\right]\\
    &\leq\|\bm h\|^2 \mathbb E\left[\lambda_{\max}\left\{ \left(\bm \Lambda_c(\bm A_c)^\top\bm \Lambda_c^*(\bm A_c)\right)^\top\left(\bm \Lambda_c(\bm A_c)^\top\bm \Lambda_c^*(\bm A_c)\right)\right\}\right]\\
    &\leq \|\bm h\|^2 \mathbb E\left[\sigma_{\max}^2\left(\bm \Lambda_c(\bm A_c)^\top\bm \Lambda_c^*(\bm A_c)\right)\right]\\
    &<\infty.
\end{align*}
where the final inequality follows from Assumption~\ref{assm:clt_technical_fixed}. 

Denote $\mathrm{Var}(\bm W_c) = \bm \Sigma$ and $\mathbb E[\bm W_c] = \bm\mu$. Now, consider term $T_2$ in Equation~\eqref{eq:Tterms}
\[
    T_2 = \hat{\bm \lambda}^\top\frac{1}{n}\sum_{c=1}^n\bm \Lambda_c(\bm A_c)^\top\bm \epsilon_c.
\]
First, note that under Assumption \ref{assm:clt_technical_fixed}, we have that by the CLT,
\[
    \frac{1}{\sqrt{n}}\sum_{c=1}^n\bm \Lambda_c(\bm A_c)^\top\bm \epsilon_c\dto \mathcal N\left(\bm 0, \mathbb E\left[(\bm \Lambda_c(\bm A_c)^\top\bm \epsilon_c)^2\right]\right) = \mathcal N\left(\bm 0, \sigma^2 \mathbb E\left[\bm \Lambda_c(\bm A_c)^\top\bm \Lambda_c(\bm A_c)\right]\right) = \mathcal N(\bm 0, \sigma^2\bm V).
\]
Now because $\hat{\bm \lambda}\Pto\bm \lambda^*$, we obtain
\[
     T_2 = {\bm \lambda^*}^\top\frac{1}{n}\sum_{c=1}^n\bm \Lambda_c(\bm A_c)^\top\bm \epsilon_c + o_p(n^{-1/2}).
\]
Define $Z_c = {\bm \lambda^*}^\top\bm \Lambda_c(\bm A_c)\bm \epsilon_c$. Note that $ Z_c$ has mean $\bm 0$, variance $\sigma_z^2 = \sigma^2{\bm \lambda^*}^\top\bm V\bm \lambda^* = \sigma^2\mathbb E\left[\|{\mathcal P}\bm w_{\mathrm{IPW},c}\|^2\right]$ and is un-correlated with $\bm W_c$. Thus, we have,
\[
    \sqrt{n}\left(\begin{pmatrix}\bar{\bm W}\\
    \bar{ Z}\end{pmatrix} - \begin{pmatrix}\bm \mu\\ 0\end{pmatrix}\right)\dto\mathcal N\left(\begin{pmatrix}\bm 0\\ 0\end{pmatrix}, \begin{bmatrix}\bm \Sigma & \bm 0\\ \bm 0 & \sigma_z^2\end{bmatrix}\right).
\]
Now, define the transformation
\[
    f(\bm x, \bm y) =-(\bm \lambda^* + \bm V^{-1}\bm y)^\top\bm x = (\mathbb E[\bm v_c] - \bm y)^\top\bm V^{-1}\bm x,
\]
and let $f_1(\bm x, \bm y, z):=(f(\bm x, \bm y), z)^\top$, where, recall that $\bm V = \mathbb E[\bm \Lambda_c(\bm A_c)^\top \bm \Lambda_c(\bm A_c)]$. Then, we have
\[
    \nabla f = \begin{pmatrix}
        -(\bm \lambda^* + \bm V^{-1}\bm y)\\
        -\bm V^{-1}\bm x
    \end{pmatrix} = \begin{pmatrix}
        \bm V^{-1}(\mathbb E \bm v_c - \bm y)\\
        -\bm V^{-1}\bm x
    \end{pmatrix},
\]
and
\[
    \nabla f_1 = \begin{bmatrix}
        \nabla f & \bm 0\\
        0 & 1
    \end{bmatrix}
\]
Now, observe that $T_1 = f(\bar{\bm W})$. Thus, we apply the delta method on the vector $(\bar{\bm W}^\top, \bar Z)^\top$ with the transformation $f_1$ to obtain
\[
    \sqrt{n}\left(\begin{pmatrix}T_1\\T_2\end{pmatrix} - \begin{pmatrix}f(\bm \mu) \\  0\end{pmatrix}\right)\dto \mathcal N\left( \bm 0, \begin{bmatrix}\nabla f(\bm \mu)^\top\bm \Sigma \nabla f(\bm \mu) &  0\\  0 & \sigma_Z^2\end{bmatrix}\right).
\]
Now, we find the entries in this expression of the CLT. First note that,
\[
    \bm \mu = \begin{pmatrix}
        \mathbb E[\bm \Lambda_c(\bm A_c)^\top\bm \Lambda_c^*(\bm A_c)\bm h]\\
        \bm 0
    \end{pmatrix},
\]
and hence, 
\begin{align*}
    f(\bm \mu) &= \mathbb E[\bm v_c]^\top\bm V^{-1}\mathbb E[\bm \Lambda_c(\bm A_c)^\top\bm \Lambda_c^*(\bm A_c)\bm h]\\
        &=\mathbb E[\bm w_{\mathrm{IPW},c}^\top\bm \Lambda_c(\bm A_c)\bm h_{{\bm \Lambda^\ast}}]\\
        &=\mathbb E\left[\bm w_{\mathrm{IPW},c}^\top P^*\bm g^{(\bm A_c)}_c\right] = \mu_f(\bm \Lambda, {\bm \Lambda^\ast}),
\end{align*}
where,
\[
    \bm h_{{\bm \Lambda^\ast}} = \mathbb E[\bm \Lambda_c(\bm A_c)^\top \bm \Lambda_c(\bm A_c)]^{-1}\mathbb E[\bm \Lambda_c(\bm A_c)^\top\bm \Lambda_c^*(\bm A_c)\bm h] = {\mathcal P}\bm g^{(\bm A_c)}_c,
\]
is the theoretical `OLS estimate' when regressing $\bm \Lambda_c^*(\bm A_c)\bm h$ on $\bm \Lambda_c(\bm A_c)$. Now to figure out the variance, 
\[
    \nabla f(\bm \mu) = \begin{pmatrix}
        \bm V^{-1}\mathbb E[\bm v_c]\\
        -\bm V^{-1}\mathbb E[\bm \Lambda_c(\bm A_c)^\top\bm \Lambda_c^*(\bm A_c)\bm h].
    \end{pmatrix} = \begin{pmatrix}
        \bm V^{-1}\mathbb E[\bm v_c]\\
        -\bm h_{{\bm \Lambda^\ast}}
    \end{pmatrix}.
\]
Now, we start with obtaining an expression for the limiting variance of $T(\wipwprojlamb(e))$. Note that it is given by
\[
    \sigma^2(\bm \Lambda, {\bm \Lambda^\ast}) = \mathbb E[\bm v_c^\top]\bm V^{-1}\bm \Sigma_{11}\bm V^{-1}\mathbb E[\bm v_c] - 2\mathbb E[\bm v_c^\top]\bm V^{-1}\bm \Sigma_{12}\bm h_{{\bm \Lambda^\ast}} + \bm h_{{\bm \Lambda^\ast}}^\top\bm \Sigma_{22}\bm h_{{\bm \Lambda^\ast}}.
\]
We consider each block of the covariance matrix $\bm \Sigma$. We have
\begin{align*}
    \bm \Sigma_{11} =& \underbrace{\mathbb E[\bm \Lambda_c(\bm A_c)^\top\bm \Lambda_c^*(\bm A_c)\bm h\bm h^\top \bm \Lambda_c^*(\bm A_c)^\top \bm \Lambda_c(\bm A_c)]}_{\mathrm{A2}} - \underbrace{\bm V\bm h_{{\bm \Lambda^\ast}}\bm h_{{\bm \Lambda^\ast}}^\top\bm V}_{\mathrm{A1}}\\
    \bm \Sigma_{12} =& \underbrace{\mathbb E[\bm \Lambda_c(\bm A_c)^\top\bm \Lambda_c^*(\bm A_c)\bm h \mathbb E[\bm v_c^\top]\bm V^{-1}\bm \Lambda(\bm A_c)^\top\bm \Lambda_c(\bm A_c)]}_{\mathrm{A2}} - \underbrace{\mathbb E[\bm \Lambda_c(\bm A_c)^\top\bm \Lambda_c^*(\bm A_c)\bm h \bm v_c^\top]}_{\mathrm{A4}}\\
    \bm \Sigma_{22} =& \underbrace{\mathbb E[\bm \Lambda_c(\bm A_c)^\top\bm \Lambda_c(\bm A_c)\bm V^{-1}\mathbb E[\bm v_c] \mathbb E [\bm v_c^\top] \bm V^{-1}\bm \Lambda_c(\bm A_c)^\top\bm \Lambda_c(\bm A_c)]}_{\mathrm{A3}} + \underbrace{\mathbb E[\bm v_c\bm v_c^\top]}_{\mathrm{A1}} - \underbrace{2\mathbb E[\bm \Lambda_c(\bm A_c)^\top\bm \Lambda_c(\bm A_c)\bm V^{-1}\mathbb E[\bm v_c]\bm v_c^\top]}_{\mathrm{A4}}
\end{align*}
Now, the terms marked $A1$ appear in the expression of $\sigma^2(\bm \Lambda, {\bm \Lambda^\ast})$ as:
\begin{align*}
    S_1 &= -\mathbb E[\bm v_c^\top]\bm V^{-1}\bm V\bm h_{{\bm \Lambda^\ast}}\bm h_{{\bm \Lambda^\ast}}^\top\bm V\bm V\mathbb E[\bm v_c] + \bm h_{{\bm \Lambda^\ast}}^\top\mathbb E[\bm v_c\bm v_c^\top]\bm h_{{\bm \Lambda^\ast}}\\
    &= \mathbb E[(\bm h_{{\bm \Lambda^\ast}}^\top\bm v_c)^2] - (\mathbb E[\bm h_{{\bm \Lambda^\ast}}^\top\bm v_c])^2\\
    &= \mathrm{Var}(\bm h_{{\bm \Lambda^\ast}}^\top\bm v_c)
\end{align*}
Similarly, the terms marked $A2$ appear as:
\begin{align*}
    S_2 =& \mathbb E[\bm v_c^\top]\bm V^{-1}\mathbb E[\bm \Lambda_c(\bm A_c)^\top\bm \Lambda_c^*(\bm A_c)\bm h\bm h^\top\bm \Lambda_c^*(\bm A_c)^\top\bm \Lambda_c(\bm A_c)] \\
    &- 2\mathbb E[\bm v_c^\top]\bm V^{-1}\mathbb E[\bm \Lambda_c(\bm A_c)^\top\bm \Lambda_c^*(\bm A_c)\bm h \mathbb E[\bm v_c^\top]\bm V^{-1}\bm \Lambda_c(\bm A_c)^\top\bm \Lambda_c(\bm A_c)]\bm h_{{\bm \Lambda^\ast}}\\
    =& \mathrm{Var}(\bm h^\top\bm \Lambda_c^*(\bm A_c)^\top\bm \Lambda_i(\bm A_c)\bm V^{-1}\mathbb E[\bm v_c]) + (\mathbb E[\bm h_{{\bm \Lambda^\ast}}^\top\bm v_c])^2 \\
    &- 2\mathrm{Cov}(\bm h^\top\bm \Lambda_c^*(\bm A_c)^\top\bm \Lambda_c(\bm A_c)\bm V^{-1}\mathbb E[\bm v_c],\ \bm h_{{\bm \Lambda^\ast}}^\top\bm \Lambda_c(\bm A_c)^\top\bm \Lambda_c(\bm A_c)\bm V^{-1}\mathbb E [\bm v_c]) - 2(\mathbb E[\bm h_{{\bm \Lambda^\ast}}^\top\bm v_c])^2\\
    =& \mathrm{Var}(\bm h^\top\bm \Lambda_c^*(\bm A_c)^\top\bm \Lambda_c(\bm A_c)\bm V^{-1}\mathbb E[\bm v_c]) - 2\mathrm{Cov}(\bm h^\top\bm \Lambda_c^*(\bm A_c)^\top\bm \Lambda_c(\bm A_c)\bm V^{-1}\mathbb E[\bm v_c],\ \bm h_{{\bm \Lambda^\ast}}^\top\bm \Lambda_c(\bm A_c)^\top\bm \Lambda_c(\bm A_c)\bm V^{-1}\mathbb E [\bm v_c])\\
    &- (\mathbb E[\bm h_{{\bm \Lambda^\ast}}^\top\bm v_c])^2,
\end{align*}
where we have used the fact that 
\[
    \mathbb E[\bm h^\top\bm \Lambda_c^*(\bm A_c)^\top\bm \Lambda_c(\bm A_c)\bm V^{-1}\mathbb E[\bm v_c]] = \mathbb E[\bm h_{{\bm \Lambda^\ast}}^\top\bm \Lambda_c(\bm A_c)^\top\bm \Lambda_c(\bm A_c)\bm V^{-1}\mathbb E[\bm v_c]] = \mathbb E[\bm h_{{\bm \Lambda^\ast}}^\top\bm v_c].
\]
The term marked $A3$ appears as:
\[
    S_3 = \mathrm{Var}(\bm h_{{\bm \Lambda^\ast}}^\top\bm \Lambda_c(\bm A_c)^\top\bm \Lambda_c(\bm A_c)\bm V^{-1}\mathbb E[\bm v_c]) + (\mathbb E[\bm h_{{\bm \Lambda^\ast}}^\top\bm v_c])^2
\]
Finally, using similar calculations, the term marked $A4$ appears as
\begin{align*}
    S_4 =& 2\mathbb E[\bm v_c^\top]\bm V^{-1}\mathbb E[\bm \Lambda_c(\bm A_c)^\top\bm \Lambda_c^*(\bm A_c)\bm h \bm v_c^\top]\bm h_{{\bm \Lambda^\ast}}\\
    &-2\bm h_{{\bm \Lambda^\ast}}^\top\mathbb E[\bm \Lambda_c(\bm A_c)^\top\bm \Lambda_c(\bm A_c)\bm V^{-1}\mathbb E[\bm v_c]\bm v_c^\top\bm h_{{\bm \Lambda^\ast}}]\\
    =& 2\mathbb E[\bm h^\top\bm \Lambda_c^*(\bm A_c)^\top\bm \Lambda_c(\bm A_c)\bm V^{-1}\mathbb E [\bm v_c]\bm v_c^\top\bm h_{{\bm \Lambda^\ast}}]\\
    &- 2\mathbb E[\bm h_{{\bm \Lambda^\ast}}^\top\bm \Lambda_c(\bm A_c)^\top\bm \Lambda_c(\bm A_c)\bm V^{-1}\mathbb E[\bm v_c]\bm v_c^\top\bm h_{{\bm \Lambda^\ast}}]\\
    =& 2\mathrm{Cov}(\bm h^\top\bm \Lambda_c^*(\bm A_c)^\top\bm \Lambda_c(\bm A_c)\bm V^{-1}\mathbb E[\bm v_c],\ \bm v_c^\top\bm h_{{\bm \Lambda^\ast}})\\
    &- 2\mathrm{Cov}(\bm h_{{\bm \Lambda^\ast}}^\top\bm \Lambda_c(\bm A_c)^\top\bm \Lambda_c(\bm A_c)\bm V^{-1}\mathbb E[\bm v_c],\ \bm v_c^\top\bm h_{{\bm \Lambda^\ast}})
\end{align*}
Combining all these terms yilds,
\begin{align*}
    \nabla f(\bm \mu)^\top\bm \Sigma \nabla f(\bm \mu) =& \mathrm{Var}(\bm h_{{\bm \Lambda^\ast}}^\top\bm v_c)+ \mathrm{Var}(\bm h^\top\bm \Lambda_c^*(\bm A_c)^\top\bm \Lambda_c(\bm A_c)\bm V^{-1}\mathbb E[\bm v_c])\\
    &+ \mathrm{Var}(\bm h_{{\bm \Lambda^\ast}}^\top\bm \Lambda_c(\bm A_c)^\top\bm \Lambda_c(\bm A_c)\bm V^{-1}\mathbb E[\bm v_c])\\
    &- 2\mathrm{Cov}(\bm h^\top\bm \Lambda_c^*(\bm A_c)^\top\bm \Lambda_c(\bm A_c)\bm V^{-1}\mathbb E[\bm v_c],\ \bm h_{{\bm \Lambda^\ast}}^\top\bm \Lambda_c(\bm A_c)^\top\bm \Lambda_c(\bm A_c)\bm V^{-1}\mathbb E[\bm v_c])\\
    &+ 2\mathrm{Cov}(\bm h^\top\bm \Lambda_c^*(\bm A_c)^\top\bm \Lambda_c(\bm A_c)\bm V^{-1}\mathbb E[\bm v_c],\ \bm v_c^\top\bm h_{{\bm \Lambda^\ast}})\\
    &- 2\mathrm{Cov}(\bm h_{{\bm \Lambda^\ast}}^\top\bm \Lambda_c(\bm A_c)^\top\bm \Lambda_c(\bm A_c)\bm V^{-1}\mathbb E[\bm v_c], \bm v_c^\top\bm h_{{\bm \Lambda^\ast}})\\
    =& \mathrm{Var}(\bm h^\top\bm \Lambda_c^*(\bm A_c)^\top\bm \Lambda_c(\bm A_c)\bm V^{-1}\mathbb E[\bm v_c] - \bm h_{{\bm \Lambda^\ast}}^\top\bm \Lambda_c(\bm A_c)^\top\bm \Lambda_c(\bm A_c)\bm V^{-1}\mathbb E[\bm v_c] + \bm h_{{\bm \Lambda^\ast}}^\top\bm v_c)\\
    =&\mathrm{Var}((\bm \Lambda_c^*(\bm A_c)\bm h - \bm \Lambda_c(\bm A_c)\bm h_{{\bm \Lambda^\ast}})^\top\bm \Lambda_c(\bm A_c)\bm V^{-1}\mathbb E[\bm v_c] + \bm h_{{\bm \Lambda^\ast}}^\top\bm v_c)
\end{align*}
Now, note that
\begin{align*}
    &\bm \Lambda_c^*(\bm A_c)\bm h = \bm g_c^{(\bm A_c)}\textrm{ and , }\\
    &\bm \Lambda_c(\bm A_c)\bm V^{-1}\mathbb E\bm v_c = \bm \Lambda_c(\bm A_c)\bm V^{-1}\mathbb E[\bm \Lambda_c(\bm A_c)^\top\bm w_{\mathrm{IPW}}] = {\mathcal P}\bm w_{\mathrm{IPW}},
    \Lambda_c(\bm A_c)\bm h_{{\bm \Lambda^\ast}} = {\mathcal P}\bm g_c^{(\bm A_c)}.
\end{align*}
Finally, note that $\bm h_{{\bm \Lambda^\ast}}^\top\bm v_c = (\bm \Lambda_c(\bm A_c)\bm h_{{\bm \Lambda^\ast}})^\top\bm w_{\mathrm{IPW},c} = \bm w_{\mathrm{IPW},c}^\top \mathcal P\bm g_c^{(\bm A_c)}$. This shows that
\[
    \nabla f(\bm \mu)^\top\bm \Sigma \nabla f(\bm \mu) = \mathrm{Var}\left[\left((I - {\mathcal P})\bm g_c^{(\bm A_c)}\right)^\top\left({\mathcal P}\bm w_{\mathrm{IPW},c}\right)+\bm w_{\mathrm{IPW},c}^\top \mathcal P\bm g_c^{(\bm A_c)}\right].
\]

Finally, applying the transformation $(t_1,t_2)\mapsto t_1+t_2$, we again invoke the Delta method on the joint distribution of $(T_1, T_2)$ to obtain,
\begin{align*}
&\sqrt{n}\left(T\left(\wipwprojlamb(e)\right) - \bm \mu_f(\bm \Lambda, {\bm \Lambda^\ast})\right) = \sqrt{n}\left(T_1+T_2 - \bm \mu\right)\\
\dto& \mathcal N\left(0,\mathrm{Var}\left[\left((I - {\mathcal P})\bm g_c^{(\bm A_c)}\right)^\top\left({\mathcal P}\bm w_{\mathrm{IPW},c}\right)+\bm w_{\mathrm{IPW},c}^\top \mathcal P\bm g_c^{(\bm A_c)}\right] + \sigma^2\mathbb E\left[\|{\mathcal P}\bm w_{\mathrm{IPW,c}}||^2\right] \right),
\end{align*}
which proves the required asymptotic normality. We next give a proof sketch for establishing asymptotic normality of the weighted projection estimator.
\end{proof}
\subsubsection{Proof sketch for establishing asymptotic normality of the weighted projection estimator}
\label{sec:proof_sketch_weighted_projection}
We begin by noting that the weighted projection estimator's weights are given by $\wipwlamb(e) = \bm R\wipwlambtilde(e)$, where $\wipwlambtilde(e)$ satisfied
\begin{align*}
    {\rm Minimize:}~\frac{1}{2}\tilde{\bm w}^T\bm D\tilde{\bm w}~\textrm{ Subject to: } \bm \Lambda^T\bm D(\tilde{\bm w} - \tilde{\bm w}_{\rm IPW}) = \bm 0.
\end{align*}
Similarly, we write down the Lagarangian in this case as
\begin{align*}
    \mathcal L(\bm \lambda, \tilde{\bm w}) = \frac{1}{2}\tilde{\bm w}^T\bm D\tilde{\bm w} + \bm \lambda^T\bm \Lambda^T\bm D(\tilde{\bm w} - \tilde{\bm w}_{\rm IPW}).
\end{align*}
One similarly obtains that the first order conditions give $\wipwlambtilde(e) = -\bm \Lambda \hat{\bm \lambda}$, where $\hat{\bm \lambda}$ maximizes the dual:
\begin{align*}
    \mathcal L^*(\bm \lambda)= -\left(\frac{1}{2}\bm \lambda^T\bm \Lambda^T\bm D\bm \Lambda \bm \lambda + \bm \lambda^T\bm \Lambda^T\bm D\tilde{\bm w}_{\rm IPW}\right).
\end{align*}
From the fact that $\wipwlamb(e) = \bm R\wipwlambtilde(e)$, we get that $\wipwlamb(e)_c = -\bm \Lambda_c(A_c)\hat{\bm \lambda}$. Next, invoking the first-order conditions on the dual, we get that $\hat{\bm \lambda}$ satisfies
\begin{align*}
    &\bm \Lambda^T\bm D\bm \Lambda\hat{\bm \lambda} = -\bm \Lambda^T\bm D\tilde{\bm w}_{\rm IPW}\\
    \textrm{or equivalently, }&\frac{1}{n}\sum_{c=1}^n \mathbb E\left[\bm \Lambda_c(\bm A_c)^T\bm \Lambda_c(\bm A_c)\middle| \bm X_c\right]\hat{\bm \lambda} = -\frac{1}{n}\sum_{c=1}^n \tilde{\bm v}_c,
\end{align*}
where, $\tilde{\bm v}_c = \mathbb E\left[\bm \Lambda_c(\bm A_c)^T\bm w_{{\rm IPW},c}\middle| \bm X_c\right]$. One can now analogously follow the steps from Equation~\eqref{eq:Tterms} in the above proof to get the desired CLT for $T(\wipwlamb(e))$.
\subsection{Proof of Theorem~\ref{thm:clt_fixed}}
\label{sec:proof_thm_clt_fixed}
\begin{proof}

To show the asymptotic normality for the balancing estimator, $T(\wnulllamb)$, we follow the proof shown in Section~\ref{sec:proof_thm_known_propensity_clt_fixed}. The only modification is the expression for $\bm v_c$, which is given by
\[
    \bm v_c = \frac{1}{M_c}\sum_{\bm a_c\in \{0,1\}^{M_c}}\bm \Lambda_c(\bm a_c)^\top\bm 1 f(\bm a_c, \bm X_c).
\]
Let us use $\bm v_c^{(\stwo)}$ to denote the $\bm v_c$ in the proof in Section~\ref{sec:proof_thm_known_propensity_clt_fixed}. Then, by the law of total expectation, we have
\begin{align*}
    \mathbb E[\bm v_c^{(\stwo)}] &= \mathbb E\left[\mathbb E[\bm v_c^{(\stwo)}\mid\bm X_c]\right]\\ 
    &=  \mathbb E\left[\mathbb E[\bm\Lambda_c(\bm A_c)^\top\bm w_{\mathrm{IPW,c}}\mid\bm X_c]\right]\\
    &=\mathbb E\left[\frac{1}{M_c}\sum_{\bm a\in \{0,1\}^{m_c}}\bm \Lambda_c(\bm a)^\top\bm 1f(\bm a;\bm X_c)\right] \\
    & = \mathbb E[\bm v_c].
\end{align*}
This shows that the expectation of $\bm v_c$ remains identical, so that any part of the derivation in Section~\ref{sec:proof_thm_known_propensity_clt_fixed} that only depends on $\bm v_c$ through $\mathbb E\bm v_c$ follows the same way. However, we need to verify that $\mathbb E[\|\bm v_c\|^2]<\infty$. We have shown this for $\bm v_c^{(\stwo)}$. We now show it for the balancing estimator as follows,
\begin{align*}
    \mathbb E[\|\bm v_c\|^2] &= \mathbb E\left[\left\|\mathbb E[\bm \Lambda_c(\bm A_c)^\top\bm w_{\mathrm{IPW},c}\mid \bm X_c]\right\|^2\right]\\
    &\leq \mathbb E\left[\mathbb E[\| \bm \Lambda_c(\bm A_c)^\top\bm w_{\mathrm{IPW},c}\|^2\mid \bm X_c]\right]\\
    &\leq \mathbb E\left[\| \bm \Lambda_c(\bm A_c)^\top\bm w_{\mathrm{IPW},c}\|^2\right]\\
    &=\mathbb E\left[\|\bm v_c^{(\stwo)}\|^2\right]<\infty.
\end{align*}
The other difference is the fact that the balancing equation has a solution only when $B_{\bm \Lambda}$ holds. But, under Assumption~\ref{assm:balancing_feasibility}, $P(B_{\bm \Lambda})\to 1$. Hence, one can define $\wnulllamb$ arbitrarily on $B_{\bm \Lambda}^c$, and because $P(B_{\bm \Lambda}^c)\to 0$, this extension does not affect the limiting distribution. Then, one can analogously follow the proof in Section~\ref{sec:proof_thm_known_propensity_clt_fixed} using this new expression for $\bm v_c$ (and note that $\bm v_c = \mathbb E[\bm v_c^{\stwo}\mid \bm X_c]$) and conclude that
\begin{align*}
&\sqrt{n}\left(T\left(\wnulllamb\right) - \bm \mu_f(\bm \Lambda, {\bm \Lambda^\ast})\right)\\
    \dto& \mathcal N\left(0,\mathrm{Var}\left[\left((I - {\mathcal P})\bm g_c^{(\bm A_c)}\right)^\top\left({\mathcal P}\bm w_{\mathrm{IPW},c}\right) + \mathbb E[\bm w_{\mathrm{IPW},c}^\top P^*\bm g_c^{(\bm A_c)}\mid \bm X_c]\right] + \sigma^2\mathbb E\left[\|{\mathcal P}\bm w_{\mathrm{IPW,c}}||^2\right] \right).
\end{align*}
\end{proof}

\subsection{Proof of Theorem \ref{thm:var_estimation}}
\label{sec:proof_thm_var_estimation}

The complete statement of Theorem~\ref{thm:var_estimation} is presented in the following lemma, which is proved subsequently.
\begin{lemma}[Consistent variance estimators: full variance expressions]
\label{lem:var_estimation}
Define $\bm v_c^{(\stwo)} = \bm \Lambda_c(\bm A_c)^\top \bm w_{\mathrm{IPW},c}$ and $\bm v_c^{(\sthree)} =  \mathbb E[\bm \Lambda_c(\bm A_c)^\top \bm w_{\mathrm{IPW},c}\mid \bm X_c]$. Next, let $\hat {\bm \lambda}^{j}$ be a solution to the equation,
$$
    \frac{1}{n}\sum_{c=1}^n \bm \Lambda_c(\bm A_c)^\top\bm \Lambda_c(\bm A_c)\bm \lambda = -\frac{1}{n}\sum_{c=1}^n\bm v^{(j)}_c,
$$
    for $j\in\{\stwo,\sthree\}$, and, $\hat {\bm \lambda}^{\sone}$ be a solution to 
$$
    \frac{1}{n}\sum_{c=1}^n\mathbb E[ \bm \Lambda_c(\bm A_c)^\top\bm \Lambda_c(\bm A_c)\mid \bm X_c]\bm \lambda = -\frac{1}{n}\sum_{c=1}^n\bm v^{(\sthree)}_c.
$$
Also, define,
\begin{align*}
\left(\bm\eta_c^{\stwo}\right)^\top &= \begin{pmatrix}
    \bm \Lambda_c(\bm A_c)^\top\wipwprojlamb(e)_c - \bm v_c^{(\stwo)} & \wipwprojlamb(e)_c^\top\bm y_c - T\left(\wipwprojlamb(e)\right)
    \end{pmatrix},\\
    \left(\bm\eta_c^{\sthree}\right)^\top &= \begin{pmatrix}
    \bm \Lambda_c(\bm A_c)^\top\left(\wnulllamb\right)_c - \bm v_c^{(\sthree)} & \left(\wnulllamb\right)_c^\top\bm y_c - T\left(\wnulllamb\right)
    \end{pmatrix},\\
    \left(\bm\eta_c^{\sone}\right)^\top &= \begin{pmatrix}
        -\mathbb [\bm \Lambda_i(\bm A_c)^\top\bm \Lambda_c\bm A_c)\mid \bm X_c]\hat{\bm \lambda}^{\sone} - \bm v_c^{(\stwo)} & \wipwlamb(e)_c^\top\bm y_c - T\left(\wipwlamb(e)\right)
        \end{pmatrix}.
    \end{align*}
    and
    \[
    \bm L = \begin{pmatrix}
        \left(\left(\bm R\bm \Lambda\right)^+\bm Y\right)^\top & -1
    \end{pmatrix}^\top.
    \]
    Finally, the variance estimator is given by
    \begin{align}
        \label{eqn:var_est_exprn}
        \hat \sigma^{2}_j = \frac{1}{n}\sum_{c=1}^n\left( \left(\bm \eta_c^{j}\right)^\top\bm L\right)^2.
    \end{align}
    Then, under the setting and assumptions in Theorem \ref{thm:clt_fixed}, we have that $\hat \sigma^{2}_j\Pto \sigma^2_{jf}(\bm \Lambda, \bm \Lambda^\ast)$, as $n\rightarrow \infty$, $\forall j \in \{\sone, \stwo, \sthree\}$.
\end{lemma}
\begin{proof}
The proof follows the same technique as in \cite{chanetal2016,additivefactorial}. For simplicity, we prove the result only for the case of $j=\stwo$. The cases of the other estimators follow similarly. For notational simplicity, we use $\hat{\bm w}$ to denote the fitted weights instead of $\wipwprojlamb(e)$, and in general, we drop the the script $j$ from notations. Note that from the proof in Section \ref{sec:proof_thm_known_propensity_clt_fixed}, we have shown that $\hat{\bm w}_c = -\bm \Lambda_c(\bm A_c)\hat{\bm \lambda}$. Thus, in general, one might consider the class of weights given by $\bm w_c(\bm \lambda) = -\bm \Lambda_c(\bm A_c)\bm \lambda$ and note that $\hat{\bm w}_c = \bm w_c(\hat{\bm \lambda})$. Now, define $\bm \theta = (\bm \lambda^\top, t)^\top$ and \begin{align*}
\bm\eta_c(\bm \theta) = \begin{pmatrix}-\bm \Lambda_c(\bm A_c)^\top\bm \Lambda_c(\bm A_c)\bm \lambda - \bm \Lambda_c(\bm A_c)^\top\bm w_{\mathrm{IPW},c}\\ \bm w_c(\bm \lambda)^\top\bm y_c - t\end{pmatrix}.
\end{align*}
Then note that $\hat{\bm \theta} = (\hat{\bm \lambda}, T(\hat{\bm w}))^\top$ satisfies
\[
    \sum_{c=1}^n \bm \eta_c(\hat{\bm \theta)} = \bm 0,
\]
showing that $\hat{\bm \theta}$ is a Z-estimator for $\bm \theta^* = (\bm \lambda^*, \mu_f(\bm \Lambda, {\bm \Lambda^\ast}))$ that satisfies $\mathbb E[\bm \eta_i(\bm \theta^*)] = \bm 0$. We again want a CLT for $\hat{\bm \theta}$ \citep[Theorem 5.41]{vandervaart1998asymptotic}. 

Note, we have that
\begin{align*}
    {\bm \eta_c}'(\bm \theta) = \begin{bmatrix}
        -\bm \Lambda_c(\bm A_c)^\top\bm \Lambda_c(\bm A_c) & \bm 0\\
        \bm 0^\top & -1
    \end{bmatrix}\implies \mathbb E{\bm \eta_c}'(\bm \theta) = \begin{bmatrix}
        -\bm V & \bm 0\\
        \bm 0^\top & -1
    \end{bmatrix},
\end{align*}
which due to Assumption \ref{assm:clt_technical_fixed}, is non-singular. Next, we show that $\mathbb \|\bm \eta_c(\bm\theta^*)\|<\infty$. Note that the first block of $\bm \eta_c(\bm \theta^*)$ is given by
\[
    -\bm \Lambda_c(\bm A_c)^\top\bm \Lambda_c(\bm A_c)\bm \lambda^* - \bm \Lambda_c(\bm A_c)^\top\bm w_{\mathrm{IPW},c},
\]
which has a finite expected squared norm under Assumption~\ref{assm:clt_technical_fixed} following the proof in Section~\ref{sec:proof_thm_known_propensity_clt_fixed}. The squared-norm (second moment) of the second block is given by
\begin{align*}
    &(\bm w_c(\bm \lambda^*)^\top\bm y_c- \mu_{f}(\bm \Lambda, {\bm \Lambda^\ast}))^2\\
    =&(-{\bm \lambda^*}^\top\bm \Lambda_c(\bm A_c)^\top\bm \Lambda_c(\bm A_c)\bm h  - {\bm \lambda^*}^\top\bm \Lambda_c(\bm A_c)^\top\bm \epsilon_c - \mu_f(\bm \Lambda, {\bm \Lambda^\ast}))^2\\
    < & 4\left(-{\bm \lambda^*}^\top\bm \Lambda_c(\bm A_c)^\top\bm \Lambda_c(\bm A_c)\bm h\right)^2 + 4\left({\bm \lambda^*}^\top\bm \Lambda_c(\bm A_c)^\top\bm \epsilon_c\right) + 4\mu_f^2(\bm \Lambda, {\bm \Lambda^\ast}).
\end{align*}
Of these, the first term has finite expectation, again following the proof in Section~\ref{sec:proof_thm_known_propensity_clt_fixed}, while the last term is a constant. Next, we show that $\mathbb E\left({\bm \lambda^*}^\top\bm \Lambda_c(\bm A_c)^\top\bm \epsilon_c\right)<\infty$, which follows from the fact
\begin{align*}
    \mathbb E\left[({\bm \lambda^*}^\top\bm \Lambda_c(\bm A_c)^\top\bm \epsilon_c)^2\right]
    = &\mathbb E\left[\mathbb E[({\bm \lambda^*}^\top\bm \Lambda_c(\bm A_c)^\top\bm \epsilon_c)^2\mid\bm X_c]\right]\\
    =& \sigma^2\mathbb E\left[    {\bm \lambda^*}^\top\bm \Lambda_c(\bm A_c)^\top \bm \Lambda_c(\bm A_c)\bm \lambda^*  \right],
\end{align*}
which is again finite, by the proof in Section~\ref{sec:proof_thm_known_propensity_clt_fixed}. Finally note that all the second partial derivatives of $\bm \eta_c(\bm \theta)$ are 0, and hence, we now invoke \citep[Theorem 5.41]{vandervaart1998asymptotic} to conclude that 
\[
    \sqrt{n}(\hat{\bm \theta} - \bm \theta^*)\dto \mathcal N\left(\bm 0, \bm U^{-1}\mathbb E[\bm \eta_i(\bm \theta^*)\bm \eta_i(\bm \theta^*)^\top](\bm U^{-1})^\top\right),
\]
where, $\bm U = \mathbb E[\bm \eta_i'(\bm \theta^*)]$. From here, we derive a CLT for $T(\hat{\bm w})$, as it is the last component of $\hat{\bm \theta}$, as:
\[
    \sqrt{n}(T(\hat{\bm w}) - \mu_f(\bm \Lambda, {\bm \Lambda^\ast}))\dto  \mathcal N\left( 0, \mathbb E\left(\bm \eta_i(\bm \theta^*)^\top\bm U_l^{-1} \right)^2\right), 
\]
where $\bm U_l^{-1}$ denotes the last row of $\bm U^{-1}$. Using expression for the inversion of block matrices, one obtains that the 
last row of $\bm U^{-1}$ is given by, 
\[
    \begin{bmatrix}(\bm V^{-1}\mathbb E[\bm \Lambda_c(\bm A_c)^\top\bm Y_c])^\top & -1\end{bmatrix}.
\]
Note that the first component of the vector is the theoretical OLS coefficient of regressing $\bm Y_c$ onto the columnspace of $\bm \Lambda_i(\bm A_c)$, so that $\bm L$ is a consistent estimator of $\bm U^{-1}_l$. Our proof is complete, if we just show that
\[
    \frac{1}{n}\sum_{c=1}^n \bm \eta_c(\hat{\bm \theta)}\Pto \mathbb E[\bm \eta_c(\bm \theta^*)].
\]
To see this, first write $\frac{1}{n}\sum_{c=1}^n \bm \eta_c(\hat{\bm \theta}) = \frac{1}{n}\sum_{c=1}^n (\bm \eta_c(\hat{\bm \theta)} - \bm \eta_c(\bm \theta^*)) + \frac{1}{n}\sum_{c=1}^n \bm \eta_c(\bm \theta^*)$. It suffices to show that the first summand goes to 0 in probability. The proof is completed from the fact that $\hat{\bm \theta}\Pto \bm \theta^*$, and Lemma~\ref{lem:eta_lipschitz} from which we can write that with probability going to 1,
\begin{align*}
    \left\|\frac{1}{n}\sum_{c=1}^n \left(\bm \eta_c(\hat{\bm \theta}) - \bm \eta_c(\bm \theta^*)\right)\right\|\leq \|\hat{\bm \theta} - \bm \theta^*\|\frac{1}{n}\sum_{c=1}^n \dot{\eta}_c.
\end{align*}
Because $\mathbb E |\dot\eta^2|<\infty$, we conclude that $\frac{1}{n}\sum_{c=1}^n (\bm \eta_c(\hat{\bm \theta)} - \bm \eta_c(\bm \theta^*))\Pto 0$, thereby completing the proof.
\end{proof}

\begin{lemma}
    \label{lem:eta_lipschitz}
    Under Assumption~\ref{assm:clt_technical_fixed}, there is a neighborhood about $\bm \theta^*$ and a $\dot{\eta}_c:=\dot\eta_c(\bm \Lambda_c(\bm A_c),\bm \epsilon_c)$ with $\mathbb E|\dot\eta_c|<\infty$, such that for $\bm \theta_1,\bm\theta_2$ in that neighborhood, we have that $\|\bm \eta_c(\bm \theta_1) - \bm \eta_c(\bm \theta_2)\|\leq \dot{\eta}_c\|\bm \theta_1 - \bm \theta_2\|$.
\end{lemma}
\begin{proof}
    Note that,
\begin{align}
\label{eqn:eta_lipschitz}
\begin{aligned}
\|\bm \eta_c(\bm \theta_1) - \bm \eta_c(\bm \theta_2)\|^2&\leq  \left\|\bm \Lambda_c(\bm A_c)^\top\bm \Lambda_c(\bm A_c)(\bm \lambda_1 - \bm \lambda_2)\right\|^2 + \left(\bm y_c^\top\bm \Lambda_c(\bm A_c)(\bm \lambda_1 - \bm \lambda_2)\right)^2 + (t_1 - t_2)^2\\
&\leq \sigma^2_{\mathrm{max}}(\bm \Lambda_c(\bm A_c)^\top\bm \Lambda_c(\bm A_c))\|\bm \lambda_1 - \bm \lambda_2\|^2 + \|\bm y_c^\top\bm \Lambda_c(\bm A_c)\|^2\|\bm \lambda_1 - \bm \lambda_2\|^2 + (t_1-t_2)^2.
\end{aligned}
\end{align}
we write $\bm y_c^\top\bm \Lambda_c(\bm A_c) = \bm h^\top\bm \Lambda_c^*(\bm A_c)^\top\bm \Lambda_c(\bm A_c) + \bm \epsilon_c^\top\bm \Lambda_c(\bm A_c)$, which implies
$$\begin{aligned}
\|\bm y_c^\top\bm \Lambda_c(\bm A_c)\|^2\leq &2\sigma^2_{\mathrm{max}}\left(\bm \Lambda_c(\bm A_c)^\top{\bm \Lambda^\ast}(\bm A_c)\right)\|\bm h\|^2 + 2\bm \epsilon_c^\top\bm \Lambda_c(\bm A_c)^\top\bm \Lambda_c(\bm A_c)\bm \epsilon_c\\
\leq &  2\sigma^2_{\mathrm{max}}\left(\bm \Lambda_c(\bm A_c)^\top{\bm \Lambda^\ast}(\bm A_c)\right)\|\bm h\|^2 + 2\sigma_{\mathrm{max}}\left(\bm \Lambda_c(\bm A_c)^\top\bm \Lambda(\bm A_c)\right)\| \bm \epsilon_c\|^2.
\end{aligned}
$$ 
These relations imply that
\begin{align*}
    &\|\bm \eta_c(\bm \theta_1) - \bm \eta_c(\bm \theta_2)\|^2\\
    \leq& \underbrace{ \left(\sigma^2_{\mathrm{max}}(\bm \Lambda_c(\bm A_c)^\top\bm \Lambda_c(\bm A_c)) + 2\sigma^2_{\mathrm{max}}\left(\bm \Lambda_c(\bm A_c)^\top{\bm \Lambda^\ast}(\bm A_c)\right)\|\bm h\|^2 + 2\sigma_{\mathrm{max}}\left(\bm \Lambda_c(\bm A_c)^\top\bm \Lambda(\bm A_c)\right)\| \bm \epsilon_c\|^2 + 1\right) }_{\dot\eta_c}\|\bm \theta_1 - \bm\theta_2\|^2.
\end{align*}
Then note that, 
\begin{align*}
    &\mathbb E \left[|\dot\eta_c|\right]\\
    \leq & \mathbb E\left[\sigma^2_{\mathrm{max}}(\bm \Lambda_c(\bm A_c)^\top\bm \Lambda_c(\bm A_c))\right] + 2\mathbb E\left[\sigma^2_{\mathrm{max}}\left(\bm \Lambda_c(\bm A_c)^\top{\bm \Lambda^\ast}(\bm A_c)\right)\right] \|\bm h\|^2 + 2\mathbb E\left[M_c\sigma_{\mathrm{max}}\left(\bm \Lambda_c(\bm A_c)^\top\bm \Lambda(\bm A_c)\right)\right]\sigma^2 + 1.
\end{align*}
The finiteness of the above expression now follows from Assumption~\ref{assm:clt_technical_fixed}.
\end{proof}

\subsection{Proof of Theorem~\ref{thm:decide_true_structure}}
\label{sec:proof_decide_true_structure}
\begin{proof}
Fix $1\le l\le L-1$. We first prove parts (a) and (b) for the comparison between $\bm{\Lambda}_l$ and $\bm{\Lambda}_L$. Throughout the proof, write
\[
\hat{\bm{w}}^{(1)}=\hat{\bm w}^{\sthree}_{\bm \Lambda_l}, \qquad \hat{\bm{w}}^{(2)}=\hat{\bm w}^{\sthree}_{\bm \Lambda_L}, \qquad \Delta\hat{\bm{w}}=\hat{\bm{w}}^{(1)}-\hat{\bm{w}}^{(2)}.
\]
Here $\hat{\bm w}^{\sthree}_{\bm \Lambda_j}$ denotes the balancing weight constructed using the low-rank structure $\bm{\Lambda}_j$, for $j\in \{l, L\}$.

First suppose that $H_l$ holds. Since $\operatorname{col.sp.}(\bm{\Lambda}_l)\le \operatorname{col.sp.}(\bm{\Lambda}_L)$, there exists a matrix $\bm{C}_l$ such that $\bm{\Lambda}_l=\bm{\Lambda}_L\bm{C}_l$. Under $H_l$, the true outcome regression lies in the column space of $\bm{\Lambda}_l$. Hence, for some vector $\bm{h}_l$,
\[
\bm{g}=\bm{\Lambda}_l\bm{h}_l=\bm{\Lambda}_L\bm{C}_l\bm{h}_l.
\]
On the event $B_{\bm{\Lambda}_l}\cap B_{\bm{\Lambda}_L}$, the two balancing weights satisfy
\[
\bm{\Lambda}_l^\top\bm{R}^\top\hat{\bm{w}}^{(1)}=\bm{\Lambda}_l^\top\bm{f}, \qquad \bm{\Lambda}_L^\top\bm{R}^\top\hat{\bm{w}}^{(2)}=\bm{\Lambda}_L^\top\bm{f}.
\]
Therefore,
\[
\begin{aligned}
(\Delta\hat{\bm{w}})^\top\bm{Y} &= (\Delta\hat{\bm{w}})^\top\bm{R}\bm{\Lambda}_l\bm{h}_l+(\Delta\hat{\bm{w}})^\top\bm{\epsilon} \\
&= (\hat{\bm{w}}^{(1)})^\top\bm{R}\bm{\Lambda}_l\bm{h}_l-(\hat{\bm{w}}^{(2)})^\top\bm{R}\bm{\Lambda}_L\bm{C}_l\bm{h}_l+(\Delta\hat{\bm{w}})^\top\bm{\epsilon} \\
&= \bm{f}^\top\bm{\Lambda}_l\bm{h}_l-\bm{f}^\top\bm{\Lambda}_L\bm{C}_l\bm{h}_l+(\Delta\hat{\bm{w}})^\top\bm{\epsilon} \\
&= (\Delta\hat{\bm{w}})^\top\bm{\epsilon}.
\end{aligned}
\]
The last equality follows from $\bm{\Lambda}_l=\bm{\Lambda}_L\bm{C}_l$.

It remains to derive the limiting distribution of the standardized noise term. By the expansion used in the proof of Theorem~\ref{thm:clt_fixed}, for $j\in\{l,L\}$, there exists a deterministic vector $\bm{\lambda}_j^*$ such that
\[
\hat{\bm{w}}_{\bm{\Lambda}_j,c}^{\sthree}=-\bm{\Lambda}_{j,c}(\bm{A}_c)\hat{\bm{\lambda}}_j, \qquad \hat{\bm{\lambda}}_j-\bm{\lambda}_j^*=O_p(n^{-1/2}).
\]
Hence,
\[
\begin{aligned}
\frac{1}{\sqrt n}(\Delta\hat{\bm{w}})^\top\bm{\epsilon} &= -\hat{\bm{\lambda}}_l^\top\frac{1}{\sqrt n}\sum_{c=1}^n\bm{\Lambda}_{l,c}(\bm{A}_c)^\top\bm{\epsilon}_c+\hat{\bm{\lambda}}_L^\top\frac{1}{\sqrt n}\sum_{c=1}^n\bm{\Lambda}_{L,c}(\bm{A}_c)^\top\bm{\epsilon}_c \\
&= -\bm{\lambda}_l^{*\top}\frac{1}{\sqrt n}\sum_{c=1}^n\bm{\Lambda}_{l,c}(\bm{A}_c)^\top\bm{\epsilon}_c+\bm{\lambda}_L^{*\top}\frac{1}{\sqrt n}\sum_{c=1}^n\bm{\Lambda}_{L,c}(\bm{A}_c)^\top\bm{\epsilon}_c+o_p(1).
\end{aligned}
\]
The remainder is $o_p(1)$ because
\[
    \frac1n\sum_{c=1}^n\bm{\Lambda}_{j,c}(\bm{A}_c)^\top\bm{\epsilon}_c\overset{p}{\to}\mathbb{E}[\bm{\Lambda}_{j,c}(\bm{A}_c)^\top\bm{\epsilon}_c]=0~{\rm and }~\hat{\bm{\lambda}}_j-\bm{\lambda}_j^*=O_p(n^{-1/2}).
\]
Hence, by the central limit theorem,
\[
    \frac{1}{\sqrt n}(\Delta\hat{\bm{w}})^\top\bm{\epsilon}\overset{d}{\to}N(0,\sigma_{lL}^2),
\]
where
\[
    \sigma_{lL}^2=\mathbb{E}\left[\left\|\bm{\Lambda}_{l,c}(\bm{A}_c)\bm{\lambda}_l^*-\bm{\Lambda}_{L,c}(\bm{A}_c)\bm{\lambda}_L^*\right\|^2\right].
\]
Moreover,
\begin{align*}
\begin{aligned}
\frac1n\|\Delta\hat{\bm{w}}\|^2 &= \frac1n\sum_{c=1}^n\left\|\bm{\Lambda}_{l,c}(\bm{A}_c)\hat{\bm{\lambda}}_l-\bm{\Lambda}_{L,c}(\bm{A}_c)\hat{\bm{\lambda}}_L\right\|^2 \Pto \mathbb{E}\left[\left\|\bm{\Lambda}_{l,c}(\bm{A}_c)\bm{\lambda}_l^*-\bm{\Lambda}_{L,c}(\bm{A}_c)\bm{\lambda}_L^*\right\|^2\right]=\sigma_{lL}^2.
\end{aligned}
\end{align*}
Since $\hat\sigma\overset{p}{\to}\sigma$, Slutsky's theorem gives
\[
\frac{(\Delta\hat{\bm{w}})^\top\bm{Y}}{\hat\sigma\|\Delta\hat{\bm{w}}\|}=\frac{n^{-1/2}(\Delta\hat{\bm{w}})^\top\bm{Y}}{\hat\sigma\{n^{-1}\|\Delta\hat{\bm{w}}\|^2\}^{1/2}}\dto N(0,1).
\]
Therefore,
\[
S_{lL}=\left(\frac{(\hat{\bm{w}}_{\bm{\Lambda}_l}^{\mathrm{bal}}-\hat{\bm{w}}_{\bm{\Lambda}_L}^{\mathrm{bal}})^\top\bm{Y}}{\hat\sigma\|\hat{\bm{w}}_{\bm{\Lambda}_l}^{\mathrm{bal}}-\hat{\bm{w}}_{\bm{\Lambda}_L}^{\mathrm{bal}}\|}\right)^2\dto \chi_1^2.
\]
This proves part (a).

Next suppose that $H_l$ does not hold and that $\mu_f(\bm{\Lambda}_l,\bm{\Lambda}_L)\neq \mu_f(\bm{\Lambda}_L,\bm{\Lambda}_L)$. By Theorem~\ref{thm:clt_fixed},
\[
T(\hat{\bm{w}}_{\bm{\Lambda}_l}^{\sthree})\Pto\mu_f(\bm{\Lambda}_l,\bm{\Lambda}_L), \qquad T(\hat{\bm{w}}_{\bm{\Lambda}_L}^{\sthree})\Pto\mu_f(\bm{\Lambda}_L,\bm{\Lambda}_L).
\]
Hence,
\[
T(\hat{\bm{w}}_{\bm{\Lambda}_l}^{\mathrm{bal}})-T(\hat{\bm{w}}_{\bm{\Lambda}_L}^{\mathrm{bal}})\Pto\mu_f(\bm{\Lambda}_l,\bm{\Lambda}_L)-\mu_f(\bm{\Lambda}_L,\bm{\Lambda}_L)\neq 0.
\]
Equivalently,
\[
\frac1n(\Delta\hat{\bm{w}})^\top\bm{Y}\Pto\mu_f(\bm{\Lambda}_l,\bm{\Lambda}_L)-\mu_f(\bm{\Lambda}_L,\bm{\Lambda}_L).
\]
Since $n^{-1}\|\Delta\hat{\bm{w}}\|^2=O_p(1)$ and $\hat\sigma\Pto\sigma$, we have
\[
S_{lL}=\frac{n\left\{n^{-1}(\Delta\hat{\bm{w}})^\top\bm{Y}\right\}^2}{\hat\sigma^2\left\{n^{-1}\|\Delta\hat{\bm{w}}\|^2\right\}}\Pto\infty.
\]
This proves part (b).

We now prove part (c). Let $l^*=\arg\min\{l:\bm{\Lambda}_l\text{ is a true low-rank structure}\}$. For every $l<l^*$, the structure $\bm{\Lambda}_l$ is misspecified. By part (b),
\[
P(S_{lL}<\chi^2_{1;\alpha})\to 0.
\]
Since $L$ is fixed,
\[
P\left(\exists l<l^*:S_{lL}<\chi^2_{1;\alpha}\right)\to 0.
\]
Therefore, with probability tending to one, no misspecified structure preceding $l^*$ is selected.

On the other hand, $H_{l^*}$ holds. By part (a), $S_{l^*L}\overset{d}{\to}\chi_1^2$, and hence, by continuity of the $(1-\alpha)$-quantile,
\[
P(S_{l^*L}<\chi^2_{1;\alpha})\to 1-\alpha.
\]
Consequently,
\[
\liminf_{n\to\infty}P(\hat l=l^*)\ge 1-\alpha.
\]
Moreover, because all structures $\bm{\Lambda}_l$ with $l\ge l^*$ are true low-rank structures and all $l<l^*$ are rejected with probability tending to one,
\[
P(\bm{\Lambda}_{\hat l}\text{ is a true low-rank structure})\to 1.
\]
This proves part (c).
\end{proof}

\subsection{Proof of Theorem~\ref{thm:simplified_indiv_weights}}
\begin{proof}
    
Under the assumption that each unit $i$ in any $c^{\mathrm{th}}$ cluster contributes a unique $h_{ij}$, $\bm \Lambda$ has a block matrix structure. In order to see this, note that in this case we have the low-rank structural assumption:
\[
    g_{ci}^{(\bm a_c)} = \bm \Lambda_{ci}(\bm a_c)\bm h_{ci},
\]
for all $\bm a_c\in \{0,1\}^{m_c}$ and all $i,c$. In Examples~\ref{ex:1} to \ref{ex:3}, $\bm \Lambda_{ci}(\bm a_c)$ is a row vector of indicators. One can then define,
\[
    \bm \Lambda(\bm a) = \begin{bmatrix}
        \bm \Lambda_{11}(\bm a_1) & \bm 0 & \cdots & \bm 0\\
        \vdots & \vdots & \ddots & \vdots\\
        \bm 0 & \bm 0 & \cdots & \bm \Lambda_{n,m_n}(\bm a_n)
    \end{bmatrix},
\]
and finally, one can then define
\[
    \bm \Lambda = \begin{bmatrix}
        \bm \Lambda(\bm a^{(0)})\\
        \vdots\\
        \bm \Lambda (\bm a^{(2^M-1)})
    \end{bmatrix}.
\]
Now assume that for numbers $\{v_{ci}(\bm a_c^{(j)})\}_{c,i,j}$, we define, $\bm v(\bm a_c^{(j)}):=(v_{11}(\bm a_c^{(j)}),\dots, v_{n,m_{n}}(\bm a_c^{(j)}))^\top$, and
\[
    \bm v = \left(\bm v(\bm a_c^{(0)})^\top,\dots, \bm v(\bm a_c^{(2^M-1)})^\top\right)^\top.
\]
Then note that we have,
\begin{align*}
    \bm \Lambda^\top\bm D\bm v &=  \begin{bmatrix}
        \bm \Lambda(\bm a_c^{(0)})\\
        \vdots\\
        \bm \Lambda (a_c^{(2^M-1)})
    \end{bmatrix}^\top\begin{bmatrix}
        e(\bm a^{(0)})\bm v(\bm a_c^{(0)})\\
        \vdots\\
        e(\bm a^{(2^M-1)})\bm v(\bm a_c^{(2^M-1)})
    \end{bmatrix}\\
    &= \sum_{j=0}^{2^M-1}e(\bm a_c^{(j)})\bm \Lambda(\bm a_c^{(j)})^\top\bm v(\bm a_c^{(j)})\\
    &= \sum_{j=0}^{2^M-1}\begin{bmatrix}
        e(\bm a^{(j)})\bm \Lambda_{11}(\bm a_1^{(j)})^\top v_{11}(\bm a_1^{(j)})\\
        \vdots\\
         e(\bm a^{(j)})\bm \Lambda_{nm_n}(\bm a_n^{(j)})^\top v_{nm_n}(\bm a_n^{(j)})
    \end{bmatrix}.
\end{align*}
The $(c,i)^{\mathrm{th}}$ entry of the above block matrix is given by
\begin{align*}
    &\sum_{\bm a\in \{0,1\}^M} e(\bm a)\bm \Lambda_{ci}(\bm a_c)^\top v_{ci}(\bm a_c)\\
    =& \sum_{\bm a_c\in \{0,1\}^{M_c}} e(\bm a_c;\bm X_c)\bm \Lambda_{ci}(\bm a_c)^\top v_{ci}(\bm a_c)\\
    =&\bm \Lambda_{ci}^\top\bm D_c\bm v_{ci},
\end{align*}
where, $\bm v_{ci}$ is defined as a vector with entries $v_{ci}(\bm a_c)$, for $\bm a_c\in \{0,1\}^{M_c}$. Hence, we have,
\[
    \bm \Lambda^\top\bm D \bm v = \begin{bmatrix}
        \bm v_{11}^\top\bm D_{nm_n}\bm \Lambda_{11},\cdots,  \bm v_{nm_n}^\top\bm D_{nm_n}\bm \Lambda_{nm_n}
    \end{bmatrix}^\top.
\]
Thus, based on the above observations, we write
\begin{align*}
    \Cipwlamb(e) = \left\{\bm R\tilde{\bm w}:\bm \Lambda_{ci}^\top\bm D_c(\tilde{\bm w}_{ci} - \tilde{\bm w}_{\mathrm{IPW},c,i})=\bm 0,\forall c,i\right\}.
\end{align*}
Because $\wipwlambtilde(e)$ is obtained by minimizing $\tilde{\bm w}^\top\bm D\tilde{\bm w} = \sum_{c,i}\tilde{\bm w}_{ci}^\top\bm D_c\tilde{\bm w}_{ci}$, which is separable in the terms indexed by $(c,i)$. Thus, $\wipwlambtilde(e)_{c,i}$ can actually be obtained by minimizing $\tilde{\bm w}_{ci}^\top\bm D_c\tilde{\bm w}_{ci}$ subject to $\bm \Lambda_{ci}^\top\bm D_{c}(\tilde{\bm w}_{ci} - \tilde{\bm w}_{\mathrm{IPW},c,i})=\bm 0$, and arguing similarly as Theorem \ref{thm:weights_expression}, the optima is given by $\bm D_{c}^{-1/2}\bm P_{\bm D_{c}^{1/2}\bm \Lambda_{ci}}\bm D_{c}^{1/2}\tilde{\bm w}_{\mathrm{IPW},c,i}$, which completes the proof. Note that the invertibility of $\bm D_c^{1/2}$ follows from Assumption~\ref{assm:pos}.
\end{proof}

\subsection{Proof of Lemma~\ref{lem:epxposure_map_weights}}
\label{sec:proof_lem_exposure_map_weights}
\begin{proof}
Let us start with understanding the explicit forms of $\bm \Lambda_{ci}$ and $\bm D_c$ under the discrete exposure mapping assumption. To begin with, let us enumerate $\mathcal S_{ci} = \left\{\phi_{ci\cdot 1}, \cdots, \phi_{ci\cdot r}\right\}$. Then note that, we have
\[
    \bm \Lambda_{ci}(\bm a_c) = \left[\mathbb I\left(\phi_{ci}(\bm a_c) = \phi_{ci\cdot 1}\right),\cdots,\mathbb I\left(\phi_{ci}(\bm a_c) = \phi_{ci\cdot r}\right) \right],
\]
and 
\[
    \bm D_c = \mathrm{diag}\left(e\left(\bm a_c^{(0)},\bm X_c\right),\cdots, e\left(\bm a_c^{(2^{M_c}-1)},\bm X_c\right)\right),
\]
where $\mathrm{diag}(\bm x)$ denotes a diagonal matrix with the entries of $\bm x$ arranged across the diagonal. Before proceeding, we first start with obtaining an expression for $(\bm \Lambda_{ci}^\top\bm D_c\bm \Lambda_{ci})$ that will help us in simplifying $\bm P_{\bm D_c^{1/2}\bm \Lambda_{ci}}$. For that, note
\begin{align*}
    \bm D_c\bm \Lambda_{ci} &= \mathrm{diag}\left(e\left(\bm a_c^{(0)},\bm X_c\right),\cdots, e\left(\bm a_c^{(2^{M_c}-1)},\bm X_c\right)\right)\begin{bmatrix}
        \bm \Lambda_{ci}\left(\bm a_c^{(0)}\right)\\
        \vdots\\
        \bm \Lambda_{ci}\left(\bm a_c^{(2^{M_c}-1)}\right)
    \end{bmatrix}\\
    &=\begin{bmatrix}
       e\left(\bm a_c^{(0)},\bm X_c\right) \bm \Lambda_{ci}\left(\bm a_c^{(0)}\right)\\
        \vdots\\
         e\left(\bm a_c^{(2^{M_c}-1)},\bm X_c\right)\bm \Lambda_{ci}\left(\bm a_c^{(2^{M_c}-1)}\right)
    \end{bmatrix}.
\end{align*}
This shows that
\begin{align*}
    \bm \Lambda_{ci}^\top\bm D_c\bm \Lambda_{ci} &= \sum_{0\leq j\leq 2^{M_c - 1}}e\left(\bm a_c^{(j)}, \bm X_c\right)\bm \Lambda_{ci}\left(\bm a_c^{(j)}\right)^\top\bm \Lambda_{ci}\left(\bm a_c^{(j)}\right).
\end{align*}
Next, note that for any $0\leq j\leq 2^{M_c}-1$,
\begin{align*}
    &\bm \Lambda_{ci}\left(\bm a_c^{(j)}\right)^\top\bm \Lambda_{ci}\left(\bm a_c^{(j)}\right)\\
    =&\begin{bmatrix}\mathbb I\left(\phi_{ci}(\bm a^{(j)}_c) = \phi_{ci\cdot 1}\right)\\
    \vdots\\
    \mathbb I\left(\phi_{ci}(\bm a^{(j)}_c) = \phi_{ci\cdot r}\right) \end{bmatrix}
    \begin{bmatrix}\mathbb I\left(\phi_{ci}(\bm a^{(j)}_c) = \phi_{ci\cdot 1}\right)&\dots&\mathbb I\left(\phi_{ci}(\bm a^{(j)}_c) = \phi_{ci\cdot r}\right) \end{bmatrix}\\
    =& \mathrm{diag}\left(\mathbb I\left(\phi_{ci}(\bm a^{(j)}_c) = \phi_{ci\cdot 1}\right),\cdots,\mathbb I\left(\phi_{ci}(\bm a^{(j)}_c) = \phi_{ci\cdot r}\right) \right).
\end{align*}
This gives
\begin{align*}
     \bm \Lambda_{ci}^\top\bm D_c\bm \Lambda_{ci} &= \sum_{0\leq j\leq 2^{M_c - 1}}e\left(\bm a_c^{(j)}, \bm X_c\right)\bm \Lambda_{ci}\left(\bm a_c^{(j)}\right)^\top\bm \Lambda_{ci}\left(\bm a_c^{(j)}\right)\\
     &= \sum_{0\leq j\leq 2^{M_c - 1}}e\left(\bm a_c^{(j)}, \bm X_c\right)\mathrm{diag}\left(\mathbb I\left(\phi_{ci}(\bm a_c) = \phi_{ci\cdot 1}\right),\cdots,\mathbb I\left(\phi_{ci}(\bm a_c) = \phi_{ci\cdot r}\right) \right)\\
     &= \mathrm{diag}\left(e(\phi_{ci\cdot 1}),\cdots, e(\phi_{ci\cdot r})\right),
\end{align*}
where we define
\[
    e(\phi_{ci\cdot l}) = \sum_{\bm a_c: \phi_{ci}(\bm a_c) = \phi_{ci\cdot l}} e(\bm a_c, \bm X_c),
\]
which is positive under Assumption~\ref{assm:pos}. Thus, under Assumption~\ref{assm:pos}, $\bm \Lambda_{ci}^\top\bm D_c\bm \Lambda_{ci}$ is invertible. Hence, we now simplify the expression of our weights as
\begin{align*}
    (\wipwlamb(e))_{ci} &= \bm R_{c}\bm D_{c}^{-1/2}\bm P_{\bm D_{c}^{1/2}\bm \Lambda_{ci}}\bm D_{c}^{1/2}\tilde{\bm w}_{\mathrm{IPW},c,i}\\
    &= \bm R_{c}\bm D_{c}^{-1/2}\bm D_c^{1/2}\bm \Lambda_{ci}\left(\bm \Lambda_{ci}^\top\bm D_c\bm \Lambda_{ci}\right)^{-1}\bm \Lambda_{ci}^\top\bm D_c^{1/2}\bm D_c^{1/2}\tilde{\bm w}_{\mathrm{IPW},c,i}\\
    &= \bm R_c \bm \Lambda_{ci}\left(\bm \Lambda_{ci}^\top\bm D_c\bm \Lambda_{ci}\right)^{-1}\bm \Lambda_{ci}^\top\bm D_c\tilde{\bm w}_{\mathrm{IPW},c,i}\\
    &=  \bm R_c \bm \Lambda_{ci}\left(\bm \Lambda_{ci}^\top\bm D_c\bm \Lambda_{ci}\right)^{-1}\bm \Lambda_{ci}^\top\bm f_c\cdot M_c^{-1},
\end{align*}
where, $\bm f_c = \left(f(\bm a_c^{(0)},\bm X_c), \cdots,f(\bm a_c^{(2^{M_c{-1}})},\bm X_c) \right)^\top$. Next,
\begin{align*}
    \bm \Lambda_{ci}^\top\bm f_c &= \sum_{0\leq j\leq 2^{M_c}-1}\bm \Lambda_{ci}\left(\bm a_c^{(j)}\right)^\top f\left(\bm a_c^{(j)},\bm X_c\right)\\
    &= \sum_{0\leq j \leq 2^{M_c}-1}\begin{bmatrix}
        \mathbb I\left(\phi_{ci}\left(\bm a_c^{(j)}\right) = \phi_{ci\cdot 1}\right)f\left(\bm a_c^{(j)},\bm X_c\right)\\
        \vdots\\
        \mathbb I\left(\phi_{ci}\left(\bm a_c^{(j)}\right) = \phi_{ci\cdot r}\right)f\left(\bm a_c^{(j)},\bm X_c\right)
    \end{bmatrix}\\
    &=\begin{bmatrix}
        f(\phi_{ci\cdot 1})\\
        \vdots\\
        f(\phi_{ci\cdot r})
    \end{bmatrix},
\end{align*}
where we again define
\[
    f(\phi_{ci\cdot l}) = \sum_{\bm a_c: \phi_{ci}(\bm a_c) = \phi_{ci\cdot l}} f(\bm a_c, \bm X_c).
\]
Hence, using the expression for $\bm \Lambda_{ci}^\top\bm D_c\bm \Lambda_{ci}$ we have,
\begin{align*}
     (\wipwlamb(e))_{ci} &=  \bm R_c \bm \Lambda_{ci}\left(\bm \Lambda_{ci}^\top\bm D_c\bm \Lambda_{ci}\right)^{-1}\bm \Lambda_{ci}^\top\bm f_c\cdot M_c^{-1}\\
     &= \bm \Lambda_{ci}(\bm A_c) \begin{bmatrix}
        \frac{f(\phi_{ci\cdot 1})}{M_c\cdot e(\phi_{ci\cdot 1})}\\
        \vdots\\
         \frac{f(\phi_{ci\cdot r})}{M_c\cdot e(\phi_{ci\cdot r})}
    \end{bmatrix}\\
    &= \left[\mathbb I\left(\phi_{ci}(\bm A_c) = \phi_{ci\cdot 1}\right),\cdots, \mathbb I\left(\phi_{ci}(\bm A_c) = \phi_{ci\cdot r}\right)\right]\begin{bmatrix}
        \frac{f(\phi_{ci\cdot 1})}{M_c\cdot e(\phi_{ci\cdot 1})}\\
        \vdots\\
         \frac{f(\phi_{ci\cdot r})}{M_c\cdot e(\phi_{ci\cdot r})}
    \end{bmatrix}\\
    &= \sum_{1\leq l\leq r} \frac{f(\phi_{ci\cdot l})}{M_c\cdot e(\phi_{ci\cdot l})}\cdot \mathbb I\left(\phi_{ci}(\bm A_c) = \phi_{ci\cdot l}\right)\\
    &=  \frac{f_{\phi_{ci}}(\bm A_c,\bm X_c)}{M_c\cdot e_{\phi_{ci}}(\bm A_c,\bm X_c)}.
\end{align*}
\end{proof}

\subsection{Proof of Theorem~\ref{thm:bal_feasible}}
\label{sec:proof_bal_feasible}
\begin{proof}
Under Assumption~\ref{assm:sufficient_feasibility}, we define the function $\mathrm{type}_{ci}(\bm a_c)$ that maps a treatment pattern $\bm a_c$ to a $j\in [1:r]$ such that $\bm \Lambda^{(\beta)}_{ci}(\bm a_c) = \bm \Lambda^{(j)}$. Then note that
\begin{align*}
    \mathbb P(\mathrm{type}_{c1}(\bm A_c) = j) = \mathbb E\left[\sum_{\bm a_c: \mathrm{type}_{c1}(\bm a_c)=j}e(\bm a_c, \bm X_c)\right]
\end{align*}
is strictly positive due to Assumption~\ref{assm:pos}. Next, note that for any Borel set $S$,
\begin{align*}
    \mathbb P(\bm X_{c1}\in S) = \sum_{1\leq j\leq r} \mathbb P(\bm X_{c1}\in S\mid \mathrm{type}_{c1}(\bm A_c) = j)\mathbb P(\mathrm{type}_{c1}(\bm A_c) = j).
\end{align*}
Because $\mathbb P(\mathrm{type}_{c1}(\bm A_c) = j)>0,\forall j$ we conclude from the above that $\mathbb P(\bm X_{c1}\in S) = 0\implies \mathbb P(\bm X_{c1}\in S\mid \mathrm{type}_{c1}(\bm A_c) = j) = 0,\forall j$. In particular, because $\bm X_{c1}$ has an absolutely continuous distribution, if $S$ has Lebesgue measure 0, then we must have $\mathbb P(\bm X_{c1}\in S\mid \mathrm{type}_{c1}(\bm A_c) = j) = 0,\forall j$. This shows that the conditional distribution of $\bm X_{c1}$ conditioned on $\mathrm{type}_{c1}(\bm A_c) = j$, for any $j$, is also absolutely continuous.

Having made this observation, we next make the following simplifications under Assumption~\ref{assm:sufficient_feasibility}:
\begin{align*}
    \bm \Lambda^\top\bm f = \sum_{c=1}^n\sum_{i=1}^{M_c}\sum_{\bm a_c \in \{0,1\}^{M_c}}\bm X_{ci}^\top\bm \Lambda_{ci}^{(\beta)}(\bm a_c)f(\bm a_c, \bm X_c) = \sum_{j=1}^r\bm v_j^\top\bm \Lambda^{(j)},
\end{align*}
where, 
\[
    \bm v_j = \sum_{\substack{c,i,\bm a_c: \\ \mathrm{type}_{ci}(\bm a_c) = j}}\bm X_{ci}f(\bm a_c, \bm X_c).
\]
Next note that similarly,
\[
    \bm \Lambda^\top\bm R^\top\bm w = \sum_{c=1}^n\sum_{i=1}^{M_c}w_{ci}\bm X_{ci}^\top\bm \Lambda^{(\beta)}(\bm A_c) = \sum_{j=1}^r \left(\sum_{c,i:\mathrm{type}_{ci}(\bm A_c) = j}w_{ci}\bm X_{ci}^\top\right)\bm \Lambda^{(j)}.
\]
Now, let us define the event $B_{j,n} = \left[\mathrm{span}\left(\left\{\bm X_{c1}: \mathrm{type}_{c1}(\bm A_c) = j, 1\leq c\leq n\right\}\right)= \mathbb R^p\right]$. Then note that if $B_j$ holds for all $j$, then we find weights $\{w_{ci}\}$ with $w_{ci}=0,i>1$ such that 
\[
    \bm v_j = \sum_{c,i:\mathrm{type}_{ci}(\bm A_c) = j}w_{ci}\bm X_{ci}, \forall j.
\]
This implies that for this set of weights, the balancing equation is satisfied and hence, $B_{\Lambda}$ holds. Thus,
\begin{align}
\label{eqn:bj_cap_relation}
    \bigcap_{1\leq j\leq r}B_{j,n} \subseteq B_{\Lambda}.
\end{align}
Next, fix a $j\in [1:r]$. Define the event
\[
    E_j = \{\textrm{There exists a sub-sequence }\{c_l\}_{l\geq 1}\textrm{ such that }\mathrm{type}_{c_l1}(\bm A_{c_l})=j\}.
\]
Then note that because $\mathbb P(\mathrm{type}_{c1}(\bm A_{c})=j)>0$ and does not depend on $c$ and $\mathrm{type}_{c1}(\bm A_{c})$'s are i.i.d. across $c$ (under Assumption~\ref{assm:super_pop}),  we have $\mathbb P(E_j) = 1$.
Thus, for $E:= \cap_{1\leq j\leq R}E_j$, we have that $\mathbb P(E) = 1$.

Fix a $j$. Assuming that $E$ (and hence, $E_j$) holds, consider the sequence $\{\bm X_{c_l1}\}_{l\geq 1}$ and note that this is an i.i.d. sequence in $l$ with the same common distribution $\bm X_{c1}\mid \mathrm{type}_{c_l,1} = j$. From the argument in the first paragraph in this section, this distribution is absolutely continuous and hence the square matrix $\begin{bmatrix}
    \bm X_{c_11}&\bm X_{c_21}&\cdots & \bm X_{c_p1}
\end{bmatrix}$ must be invertible almost surely. Thus, $\mathrm{span}\left\{\bm X_{c1}:\mathrm{type}_{c1}(\bm A_c) = j, c\geq 1\right\}\supseteq\mathrm{span}\left(\left\{\bm X_{c_l,1}:1\leq l\leq p\right\}\right) = \mathbb R^p$ almost surely. Define the event, $B_j$ as $\mathrm{span}\left(\left\{\bm X_{c1}: \mathrm{type}_{c1}(\bm A_c) = j\right\}, c\geq 1\right)= \mathbb R^p$, and note that $B_{j,n}\uparrow B_j$ as $n\rightarrow \infty$. Furthermore, the discussion above suggests $\mathbb P(B_j\mid E)=1$. Hence, we also have $\lim_{n\rightarrow \infty}\mathbb P(B_{j,n}\mid E) = 1,\forall j$. This implies that 
\begin{align*}
    &\lim_{n\rightarrow \infty}\mathbb P\left(\bigcap_{1\leq j\leq r}B_{j,n}\middle|E\right) = 1\\
    \implies & \underset{n\rightarrow \infty}{\lim\inf}~\mathbb P(B_{\Lambda}\mid E) \geq 1\textrm{ [Equation~\eqref{eqn:bj_cap_relation}]}\\
    \implies &\lim_{n\to \infty}~\mathbb P(B_\Lambda\mid E) = 1\textrm{ [since probability is bounded by 1]}\\
    \implies &  \lim_{n\to \infty}~\mathbb P(B_\Lambda) = 1\textrm{  [since }\mathbb P(E) = 1\textrm{]},
\end{align*}
thereby showing that Assumption~\ref{assm:balancing_feasibility} is satisfied.

\end{proof}

\section{Additional Simulations}
\label{sec:further_simulations_fixed}

In this section, we extend the simulations presented in Section~\ref{sec:experiments_simulated} by providing additional details and also considering the performance under stratified interference and additive interference. For stratified interference, we again assume that interference is restricted within the first five nearest neighbors and set
\[
    \bm h = \gamma_{\mathrm{strat}}\times (1,\dots, 6)^\top\otimes \bm 1_4,
\]
while for additive interference, we follow the convention set in Example~\ref{ex:7} and set,
\[
    \bm h = \gamma_{\mathrm{add}}\times (0,1,0,2,\cdots, 0,15)\otimes \bm 1_4.
\]
We can vary the parameters $\gamma_{\mathrm{strat}}$ and $\gamma_{\mathrm{add}}$ under their respective settings and directly control the SNR parameter $\eta$. In addition to the balancing estimator, we also consider the performance of the projection estimator. Furthermore, for the setting with nearest-neighbors interference, we consider the performance of the balancing estimator with a data-adaptive choice of the low-rank structure (as described in Section~\ref{sec:experiments_simulated}).

\subsection{Performance of confidence intervals}
\label{sec:further_simulations_fixed_ci}

Figure~\ref{fig:fixed_ci_graph_coverage} extends the results of Figure~\ref{fig:fixed_ci_main} by also presenting the results of the projection estimator under the same settings. This is the only setting where we evaluate the data-adaptive version of the balancing estimator. Figure~\ref{fig:fixed_ci_stratified_coverage} summarizes the corresponding results under stratified interference, while Figure~\ref{fig:fixed_ci_additive_coverage} shows the results under additive interference. In all of these cases, the balancing equations are always feasible.

We discuss several new findings beyond those mentioned in Section~\ref{sec:experiments_simulated}. First, the top row of each figure shows that under all interference settings, the performance of the projection estimator is between the balancing and IPW estimators in terms of standard deviation.  This is consistent with our discussion following Corollary~\ref{cor:propensity_correct_specification_fixed}. Here, across all settings, we find that the IPW estimator based on the estimated propensity scores exhibits performance similar to the IPW estimator with the true propensity score.  Finally, similar to Figure~\ref{fig:fixed_ci_main}, all the proposed estimators achieve a much smaller variance than the IPW estimator under all interference settings.

The middle row in Figures~\ref{fig:fixed_ci_stratified_coverage},~\ref{fig:fixed_ci_graph_coverage},~and~\ref{fig:fixed_ci_additive_coverage} also show results similar to those presented in Figure~\ref{fig:fixed_ci_main}.  They show that asymptotically, the CIs of all the proposed estimators have the desired coverage.  Unlike the balancing estimator, the coverage of the projection estimator also starts to deteriorate as $\kappa$ increases in the rightmost panel.  This is similar to the IPW estimator. Interestingly, for the additive interference case shown in Figure~\ref{fig:fixed_ci_additive_coverage}, the CI of the IPW estimator maintains the desired coverage even for extreme values of $\kappa$. However, this is simply due to its high variance, and the asymptotic normality still fails for the IPW estimator (see Figure~\ref{fig:additive_anomaly_quantile}). These findings are consistent with our recommendation that the balancing estimator is most appropriate under the linear model assumption.

The bottom row of Figures~\ref{fig:fixed_ci_stratified_coverage},~\ref{fig:fixed_ci_graph_coverage},~and~\ref{fig:fixed_ci_additive_coverage}
summarizes the length of the various confidence intervals. The results are similar to those shown in Figure~\ref{fig:fixed_ci_main}. As expected, the CI length of the projection estimator is placed between the CI lengths of the balancing and IPW estimators.

Another noteworthy observation is made based on the center panels in the bottom row of Figures~\ref{fig:fixed_ci_stratified_coverage},~\ref{fig:fixed_ci_graph_coverage},~and,~\ref{fig:fixed_ci_additive_coverage}. We find that as the SNR increases, the CI length of the IPW estimator approaches those of the other CIs. This finding is consistent with Theorem~\ref{thm:propensity_estimators_fixed_h} and Corollary~\ref{cor:correct_specification_fixed}, which show that the benefit of exploiting the low-rank structure is substantial only when the signal-to-noise ratio is not too high. 

\begin{figure}[t]
\centering
    \includegraphics{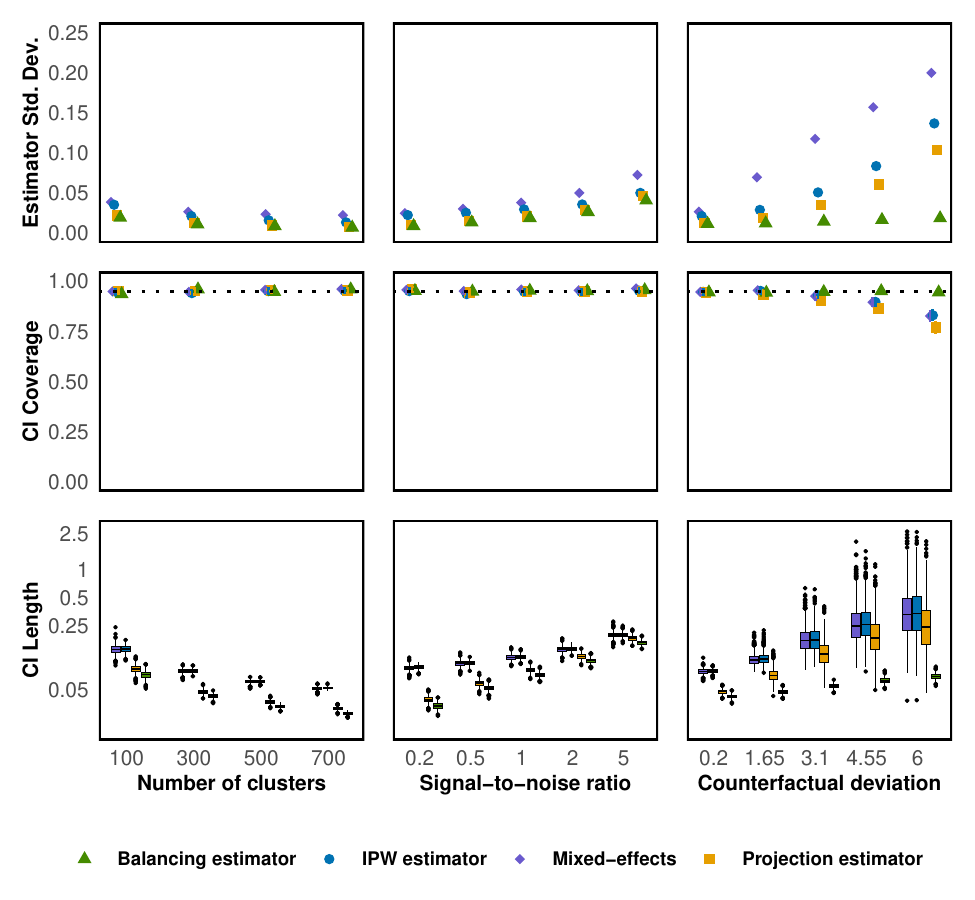}
    \caption{Performance of the various confidence intervals under stratified interference. Exactly the same settings and conventions as Figure~\ref{fig:fixed_ci_main}.}
    \label{fig:fixed_ci_stratified_coverage}
\end{figure}


\begin{figure}
    \centering
    \includegraphics{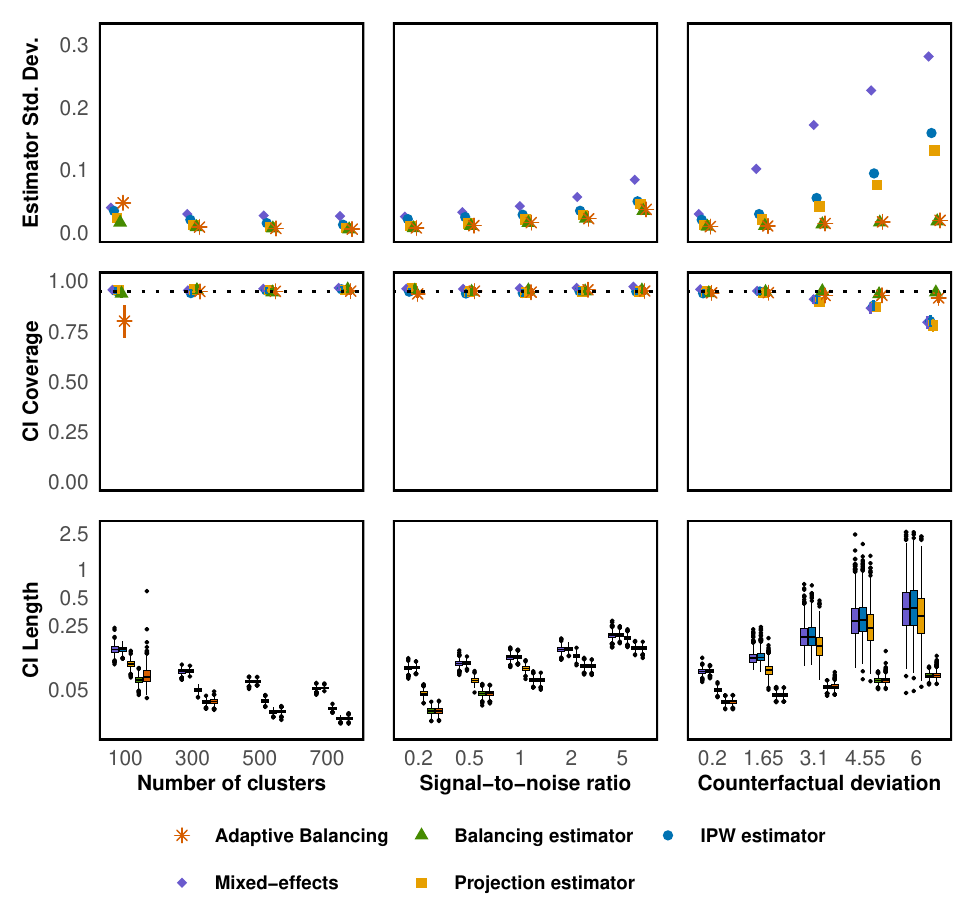}
    \caption{Performance of the various confidence intervals under nearest neighbors interference. Exactly the same settings and conventions as Figure~\ref{fig:fixed_ci_main}.}
    \label{fig:fixed_ci_graph_coverage}
\end{figure}


\begin{figure}
    \centering
    \includegraphics{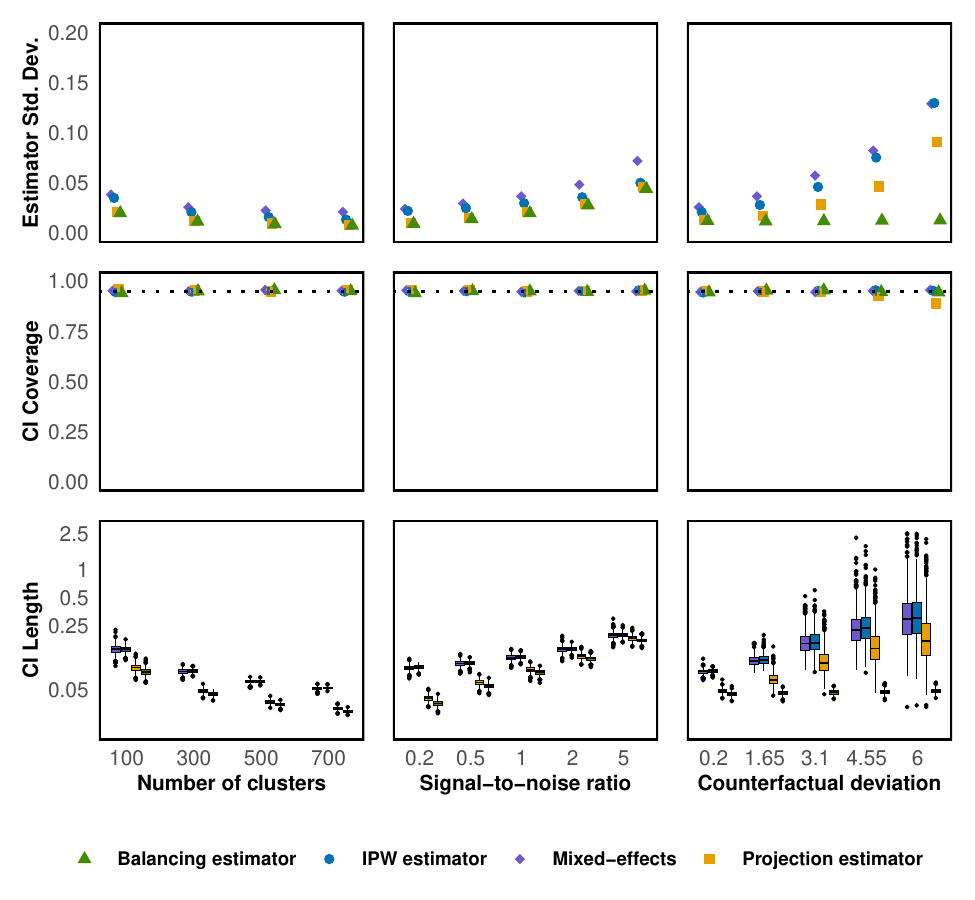}
    \caption{Performance of the various confidence intervals under additive interference. Exactly the same settings and conventions as Figure~\ref{fig:fixed_ci_main}.}
    \label{fig:fixed_ci_additive_coverage}
\end{figure}


\subsection{Bias of estimators}
\label{sec:further_simulations_fixed_bias}

Under the same settings as those of Figure~\ref{fig:fixed_ci_main}, Figure~\ref{fig:fixed_bias_main} compares the bias of the balancing, the projection, the IPW estimator with true and estimated propensity scores, over 1000 independent replicates. For the balancing estimator, we only use the replicates where the balancing equation is feasible.
Figures~\ref{fig:fixed_bias_stratified}~and~\ref{fig:fixed_bias_additive} summarize the results under analogous settings with stratified and additive interference, respectively. All these figures provide empirical evidence that is consistent with the following: (1) the projection estimator and the IPW estimator with true propensity score are always unbiased, with the counterpart with estimated propensity also exhibiting approximate unbiasedness; (2) while the balancing estimator is technically biased in finite samples (Section~\ref{sec:low_rank_weights}), this bias is small. 

\begin{figure}
\centering
\includegraphics{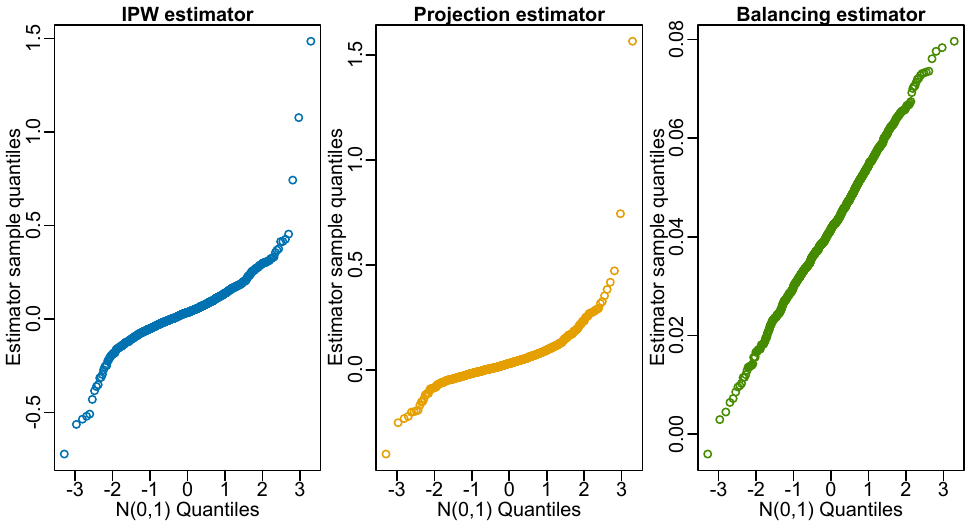}
    \caption{QQ-plot of the IPW, projection, and balancing estimators (against the standard normal quantiles) under $\kappa = 6$ from the rightmost panel of Figure~\ref{fig:fixed_ci_additive_coverage}. Even though the estimators are not scaled or centered, a non-straight line denotes departure from normality. Such a departure occurs for the IPW and projection estimators, but not for the balancing estimator. Thus, in the $\kappa = 6$ setting from the rightmost panel of Figure~\ref{fig:fixed_ci_additive_coverage}, the IPW estimator does not achieve normality for the particular sample size, and its empirical coverage simply reflects its high variance.}
    \label{fig:additive_anomaly_quantile}
\end{figure}

\begin{figure}
\centering
\includegraphics{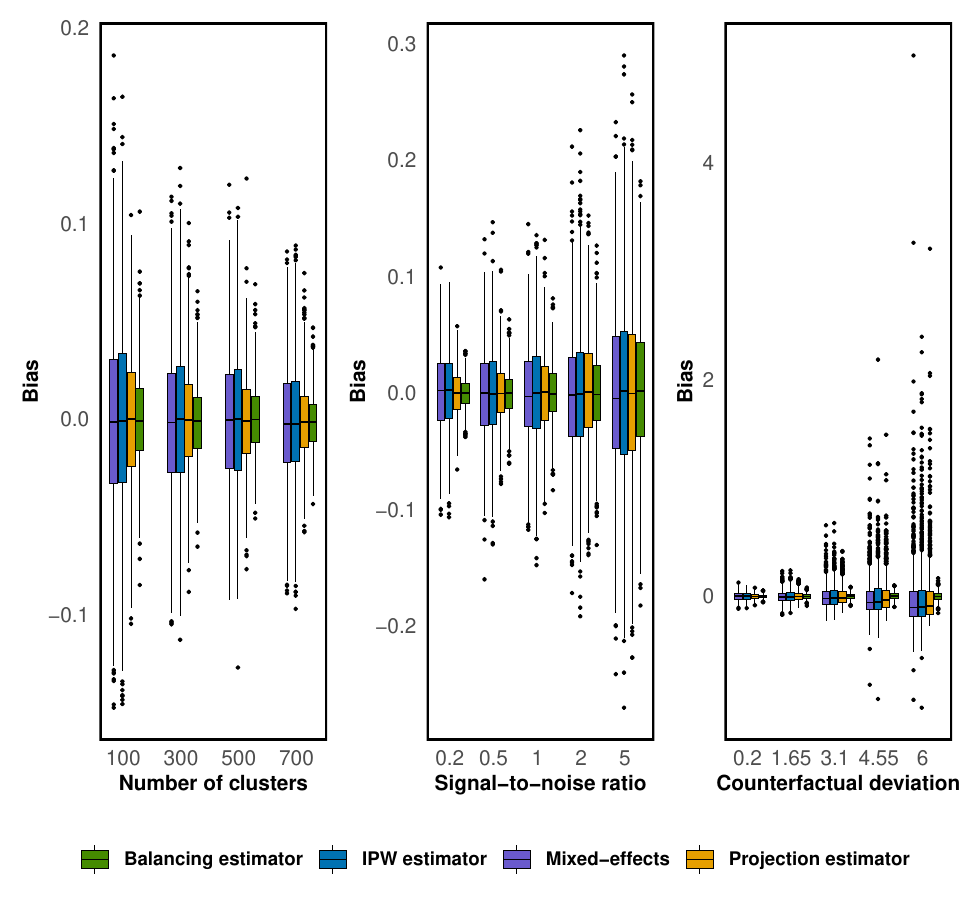}
    \caption{Bias of the various estimators under nearest neighbor interference. For the left, we fix $\kappa = 0.2,\eta = 0.2$ and vary $n$. For the center, we fix $n = 80,\eta = 0.2$ and vary $\kappa$. For the right, we fix $n = 80, \kappa = 0.2$ and vary $\eta$, with the x-axis plotting $\log(\eta)$. For all the simulations, we set $\sigma = 1$.}
    \label{fig:fixed_bias_main}
\end{figure}


\begin{figure}
    \centering
    \includegraphics{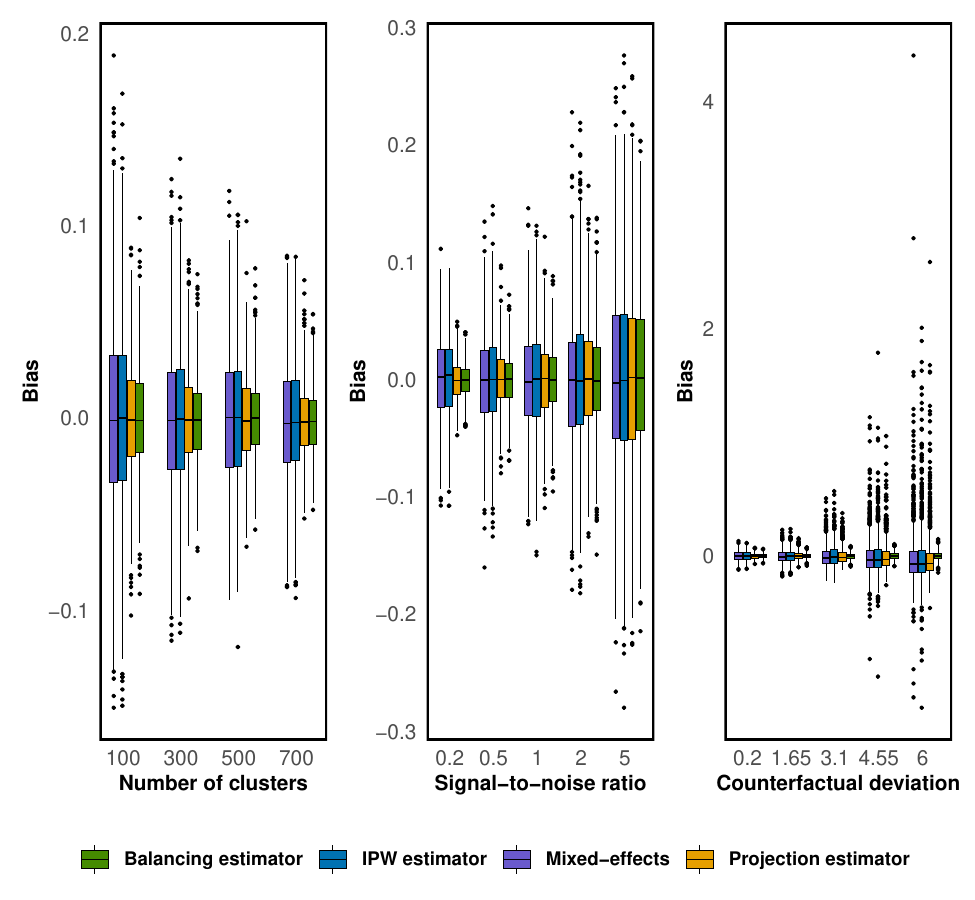}
    \caption{Bias of the various estimators under stratified interference. Exactly the same setting as Figure~\ref{fig:fixed_bias_main}.}
    \label{fig:fixed_bias_stratified}
\end{figure}

\begin{figure}
    \centering
    \includegraphics{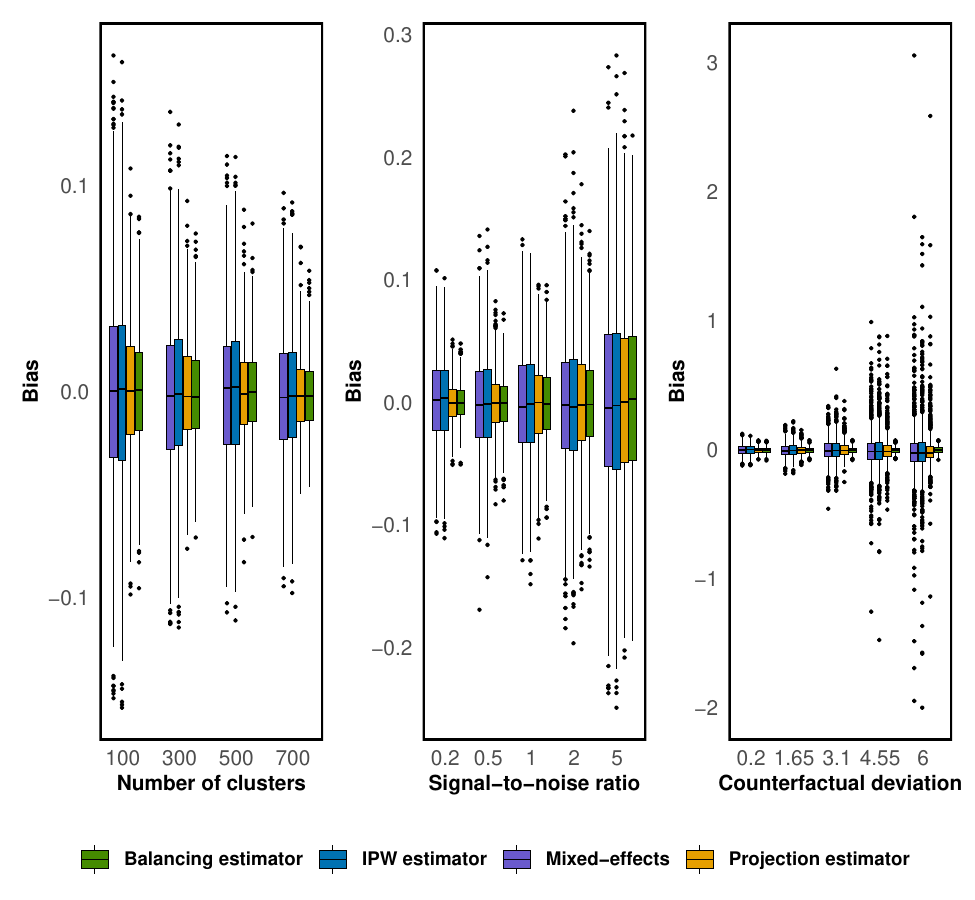}
    \caption{Bias of the various estimators under additive interference. Exactly the same setting as Figure~\ref{fig:fixed_bias_main}.}
    \label{fig:fixed_bias_additive}
\end{figure}

\section{Explicit mathematical expressions of low-rank structures}
\label{sec:imbal_real_data}

In this appendix, we give explicit mathematical expressions for the low-rank structures discussed in Section~\ref{sec:real_data_analysis_setup}.

\begin{itemize}
    \item {\bf Anonymous interference within neighborhood:} In the dataset, the number of treated neighbors for each of the units can range from 0 to 15 (excluding 14). Let $t_{ci\cdot 1}$ denote this number for the $i^{\mathrm{th}}$ unit in the $c^{\mathrm{th}}$ cluster. Then in this case, let $\bm e_{ci\cdot 1}$ denote a $15$-dimensional indicator vector having 1 at position $t_{ci\cdot 1}$ and 0 everywhere else. Then the low-rank structure takes the form: $\left(\bm \Lambda_{\text{neighbor1}}(\bm a_c)\right)_{ci} = (\mathbb I(a_{ci} = 0), \mathbb I(a_{ci}=1), \bm e_{ci\cdot 1})$.
    
    As discussed in Section~\ref{sec:real_data_analysis_setup}, we observe that $t_{ci\cdot}$ has a large number of unique values, leading to an effective treatment space of high dimension that might affect with our ability to achieve balance. We therefore consider another low-rank structure and reduce the number of effective treatments by discretizing this variable.  Specifically, we use three categories (low, medium, and high) using the 33 and 67 percentiles as thresholds; let us call this coarsened treatment status $r_{ci\cdot 2}$, and posit that the low-rank structure depends on $t_{ci\cdot 1}$ only through $r_{ci\cdot 1}$.  Under this assumption, the new low-rank structure takes the form:\\
    $\left(\bm \Lambda^*_{\text{neighbor1}}(\bm a_c)\right)_{ci} = \left(\mathbb I(a_{ci}=0),\mathbb I(a_{ci}=1),\mathbb I(r_{ci.1}=\text{low}  ),\mathbb I(r_{ci.1}=\text{medium}),\mathbb I(r_{ci.1}=\text{high} )\right)$.
    
    \item \textbf{Anonymous interference within neighborhoods of neighbors}: Recall that the neighbors of neighbors of a unit are defined as the units that are exactly at distance 2 in the graph. Let $r_{ci\cdot 2}$ denote the number of treated neighbors of neighbors (based on the dataset, this can take one of 40 different possible values). Let $\bm e_{ci\cdot 2}$ denote a 40-dimensional indicator vector for this. Then, the low-rank structure takes the form: $\left(\bm \Lambda_{\text{neighbor2}}(\bm a_c)\right)_{ci}  = (\mathbb I(a_{ci} = 0), \mathbb I(a_{ci}=1), \bm e_{ci\cdot 1},\bm e_{ci\cdot 2})$.
    
    Analogously, we again coarsen the number of treated in neighbors of neighbors to $r_{ci.2}$ taking values in $\{\textrm{low(l), medium(m), high(h)}\}$.  Then the modified low-rank structure given by:\\
    \begin{align*}
        &\left(\bm \Lambda^*_{\text{neighbor2}}(\bm a_c)\right)_{ci}\\
        =&\left(\mathbb I(a_{ci}=0),\mathbb I(a_{ci}=1),\mathbb I(r_{ci.1}=\mathrm l),\mathbb I(r_{ci.1}=\mathrm m),\mathbb I(r_{ci.1}=\mathrm h),\mathbb I(r_{ci.2}=\mathrm l),\mathbb I(r_{ci.2}=\mathrm m),\mathbb I(r_{ci.2}=\mathrm h)\right).
    \end{align*}
    
\end{itemize}

\end{document}